\documentclass[11pt,a4paper]{article}

\usepackage[margin=1in]{geometry}

\usepackage{amsmath}  
\usepackage{amsthm}  
\usepackage{amssymb}  
\usepackage{enumitem}  
\usepackage{filecontents}  
\usepackage{hyperref}
\usepackage{tikz}

\hypersetup{colorlinks=true,citecolor=blue, linkcolor=blue, urlcolor=blue}

\newtheorem{thm}{Theorem}
\newtheorem{prop}[thm]{Proposition}
\newtheorem{lem}[thm]{Lemma}
\newtheorem{cor}[thm]{Corollary}
\theoremstyle{definition}
\newtheorem{dfn}[thm]{Definition}
\newtheorem{ex}[thm]{Example}
\newtheorem{obs}[thm]{Observation}
\newtheorem*{rem}{Remark}

\let\dl=\delta
\let\Dl=\Delta
\let\ld=\lambda
\let\Ld=\Lambda
\let\om=\omega
\let\Gm=\Gamma
\let\epsilon=\varepsilon

\def\AA{\mathbb{A}}
\def\Anz{\AA\setminus\{0\}}
\def\CC{\mathbb{C}}
\def\NN{\mathbb{N}}

\def\ZZ{\mathbb{Z}}

\def\cD{\mathcal{D}} 
\def\cE{\mathcal{E}}
\def\cF{\mathcal{F}}
\def\cG{\mathcal{G}}
\def\cI{\mathcal{I}}
\def\cM{\mathcal{M}}

\def\cO{\mathcal{O}}
\def\cP{\mathcal{P}}
\def\cR{\mathcal{R}}
\def\cS{\mathcal{S}}
\def\cT{\mathcal{T}}
\def\cU{\mathcal{U}}

\def\ba{{\bf a}}
\def\bb{{\bf b}}
\def\bc{{\bf c}}

\def\be{{\bf e}}
\def\bg{{\bf g}}

\def\bs{{\bf s}}

\def\bv{{\bf v}}
\def\bw{{\bf w}}
\def\bx{{\bf x}}
\def\by{{\bf y}}
\def\bz{{\bf z}}

\def\NCSP{{\rm \#CSP}}
\def\NEQ{\mathrm{NEQ}}
\def\EQ{\mathrm{EQ}}
\def\ONE{\mathrm{ONE}}
\def\NAND{\mathrm{NAND}}
\def\hol{{\sf Holant}} 
\def\numP{{\sf \#P}} 
\def\HN{{\sf HolantNorm}}
\def\HA{{\sf HolantArg}}
\def\EVEN{\mathrm{EVEN}}

\let\sse=\subseteq
\def\ang#1{\langle #1 \rangle}
\def\vc#1#2{#1 _1\zd #1 _{#2}}
\def\zd{,\ldots,}
\def\smm#1{\left(\begin{smallmatrix} #1 \end{smallmatrix}\right)}
\def\pmm#1{\begin{pmatrix} #1 \end{pmatrix}}

\def\allf{\Upsilon}  
\def\ppsh{pps\textsubscript{h}}
\def\pph{pp\textsubscript{h}}
\def\GL{\operatorname{GL}_2(\AA)}

\def\ari{\operatorname{arity}}
\def\abs#1{\left| #1 \right|}
\def\hc#1{\ang{ #1 }_\mathrm{h}}  
\def\bhc#1#2#3{\ang{ #1 \sqcup #2 }_{\mathrm{h}, #3 }}  
\def\tcl#1{C\left( #1 \right)}  
\def\qc#1{\ang{ #1 }_\mathrm{q}}  
\def\arg{\operatorname{arg}}
\def\Arg{\operatorname{Arg}}
\def\norm#1{\left\Vert #1 \right\Vert}

 \def\calF{\mathcal{F}}
 \def\calG{\mathcal{G}}

\renewcommand{\t}[1]{\ensuremath{^{\otimes #1}}}
\newcommand{\ket}[1]{\left| #1 \right>}

\newcommand{\GHZ}{\mathrm{GHZ}}
\newcommand{\CNOT}{\operatorname{CNOT}}

\def\domain{D}

\setlist[description]{font=\normalfont\itshape}

\title{Holant clones and the approximability of conservative holant problems}
\author{Miriam Backens \and Leslie Ann Goldberg}
\date{6 January 2020}

\begin{document}

\maketitle

\begin{abstract}
We construct a theory of holant clones to capture 
the notion of expressibility in the holant framework. 
Their role is analogous to the role played by functional clones in
the study of weighted counting Constraint Satisfaction Problems.
We explore the landscape of conservative holant clones and 
determine the situations in which 
a set $\mathcal{F}$ of functions is ``universal in the conservative case'',
which means that 
all functions are contained in 
the holant clone generated by $\mathcal{F}$  together with all unary functions.
When $\mathcal{F}$ is not universal in the conservative case, we give concise generating
sets for  the clone.
We demonstrate the usefulness of the holant clone theory by using it to give
a complete complexity-theory classification for the problem of
approximating the solution to conservative holant problems.
We show that approximation is intractable exactly when $\mathcal{F}$ is universal in the conservative case.
\end{abstract}

\section{Introduction}
 
A \emph{computational counting problem}
is typically defined with respect to a  graph.
A \emph{configuration} is  an assignment of
``spins''   to the vertices or edges of this graph.
Each configuration has a \emph{weight}, which is obtained by assigning local weights to small (constant-sized) sub-configurations, and taking the product of these local weights.
The goal is to compute or to  approximate the \emph{partition function},
which is the sum of the weights of all configurations.
Here are three well-known examples of computational counting problems.
\begin{description}
\item[The Ising Model:]  
This statistical-physics model has a parameter~$\lambda$ which is a function of
the temperature and the strength of edge interactions.
Each assignment assigns a spin from $\{-1,+1\}$ to each vertex of the given graph $G=(V,E)$.
Each edge provides a local weight which is $\lambda$ if its endpoints have the same spin, and
$1$ otherwise. Thus, the partition function is
$\sum_{\sigma\colon V \to \{-1,+1\}} \lambda^{m(\sigma)}$, where $m(\sigma)$ is the number of monochromatic
edges in configuration~$\sigma$.
\item[The Ising Model with Local Fields:]
The configurations are the same as those of the Ising model, but, in addition to the local weights already
described, each vertex~$v$ provides a local weight which is some value $\mu_v$ if
the spin of~$v$ is~$+1$, and is $1$~otherwise.
Thus, the partition function is
$\sum_{\sigma\colon V \to \{-1,+1\}} \lambda^{m(\sigma)}\prod_{v\in V \colon \sigma(v)=+1} \mu_v$.
\item[The  Monomer-Dimer Model:]
The model has a parameter~$\gamma$.
Again, there are two spins -- we will call them $0$ and~$1$.
This time, a configuration assigns a spin to each \emph{edge} of the given  graph $G=(V,E)$.
Each vertex~$v$ provides a local weight, which is defined as follows.
\begin{itemize}
\item If no edges adjacent to~$v$ have spin~$1$, it provides weight~$1$.
\item If one edge adjacent to~$v$ has spin~$1$, it provides weight $\gamma^{1/2}$.
\item If more than one edge adjacent to~$v$ has spin~$1$, it provides weight $0$.
\end{itemize}
Suppose, for some subset $M$ of $E$, that a configuration assigns spin~$1$ to edges in~$M$ and
spin~$0$ to edges in $E\setminus M$. The weight of the configuration is 
$\gamma^{|M|}$ if $M$ is a matching of~$G$ and $0$~otherwise. The partition function is therefore 
$\sum_M   \gamma^{|M|}$, where the sum is over all  matchings of~$G$.
\end{description}
 
All three of these examples have their own literature.
However, in this paper we will be most interested in their roles as representatives  of broad classes of 
counting problems. The Ising Model (with or without local fields) 
is an example of a problem in a class of problems called weighted counting Constraint Satisfaction Problems (CSPs).
The Monomer-Dimer model is an example
of a related class of problems called holant problems.
Although this paper will be most concerned with holant problems, it will be useful to define counting CSPs
so that we can explain the research agenda in this area, the techniques, and ultimately, our contributions.

A   counting CSP is parameterised by a finite set~$\cF$ of functions over some domain~$\domain$. 
An input to the problem consists of
a set~$V$  of variables and a set~$C$ of constraints.
Each constraint~$c$ consists of a function $f_c\in\cF$ 
with some arity $\ari(f_c)$
and a   ``scope'',
which is just a tuple of $\ari(f_c)$ variables, so
 $\bv_c\in V^{\ari(f_c)}$. An assignment $\bx:V\to D$ induces a weight $\prod_{c\in C} f_c(\bx|_c)$, where $\bx|_c$ denotes the  restriction of the assignment $\bx$ to the scope of $c$. The relevant computation problem is defined as follows.

\begin{description}[noitemsep]
 \item[Name] $\NCSP(\cF)$
 \item[Instance] A tuple $(V,C)$, where $V$ is a finite set of variables and $C$ is a finite set of constraints over $\cF$.
 \item[Output] The value $\sum_{\bx:V\to\{0,1\}} \prod_{c\in C} f_c(\bx|_c)$.
\end{description}

Given the definition of $\NCSP(\cF)$, it is easy to see that
the Ising model corresponds to the case where $D=\{-1,+1\}$,
$f$ is the binary function over~$D$ which is~$\lambda$ if its arguments are the same and~$1$ otherwise,
and $\calF=\{f\}$.
In the case of the Ising model with local fields, 
$\calF$ has additional unary functions $u_\mu$ 
such that $u_\mu(+1)=\mu$ and $u_\mu(-1)=1$.

An important research direction is to try to determine for which sets $\calF$
the   problem $\NCSP(\calF)$ is tractable.
Building on and generalising
quite a bit of research \cite{DG00, BG05, Gplus10, Cplus13, Dplus07, CaiChen10, Bul13, DR13}, 
the exact-computation version of
this problem was finally completely solved by Cai and Chen~\cite{CaiChenComplex}. 
They gave three conditions for tractability. 
Their theorem applies to any finite set $\calF$ of algebraic complex-valued functions defined on
an arbitrary finite domain~$\domain$. They show that \#CSP$(\calF)$ is solvable in polynomial time
if all three conditions are satisfied, and \#P-hard otherwise.
This type of theorem is called a \emph{dichotomy theorem} 
because of the dichotomy that it provides between polynomial-time solvability and \#P-hardness.
 
Much less is known about the \emph{approximation} version of $\NCSP(\cF)$.
Although a trichotomy is known in the  Boolean relational case, 
where the domain is Boolean and the  functions in $\cF$ are $0$-$1$ valued \cite{reltrichotomy},
the only general classifications that do not restrict the range of functions in $\cF$ in this way are
in the so-called ``conservative'' case,  where 
we consider all finite subsets of the set of functions containing
$\calF$  and   all unary
functions. 
We will be working in the conservative case in this paper.
The conservative case has also been studied in the context of decision and optimisation CSPs \cite{Bul11, KZ13}.
Building on and generalising work in the Boolean case~\cite{bulatov_expressibility_2013},
the following dichotomy theorem was given in~\cite{ApproxCSP}. 
It applies to any  set $\calF$ of algebraic functions  
with non-negative rational values defined
on an arbitrary finite domain~$\domain$ as long as $\calF$ contains all unary functions.
\begin{itemize}
\item If a certain condition is satisfied, then for any finite $\calG \subset \calF$, \#CSP$(\calG)$ 
can be (exactly) solved in polynomial time.
\item Otherwise, there is a finite $\calG \subset \calF$ such that even approximating the output of
\#CSP$(\calG)$ is intractable, in the sense that it is \#BIS-hard.
\end{itemize}
The problem \#BIS is a canonical approximate counting problem which is not believed to have an efficient approximation
algorithm. We will not need it further in this paper, so we do not give the details here -- the interested reader is referred
to~\cite{ApproxCSP}.
The Ising model with local fields is an example where the set $\calF$ is conservative
(since it contains every possible local vertex weight).
Thus, the result of~\cite{ApproxCSP} includes as a special case the fact that approximating the partition function of the Ising Model with Local Fields is \#BIS-hard \cite{FerroIsing}.

The important thing, from our point of view, is the methods that were used to
derive the result of~\cite{ApproxCSP}.
An important idea in the study of decision CSPs was the notion of a \emph{relational clone} \cite{CohenJeavons}.
The idea is that, if you want to understand the complexity of the decision analogue of $\NCSP(\cF)$, 
the right thing to do is to characterise the set of functions that can be ``expressed'' by $\cF$  -- and the right tools
for this (in the decision context) are  relational clones, which are the sets of relations expressible using 
certain kinds of formulas called \emph{pp-formulas}.
Informally, relational clones formalise the notion of ``gadgets'' in reductions.
The first proposal for extending  pp-formulas to the setting of
 (Boolean) weighted counting CSPs is the \emph{T-constructibility} notion of Yamakami~\cite{TomoCSP}.
 In \cite{bulatov_expressibility_2013} it was discovered that a more  close-fitting  formulation for studying 
 Boolean
 weighted counting CSPs 
 in the setting of approximation and approximation-preserving reductions
 is the notion of ``pps$_\omega$-definability''.
 Just as pp-definability is closely related to polynomial-time reductions between decision CSPs,
 the authors of \cite{bulatov_expressibility_2013} showed that pps$_\omega$-definability with its corresponding
 \emph{functional clones} captures approximation-preserving reductions between weighted counting CSPs.
The paper \cite{bulatov_expressibility_2013} explored the space of functional clones in the Boolean case, and used
this to partly classify the complexity of the approximation version of $\NCSP(\cF)$
in the Boolean conservative case. These functional clones were again used in~\cite{ApproxCSP} 
(along with other key ideas from optimisation CSPs) to derive the more 
general classification.

In this paper, we will be concerned primarily with problems such as the Monomer-Dimer model,
which cannot be captured using the weighted counting CSP framework.
As  is well known, it can be captured in the holant framework.
 To see this, let us first consider the problem of expressing the Monomer-Dimer Model in the weighted counting CSP
framework. Clearly, the domain is $D=\{0,1\}$ and the variables are the edges of the graph, since
these are assigned the spins. For each possible vertex degree~$d$,
we would like to include a constraint function $f_d$ which has value~$1$ if all of its arguments
are~$0$, has value $\gamma^{1/2}$ if exactly one of its arguments is~$1$, and has value~$0$ otherwise.
For some degree bound~$\Delta$, let $\cF$ be the collection of $\cF= \{f_1,\ldots,f_\Delta\}$ of these constraint
functions. However, it is easy to see
that the Monomer-Dimer Model on graphs of degree at most~$\Delta$ is
not the same as the problem $\NCSP(\cF)$ because each edge appears in  exactly two constraints in the
Monomer-Dimer Model, but the $\NCSP(\cF)$ problem has no such restriction.
Indeed, the Monomer-Dimer Model on bounded-degree graphs is the same as the $\NCSP(\cF)$ problem with the additional
restriction that every variable appears in exactly two constraints.
The holant framework is equivalent to the counting  CSP framework with this restriction.

Here is a more convenient definition of the holant framework (see~\cite{cai_complexity_2017} for much more
detail and discussion). Since we will work within the Boolean domain, we restrict attention to this domain from now on.
A \emph{signature grid} $\Omega=(G,\cF,\sigma)$ over some set of functions $\cF$ consists of a finite (multi-)graph $G=(V,E)$ and a map $\sigma$ which assigns to each vertex $v\in V$ a function $\sigma(v)=f_v\in\cF$ such that 
the arity of~$f_v$ is equal to the degree of~$v$. One should think of  $G$ as being the graph associated with the counting problem.
The two spins are~$0$ and~$1$. Configurations assign spins to the edges of~$G$.
The  value $\sigma(v)$ gives the function $f_v$ which will
provide the local weight at vertex~$v$.
The map $\sigma$  also determines which edge incident on $v$ corresponds to which input of $f_v$.
Thus, a configuration is an assignment $\bx:E\to\{0,1\}$ of Boolean values to the edges of~$G$
and this induces a weight $w_\bx = \prod_{v\in V} f_v(\bx|_{E(v)})$, where $\bx|_{E(v)}$ is the restriction of $\bx$ to the edges incident on $v$ (ordered according to their correspondence with the inputs of $f_v$).

The partition function, or \emph{holant} of $\Omega$, 
denoted $Z_\Omega$, is the sum of $w_\bx$ over all assignments of values $\bx$ to the edges:
\[
 Z_\Omega = \sum_{\bx:E\to\{0,1\}} \prod_{v\in V} f_v(\bx|_{E(v)}).
\]
In the Monomer-Dimer model, the function $f_v$ 
has value~$1$ if all of its arguments are~$0$, 
value~$\gamma^{1/2}$ if exactly one of its arguments is~$1$ and  
value~$0$ otherwise.

{\bf First contribution:\quad} The first contribution of this paper is to develop
a theory of holant clones. Just as the functional clone theory captures 
expressibility in the weighted CSP framework, the new theory of holant clones
captures expressibility in the holant framework.
As the reader will see, the theory differs substantially from the theory of functional
clones for weighted counting CSP. The difference arises partly because of the 
crucial role of \emph{holographic transformations} in holant problems,  partly because of the
different kind of summation over variables that turns out to be relevant, and partly because
of the key role of the bipartite case. 
In general, the holant clone generated by  a given set of  functions is a subset of the functional clone generated by the same set.
However (see Corollary~\ref{cor:holant-functional}),
as would be expected, if the ternary equality function is contained in a particular holant clone, then the two clones are the same.
Section~\ref{secfour} 
explores the connection between certain polynomial-time reductions between holant problems and 
holant clones.

{\bf Second contribution:\quad} Our second contribution is to explore the landscape of conservative holant clones.
In order to describe this, we need some notation.
Let $\AA$  be the set of algebraic complex numbers.
For each non-negative integer $k$,  let $\allf_k$ be the set of all functions $\{0,1\}^k\to\AA$.
Finally, let $\allf = \bigcup_{k\in\NN} \allf_k$ be
 the set of all algebraic-complex valued functions of Boolean variables.
 For consistency with the literature, we also use
$\cU=\allf_1$ to denote  the set of   unary functions.
We use $\hc{\cF}$ to denote the holant clone generated by a set $\cF$ (this will be defined in Section~\ref{sec:holantclone}).
We say that $\cF$ is \emph{universal in the conservative case} if
$\hc{\cF \cup \cU} = \allf$.
We can now give an informal description of our theorem. 

\begin{thm} [Informal statement of Theorem~\ref{thm:conservative_hc}]
Suppose that $\cF$ is a subset of $\allf$. 
Unless $\cF$ satisfies one of four (explicit) conditions,
it is  universal in the conservative case.
\end{thm}

 If $\cF$ is  not universal in the conservative case, then
we also give a concise generating set for $\hc{\cF \cup \cU}$ (see Lemmas~\ref{lem:closure_clone} and~\ref{lem:cE_generators}).
 
{\bf Third contribution:\quad} Our  third contribution is to use our theory of  holant clones
 to completely classify the approximation version of the Boolean conservative holant problem.
 Before describing this contribution, we discuss the relevant literature.
 Given any finite set $\cF\subseteq \allf$, the holant problem is defined as follows.

\begin{description}[noitemsep]
 \item[Name] $\hol(\cF)$
 \item[Instance] A signature grid $\Omega=(G,\cF,\sigma)$.
 \item[Output] $Z_\Omega$.
\end{description}

A complete dichotomy for the exact version of $\hol(\cF)$ is not known, but 
such a dichotomy has been discovered (between tractable and \#P-hard)
in the Boolean conservative case~\cite{cai_dichotomy_2011}. 
Their precise dichotomy is stated as Theorem~\ref{thm:holant-star-dichotomy} of this paper --
it turns out that the exact problem is tractable exactly when $\cF$ fails to be universal in the conservative case.
Note that
 the result of~\cite{cai_dichotomy_2011} has been improved by one of the authors~\cite{Backens},
who has  weakened the ``conservativity'' assumption -- requiring only certain unary ``pinning'' functions to be in~$\calF$.
In a different direction, progress was made by Lin and Wang, who drop the conservativity restriction entirely, though their classification applies only to non-negative real-valued functions \cite{lin_complexity_2017}.
See \cite{cai_complexity_2017}  for other progress which avoids restricting
to the Boolean domain, or to the conservative case, by considering other restrictions such as symmetry.
 
In the (Boolean) holant framework, as in the counting CSP framework,
much less is known about the complexity of \emph{approximating} the partition function.
The main result  that is known is a theorem of Yamakami~\cite[Theorem 3.5]{yamakami_approximation_2012}, 
 which shows that, for every ternary function $f$ in a certain class, which he calls SIG$_1$,
one of the following is true (depending on~$f$):
\begin{itemize}
\item  for any finite $\calF$ containing $f$ and  unary functions, 
$\hol(\cF)$ is tractable, or 
\item there exists a finite set $S$ of  unary functions
such that a certain complex-valued satisfiability problem
reduces in polynomial time to $\hol(S \cup \{f\})$.
\end{itemize}

We use our holant clone theory in order to develop the following  complete
approximation classification for conservative Boolean holant problems.

\begin{thm}[Informal statement of Theorem~\ref{thm:main}]\label{thm2}
 Suppose that $\cF$ is a finite subset of $\allf$. 
 \begin{itemize}
 \item
 If $\cF$ satisfies one of the four conditions of Theorem~\ref{thm:conservative_hc}
 then, by the result of \cite{cai_dichotomy_2011},
  for any finite subset $S\sse\cU$, the problem $\hol(\cF \cup S)$ is  solvable exactly in polynomial time.
 \item Otherwise,
 \cite{cai_dichotomy_2011} shows that 
 there exists a finite subset $S\sse\cU$ such that $\hol(\cF \cup S)$ is \numP-hard. 
 We show that   there is a  finite subset $S\sse\cU$ such that even
 approximating the norm of the output of $\hol(\cF \cup S)$ within a factor of~$1.01$
 is \numP-hard and so is approximating the argument of the output of $\hol(\cF \cup S)$ within an additive $\pm \pi/3$.
 \end{itemize}
\end{thm}

Our Theorem~\ref{thm2} extends Yamakami's partial classification in two ways: First,
it applies to all conservative sets of functions (instead of to sets of functions
containing a single ternary function from SIG$_1$, along with unary functions).
Second, the hardness that we obtain is \numP-hardness, meaning that
an approximate solution to the holant problem would allow some \#P-hard problem to be solved exactly.
This is in contrast to hardness based on a reduction from a particular satisfiability problem.

 As \cite{bezakova_inapproximability_2017} explains in a different context, the exact number~``$1.01$''  
in the statement of Theorem~\ref{thm2} is not
important.   For any $\epsilon>0$, the theorem may be combined with a standard powering argument
to show that it is \#P-hard to approximate $\abs{Z_\Omega} $
within a factor of $2^{n^{1-\epsilon}}$. The ``$\pi/3$''  can also be improved, though it cannot be improved by
adding an arbitrary constant, since   $\HA(\cF; 2 \pi)$
is a trivial problem for any $\cF$.

\section{Preliminaries}\label{s:preliminaries}

 Recall from the introduction that
 $\AA$  is the set of algebraic complex numbers.
For each non-negative integer $k$, we denote by $\allf_k$ the set of all functions $\{0,1\}^k\to\AA$.
The set of all algebraic-complex valued functions of Boolean variables is $\allf := \bigcup_{k\in\NN} \allf_k$.
We often identify the set of nullary functions, $\allf_0$, with $\AA$.
Given a function $f\in \allf_k$, we use $\ari(f)$ to denote the arity of~$f$.
If $\ari(f)=k$, then we refer to $f$ as a ``$k$-ary function''.

Given any positive integer $n$, let $[n]:=\{1\zd n\}$.
Given any  $k$-ary function~$f$  and any permutation~$\pi:[k]\to[k]$, let
 $f_\pi(\vc{x}{k}):=f(x_{\pi(1)}\zd x_{\pi(k)})$.

\begin{dfn}\label{dfn:tensor_product}
 Suppose that $f$ is a $k$-ary function and $g$ is an $\ell$-ary function with $k,\ell\geq 1$, then the \emph{tensor product} of $f$ and $g$ is the $(k+\ell)$-ary function
 \[
  h(\vc{x}{k+\ell}) = f(\vc{x}{k}) g(x_{k+1}\zd x_{k+\ell}).
 \]
 This definition can be extended straightforwardly to tensor products of more than two functions.
\end{dfn}

A function $h'$ of arity $k\geq 2$ is \emph{decomposable as a tensor product} if there exists some permutation $\pi:[k]\to [k]$ such that $h'_\pi$ is a tensor product.
We say that a function of arity at least 2 is \emph{degenerate} if it is a tensor product of unary functions.
Borrowing terminology from quantum theory, a function of arity at least 2 is called \emph{entangled} if it is not decomposable as a tensor product.

A function is \emph{symmetric} if it depends only on the Hamming weight of its input, i.e.\ its value is invariant under any permutation of the arguments.
If $f$ is a $k$-ary symmetric function, we write $f=[f_0,\vc{f}{k}]$, where $f_j$ is the value that $f$ takes on inputs of Hamming weight $j$.

\begin{obs}\label{obs:symmetric_degenerate_entangled}
 Any symmetric function of arity at least 2 is either degenerate or entangled.
\end{obs}

A binary function $g$ is sometimes written as a 2 by 2 matrix
\[
 g = \begin{pmatrix} g(0,0) & g(0,1) \\ g(1,0) & g(1,1) \end{pmatrix}.
\]
This matrix is invertible if and only if $g$ is non-degenerate, i.e.\ $g$ is not a tensor product of unary functions.

Denote by $\EQ_k$ the $k$-ary equality function defined as $\EQ_k = [1,0\zd 0,1]$ with $(k-1)$ zeroes in the list.
In other words, $\EQ_1$ is the constant-1 function and for $k>1$, $\EQ_k$ is the function that is 1 if all inputs are equal and 0 otherwise.
The binary disequality function is denoted by $\NEQ$, i.e.\ $\NEQ=[0,1,0]$ and the binary NAND function is $\NAND=[1,1,0]$.
The function $\ONE_k$ is the function that is 0 except on inputs of Hamming weight 1, i.e.\ $\ONE_k = [0,1,0\zd 0]$ with $(k-1)$ zeroes at the end.
Note that this function is sometimes given different names in other papers.

\subsection{Boolean holant problems} 

\label{s:boolean_holant_problems}

Recall from the Introduction that a \emph{signature grid} $\Omega=(G,\cF,\sigma)$ over some set of functions $\cF\sse\allf$ consists of a finite multigraph $G=(V,E)$ and a map $\sigma$ which assigns to each vertex $v\in V$ a function $\sigma(v)=f_v\in\cF$  
such that the arity of~$f_v$ is equal to the degree of~$v$ in~$G$. The map $\sigma$  also determines which edge incident on $v$ corresponds to which input of $f_v$.
 Any assignment $\bx:E\to\{0,1\}$ of Boolean values to the edges induces a weight $w_\bx = \prod_{v\in V} f_v(\bx|_{E(v)})$, where $\bx|_{E(v)}$ is the restriction of $\bx$ to the edges incident on $v$ (ordered according to their correspondence with the inputs of $f_v$).
The holant of $\Omega$, denoted $Z_\Omega$, is  $
 Z_\Omega = \sum_{\bx:E\to\{0,1\}} \prod_{v\in V} f_v(\bx|_{E(v)})$.
The holant problem for a finite set of functions $\cF$ is defined as follows.

\begin{description}[noitemsep]
 \item[Name] $\hol(\cF)$
 \item[Instance] A signature grid $\Omega=(G,\cF,\sigma)$.
 \item[Output] $Z_\Omega$.
\end{description}

\begin{rem}Although we study infinite sets~$\cF$ for the purpose of
defining holant clones,
we restrict the definition of the computational problem $\hol(\cF)$ 
to situations in which the set~$\cF$ is finite.
The reason that we do this is to 
avoid issues about how to represent the functions in the computational input.
This is typical in the CSP literature and also in some of the holant literature.
When infinite sets~$\cF$ are allowed, 
it is common to use the notation $\hol^*(\cF)$ 
to refer to the conservative holant problem; this notation implies that all unary functions are added to the set $\cF$.
Since we have restricted the computational problem to finite sets, 
we instead treat the conservative case by quantifying over finite subsets $S$ of~$\cF$,
as in the  statement of Theorem~\ref{thm2}. 
\end{rem}

In order to analyse the   complexity of  \emph{approximately} solving holant problems, we define two additional  computational problems in the style of \cite{bezakova_inapproximability_2017}.
For any complex number $c\in\CC$, denote by $\Arg(c)$ the principal value of the argument of $c$ in the range $[0,2\pi)$ and let $\arg(c):=\{\Arg(c)+2k\pi\mid k\in\ZZ\}$.
The first of the approximation problems is parameterised by a real number $\kappa \geq 1$,
and the second by $\rho \in [0,2\pi)$.

\begin{description}[noitemsep]
 \item[Name] $\HN(\cF; \kappa)$
 \item[Instance] A signature grid $\Omega=(G,\cF,\sigma)$.
 \item[Output] If $\abs{Z_\Omega}=0$ then the algorithm may output any rational number. Otherwise it must output a rational number $\hat{N}$ such that $\hat{N}/\kappa \leq \abs{Z_\Omega} \leq \kappa\hat{N}$.
\end{description}

\begin{description}[noitemsep]
 \item[Name] $\HA(\cF; \rho)$
 \item[Instance] A signature grid $\Omega=(G,\cF,\sigma)$.
 \item[Output] If $\abs{Z_\Omega}=0$ then the algorithm may output any rational number. Otherwise it must output a rational number $\hat{A}$ such that for some $a\in\arg(Z_\Omega)$, $|\hat{A}-a|\leq\rho$.
\end{description}

Given two  computational problems $A$ and $B$
(in this paper, these will be computational counting problems),
we write $A\leq_{PT} B$ to denote the fact that there exists a polynomial-time Turing reduction from $A$ to $B$.
If $A\leq_{PT} B$ and $B\leq_{PT} A$, we write $A\equiv_{PT} B$ and say that the two problems are inter-reducible.

Suppose that $\vc{\cF}{n}$ are finite sets of functions and that $\vc{g}{m}$ are functions.
Where no confusion is likely to arise, we often use the shorthand $\vc{\cF}{n},\vc{g}{m}$ to mean $\cF_1\cup\ldots\cup\cF_n\cup\{\vc{g}{m}\}$, so, for example,
$ \hol(\vc{\cF}{n},\vc{g}{m}) $ is a shorthand for $ \hol(\cF_1\cup\ldots\cup\cF_n\cup\{\vc{g}{m}\})$.

\subsection{Holographic transformations and holant problems on bipartite graphs}

Holographic transformations are an important technique for analysing holant problems and for proving reductions -- indeed, they are the origin of the name ``holant''.
To define these transformations, the following fact will be useful:
The functions of arity $n$ are in bijection with vectors in $\AA^{2^n}$ by mapping each function $f$ to the vector of its values, $\mathbf{f}$.

Denote by $\otimes$ the Kronecker product of matrices and let $\GL$ be the group of invertible 2 by 2 matrices over $\AA$.

\begin{rem}
 Under the above bijection between functions and vectors, and considering vectors  as single-column matrices, the tensor product of functions (cf.\ Definition~\ref{dfn:tensor_product}) corresponds exactly to the Kronecker product: the function $h$ is the tensor product of functions $f$ and $g$ if and only if $\mathbf{h}=\mathbf{f}\otimes\mathbf{g}$.
\end{rem} 

\begin{dfn}\label{def:holographic_transformation}
 Suppose $M\in\GL$ and suppose $f$ is an $n$-ary function.
 Define $M\t{1}:=M$ and $M\t{k+1}:=M\t{k}\otimes M$ for any positive integer $k$.
 If $n=0$, set $M\circ f := f$, otherwise define $M\circ f$ to be the function whose values correspond to the vector $M\t{n}\mathbf{f}$: this is called a \emph{holographic transformation} of $f$ by $M$.
 For any set of functions $\cF$, let $M\circ\cF:=\{M\circ f\mid f\in\cF\}$.
\end{dfn}

\begin{ex}
 Suppose $M = \left(\begin{smallmatrix}1&i\\1&-i\end{smallmatrix}\right)$ and $f=\EQ_2$.
 Then
 \[
  M\t{2}\mathbf{f} = \left( \begin{pmatrix}1&i\\1&-i\end{pmatrix} \otimes \begin{pmatrix}1&i\\1&-i\end{pmatrix} \right) \begin{pmatrix}1\\0\\0\\1\end{pmatrix} =  \begin{pmatrix}1&i&i&-1\\ 1&-i&i&1\\ 1&i&-i&1\\1&-i&-i&-1\end{pmatrix} \begin{pmatrix}1\\0\\0\\1\end{pmatrix} = \begin{pmatrix}0\\2\\2\\0\end{pmatrix},
 \]
 so $(M\circ f)(x,y) = 2\cdot\NEQ(x,y)$.
\end{ex}

Denote by $\cO$ the set of 2 by 2 orthogonal matrices over $\AA$; these will often be used in holographic transformations.
Let $X=\left(\begin{smallmatrix}0&1\\1&0\end{smallmatrix}\right)$ and let $I=\smm{1&0\\0&1}$.
The following two matrices (which are sometimes given different names in other papers) will also play a special role when it comes to holographic transformations:
\[
 K_1 = \frac{1}{\sqrt{2}}\begin{pmatrix}1&1\\i&-i\end{pmatrix} \qquad\text{and}\qquad K_2 = \frac{1}{\sqrt{2}}\begin{pmatrix}1&1\\-i&i\end{pmatrix}.
\]

\begin{obs}\label{obs:Z_properties}
 The two matrices differ by a bit flip:
 \[
  K_1 X = \frac{1}{\sqrt{2}}\begin{pmatrix}1&1\\i&-i\end{pmatrix} \begin{pmatrix}0&1\\1&0\end{pmatrix} = \frac{1}{\sqrt{2}}\begin{pmatrix}1&1\\-i&i\end{pmatrix} = K_2.
 \]
 Furthermore,
 \[
  K_1^T K_1 = \frac{1}{\sqrt{2}}\begin{pmatrix}1&i\\1&-i\end{pmatrix} \frac{1}{\sqrt{2}}\begin{pmatrix}1&1\\i&-i\end{pmatrix} = \begin{pmatrix}0&1\\1&0\end{pmatrix} = \frac{1}{\sqrt{2}}\begin{pmatrix}1&-i\\1&i\end{pmatrix} \frac{1}{\sqrt{2}}\begin{pmatrix}1&1\\-i&i\end{pmatrix} = K_2^T K_2.
 \]
 so $K_1^{-1}=X K_1^T$ and $K_2^{-1}=X K_2^T$.
\end{obs}

\begin{dfn}\label{dfn:bipartite_signature_grid}
 Let $\cF,\cG$ be two finite sets of functions and suppose $G=(V,W,E)$ is a bipartite graph with vertex sets $V$ and $W$.
 The \emph{bipartite signature grid} $\Omega=(G,\cF\mid\cG,\sigma)$ is a signature grid in which each vertex in $V$ is assigned a function in $\cF$ and each vertex in $W$ is assigned a function in $\cG$.
\end{dfn}

The corresponding bipartite holant problem is denoted $\hol(\cF\mid\cG)$, and similarly we define $\HN(\cF\mid\cG;\kappa)$ and $\HA(\cF\mid\cG;\rho)$.

\section{Holant gadgets, functional clones, and holant clones}
\label{s:holant_gadgets_clones}

We have noted in the Introduction that
the theory of relational clones and the theory of functional clones have 
been useful for analysing the complexity of decision CSPs and counting CSPs 
\cite{CohenJeavons, bulatov_expressibility_2013, ApproxCSP, bulatov_functional_2017, backens_boolean_2018}.
We adapt these theories to the holant setting as a way of formalising the notion of realisability by gadgets or $\cF$-gates.

A \emph{gadget} over the set of functions $\cF$, or \emph{$\cF$-gate} \cite{cai_complexity_2017} is defined by a signature grid over a finite multigraph with some ``dangling edges''.
Formally, a \emph{multigraph with dangling edges} is a tuple $G=(V,E,E')$ where $V$ are the vertices, $E$ are the internal edges, and $E'$ are the dangling edges.
An internal edge is an ordinary undirected edge, i.e.\ an unordered pair of vertices.
A dangling edge consists of a single vertex, the other endpoint of the edge is undetermined.

The holant corresponding to such a gadget is a function of the assignments to the dangling edges, this is the function realised by the gadget.
To formalise this, suppose $\abs{E'}=k$ and choose some ordering $\vc{e}{k}$ of the dangling edges.
Let $\vc{x}{k}\in\{0,1\}$ be some assignment of Boolean values to the dangling edges and suppose $\by:E\to\{0,1\}$ is an assignment of values to the internal edges.
Define $\by'$ to be the extension of $\by$ to the domain $E\cup E'$ satisfying $\by'(e_j)=x_j$ for all $e_j\in E'$.
Then the gadget $\Gm=(G,\cF,\sigma)$ defines the function $g_\Gm:\{0,1\}^k\to\AA$ given by
\begin{equation}\label{eq:gadget-definition}
 g_\Gm(\vc{x}{k}) = \sum_{\by:E\to\{0,1\}} \prod_{v\in V} g_v(\by'|_{E(v)}),
\end{equation}
where $\by'|_{E(v)}$ is the restriction of $\by'$ to the edges (both internal and dangling) which are incident on the vertex $v$, in order of their correspondence to the arguments of $\sigma(v)=f_v$.
A function $f$ is said to be \emph{realisable} over $\cF$ if there exists a gadget $\Gm$ using only functions from $\cF$ such that $f=g_\Gm$.

To formalise this notion of realisability, we first recap the definition of functional clones (which formalise the equivalent notion of realisability in counting CSPs) and then adapt this definition to the holant setting.

\subsection{Functional clones}

Most of the material in this section is taken from \cite{backens_boolean_2018} and reproduced here for completeness.

Let $\cF\sse\allf$ be a set of functions and $V=\{\vc{v}{n}\}$ a set of variables.
An \emph{atomic formula} $\varphi=g(v_{i_1}\zd v_{i_k})$ consists of a function $g\in\cF$ and a scope $(v_{i_1}\zd v_{i_k})\in V^k$, where $k=\ari(g)$.
The scope may contain repeated variables.
Given an assignment $\bx:V\to\{0,1\}$, the atomic formula $\varphi$ specifies the function\footnote{The definition of the function specified by an atomic formula has been modified slightly so as to be better compatible with the definition of holant clones in the next section. This change has no effect on the definition of functional clones (Definition~\ref{def:functional_clone}), however note the remark after Lemma~\ref{lem:functional-clone-closure}.} $f_\varphi:\{0,1\}^k\to\AA$ given by
\[
 f_\varphi(\bx|_\varphi) = g(x_{i_1}\zd x_{i_k}),
\]
where $\bx|_\varphi$ is the application of $\bx$ to the scope of $\varphi$, and $x_j=\bx(v_j)$ for all $j\in [n]$.

\begin{dfn}\label{def:pps}
 A \emph{primitive product summation formula} (pps-formula) in variables $V$ over $\cF$ has the form
 \[
  \psi = \sum_{v_{n+1},\ldots,v_{n+m}} \prod_{j=1}^s \varphi_j,
 \]
 where $\varphi_j$ are all atomic formulas over $\cF$ in the variables $V'=\{\vc{v}{n+m}\}$.
 The variables in $V$ are called \emph{free variables}, those in $V'\setminus V$ are called \emph{bound variables}.
\end{dfn}

The pps-formula $\psi$ represents a function $f_\psi:\{0,1\}^n\to\AA$ given by
\begin{equation}\label{eq:pps-formula}
 f_\psi(\bx) = \sum_{\by\in\{0,1\}^m} \prod_{j=1}^s f_{\varphi_j}(\bx,\by|_{\varphi_j}),
\end{equation}
where $\bx$ is an assignment $V\to\{0,1\}$, $\by$ is an assignment $V'\setminus V\to\{0,1\}$, and $\bx,\by|_{\varphi_j}$ is the application of those assignments to the variables appearing in the scope of $\varphi_j$.
If $f$ is represented by some pps-formula $\psi$ over $\cF$, it is said to be \emph{pps-definable} over $\cF$.

\begin{dfn}[\cite{bulatov_expressibility_2013}]\label{def:functional_clone}
 The \emph{functional clone} generated by $\cF$ is the set of all functions in $\allf$ that can be represented by a pps-formula over $\cF\cup\{\EQ_2\}$.
 It is denoted by $\ang{\cF}$.
\end{dfn}

There is a another perspective on functional clones \cite{bulatov_functional_2017}.
In the following, let $f\in\allf_k$.
We say a $(k+1)$-ary function $h$ arises from $f$ by \emph{introduction of a fictitious argument} if $h(\vc{x}{k+1})=f(\vc{x}{k})$ for all $\vc{x}{k}\in\{0,1\}$.
Let $g\in\allf_k$, then the \emph{product} of $f$ and $g$ is the function $h$ satisfying $h(\vc{x}{k}) = f(\vc{x}{k}) g(\vc{x}{k})$.
The function $h$ resulting from $f$ by a \emph{permutation of the arguments} $\pi:[k]\to[k]$ is $h(\vc{x}{k})=f(x_{\pi(1)}\zd x_{\pi(k)})$.
Furthermore, a $(k-1)$-ary function $h$ arises from $f$ by \emph{summation} if $h(\vc{x}{k-1})=\sum_{x_k\in\{0,1\}} f(\vc{x}{k})$.

\begin{lem}[{\cite[Section~1.1]{bulatov_functional_2017}}]\label{lem:functional-clone-closure}
 For any $\cF\sse\allf$, $\ang{\cF}$ is the closure of $\cF\cup\{\EQ_2\}$ under introduction of fictitious arguments, product, permutation of arguments, and summation.
\end{lem}

\begin{rem}
 The proof of Lemma~\ref{lem:functional-clone-closure} has to be modified slightly to take into account the different definition of the function defined by an atomic formula.
 The only change is the following: the function arising from $f\in\allf_k$ by introduction of a fictitious argument can no longer be represented by simply adding a free variable to the pps-formula; instead it is now represented as
 \begin{equation}\label{eq:fictitious_argument}
  \sum_{y\in\{0,1\}} f(\vc{x}{k})\EQ_2(x_{k+1},y).
 \end{equation}
 Indeed, if $f$ is represented by a pps-formula over $\cF$ according to our definition, then $f$ is represented by the same pps-formula according to the definition from \cite{bulatov_expressibility_2013}.
 On the other hand, if $f$ is represented by a pps-formula using the original definition, there are two cases:
 \begin{itemize}
  \item If every free variable appears in the scope of some atomic formula, then $f$ is represented by the same pps-formula using our definition.
  \item If there exist some free variables which do not appear in the scope of any atomic formula, then $f$ is represented under our definition by a pps-formula that is modified as in \eqref{eq:fictitious_argument}: for each free variable $v$ that does not occur in the scope of any atomic formula in the original pps-formula, we introduce a new atomic formula $\EQ_2(v,v')$ where $v'$ is a new bound variable that does not occur anywhere else in the pps-formula.
 \end{itemize}
 Hence the two definitions yield the same sets of functions.
\end{rem}

\subsection{Holant clones and their properties}\label{sec:holantclone}

Functional clones play an important role in the analysis of counting CSPs.
We now introduce the notion of \emph{holant clones}, which do the same for holant problems.\footnote{In the CSP context, it was useful to
define both functional clones and also more complicated objects called 
pps$_\omega$-definable functional clones \cite{bulatov_expressibility_2013}.
Similarly, we could 
define two versions of holant clones here, one with a limit operation. We do not do so because our main complexity  theorem, 
Theorem~\ref{thm:main}, goes through even without this refinement.
} 

Define the multiplicity $m_\varphi(v)$ of a variable $v\in V$ in an atomic formula $\varphi$ to be the number of times that $v$ appears in the scope of $\varphi$.
Define the multiplicity $m_\psi(v)$ of $v$ in the pps-formula $\psi$ as $m_\psi(v) := \sum_{j=1}^s m_{\varphi_j}(v)$, where the sum is over all atomic formulas of $\psi$.
We say $\psi$ is a \emph{\ppsh-formula} if $m_\psi(v)=1$ for all $v\in V$ and $m_\psi(v)=2$ for all $v\in V'\setminus V$.
In other words, $\psi$ is a \ppsh-formula if every free variable appears exactly once in the scope of some atomic formula and every bound variable appears exactly twice in the scope of some atomic formulas.

\begin{dfn}
 The \emph{holant clone} generated by $\cF\sse\allf$ is the set of all functions that can be represented by a \ppsh-formula over $\cF$. It is denoted $\hc{\cF}$.
\end{dfn}

In this paper, we have restricted functions to the set $\allf$, where the domain is Boolean.
However, the definition of functional clone and holant clone can be extended naturally 
to functions with other domains. The following lemma relates holant clones to holant gadgets.

\begin{lem}\label{lem:holant_clone_gadget}
 Suppose $\cF\sse\allf$ is a set of functions and $f\in\allf$ is a function.
 Then $f$ can be realised by a holant gadget over $\cF$ if and only if $f\in\hc{\cF}$.
\end{lem}
\begin{proof}
 Suppose $f\in\hc{\cF}$ has arity $n$.
 There exists a \ppsh-formula $\psi = \sum_{v_{n+1},\ldots,v_{n+m}} \prod_{j=1}^s \varphi_j$ such that
 \[
  f(\bx) = f_\psi(\bx)
  = \sum_{\by\in\{0,1\}^m} \prod_{j=1}^s g_j(\bx,\by|_{\varphi_j}),
 \]
 where $g_j$ is the function appearing in the atomic formula $\varphi_j$.
 Suppose $V'=\{\vc{v}{n+m}\}$ is the full set of variables of $\psi$ and $V=\{\vc{v}{n}\}$ is the set of free variables of $\psi$.
 
 Let $G=(W,E,E')$ be the multigraph with dangling edges defined as follows:
 \begin{itemize}
  \item $W = [s]$,
  \item for each bound variable $v\in V'\setminus V$ appearing in the scopes of atomic formulas $\varphi_j$ and $\varphi_\ell$ (where $j$ may be equal to $\ell$), there is an internal edge $\{j,\ell\}$ in $E$, and
  \item for each free variable $w\in V$ appearing in the scope of the atomic formula $\varphi_j$, there is a dangling edge $\{j\}$ in $E'$.
 \end{itemize}
 The map from bound variables to internal edges and the map from free variables to dangling edges are well-defined since each bound variable has multiplicity 2 and each free variable has multiplicity 1.
 Let $\Omega=(G,\cF,\sigma)$, where $\sigma:W\to\cF$ is the function that maps $j$ to $g_j$.
 Then $f = g_\Omega$, so $f$ is realisable by a holant gadget.
 
 Conversely, suppose $f$ is realised by some holant gadget $\Omega=(G,\cF,\sigma)$.
 By comparing \eqref{eq:gadget-definition} and \eqref{eq:pps-formula} it is straightforward to see that there exists some pps-formula representing $f$, where the bound variables correspond to the internal edges of the gadget and the free variables correspond to the dangling edges.
 Now, each bound variable of the pps-formula has multiplicity 2 since it corresponds to an internal edge in the gadget, which has two endpoints.
 Similarly, each free variable of the pps-formula has multiplicity 1 since it corresponds to a dangling edge in the gadget, which has only one defined endpoint.
 Therefore the pps-formula is actually a \ppsh-formula, and $f\in\hc{\cF}$.
\end{proof}

As for functional clones, there is an alternative definition of holant clones as the closure of the generating set of functions under certain operations. 
Suppose $f\in\allf_k$ and $g\in\allf_\ell$.
Recall from Section~\ref{s:preliminaries} that the \emph{tensor product} of $f$ and $g$ is the $(k+\ell)$-ary function $h$ satisfying $h(\vc{x}{k+\ell}) = f(\vc{x}{k}) g(x_{k+1}\zd x_{k+\ell})$.
A $(k-2)$-ary function $h$ arises from $f$ by \emph{contraction} if $h(\vc{x}{k-2}) = \sum_{y\in\{0,1\}} f(\vc{x}{k-2},y,y)$.

It will be useful to first consider the operation of taking tensor products in more detail.
We say a pps-formula $\psi$ is a \emph{\pph-formula} if $m_\psi(v)=1$ for all $v\in V$ and $V'=V$: i.e.\ $\psi$ is a \ppsh-formula with no bound variables.
For any $\cF\sse\allf$, let $\tcl{\cF}$ denote the set of all functions that can be represented by a \pph-formula over $\cF$.

\begin{lem}\label{lem:closure_free}
 Suppose $\cF\sse\allf$, then $\tcl{\cF}$ is the closure of $\cF$ under tensor product and permutation of arguments.
\end{lem}
\begin{proof}
 We first show that $\tcl{\cF}$ is closed under the two operations.
 Suppose $f,g\in\tcl{\cF}$ where $\ari(f)=k$ and $\ari(g)=\ell$.
 Then there exist some \pph-formula $\psi_f$ over $\cF$ in variables $V_f$ which represents $f$, and some \pph-formula $\psi_g$ over $\cF$ in variables $V_g$ which represents $g$, where
 \[
  \psi_f = \prod_{p=1}^{s_f} \varphi_{p} \qquad\text{and}\qquad \psi_g = \prod_{q=1}^{s_g} \theta_{q}.
 \]
 Without loss of generality, assume the sets of variables $V_f$ and $V_g$ are disjoint.
 \begin{itemize}
  \item Tensor product: suppose $h(\vc{x}{k+\ell}) = f(\vc{x}{k}) g(x_{k+1}\zd x_{k+\ell})$.
  Consider the expression
  \[
   \psi_f \psi_g = \left( \prod_{p=1}^{s_f} \varphi_{p} \right) \left( \prod_{q=1}^{s_g} \theta_{q} \right).
  \]
  The right-hand side is a valid \pph-formula with free variables $V_f\cup V_g$; the multiplicities are unchanged since $V_f$ and $V_g$ are disjoint.
  This \pph-formula represents $h$.
  \item Permutation of arguments: suppose $\pi:[k]\to[k]$ is a permutation and $h(\vc{x}{k})=f(x_{\pi(1)}\zd x_{\pi(k)})$.
  Then $\psi_f$ can be transformed into a \pph-formula for $h$ by permuting the variables in the scope of each atomic formula.
 \end{itemize}
 This concludes the proof that $\tcl{\cF}$ is closed under tensor product and permutation of arguments.
 
 To complete the argument, suppose $f\in\tcl{\cF}$ is represented by a \pph-formula $\psi_f = \prod_{j=1}^{s} \varphi_{j}$, where $g_j\in\cF$ is the function in the atomic formula $\varphi_j$.
 As $\psi_f$ is a \pph-formula, the scopes of the different atomic formulas are disjoint and no variable appears more than once.
 Now, for any assignment $\bx:V_f\to\{0,1\}$, denote by $\bx|_{\varphi_j}$ the application of $\bx$ to the scope of $\varphi_j$.
 Then $f(\bx) = \prod_{j=1}^{s} g_j(\bx|_{\varphi_j})$ is decomposable as a tensor product: i.e.\ $f$ arises from $\cF$ by tensor product and permutation of arguments.
 This implies that $\tcl{\cF}$ is contained in the closure of $\cF$ under tensor product and permutation of arguments.
 
 But $\tcl{\cF}$ is closed under these operations and it contains $\cF$, so it must itself be the closure of $\cF$ under tensor product and permutation of arguments.
\end{proof}

\begin{lem}\label{lem:tensor_closure-holant_clone}
 For any set of functions $\cF\sse\allf$, the set $\hc{\cF}$ is the closure of $\tcl{\cF}$ under contraction.
\end{lem}
\begin{proof}
 First we show that every \ppsh-formula arises from some \pph-formula by contraction: consider the \ppsh-formula $\psi_f$ over $\cF$ in variables $V_f'$ with free variables $V_f$, where
 \[
  \psi_f = \sum_{v\in V_f'\setminus V_f} \prod_{p=1}^{s_f} \varphi_{p}.
 \]
 Suppose $V_f = \{\vc{v}{n}\}$ and $V_f' = \{\vc{v}{n+m}\}$, and let $V_g := V_f \cup \{\vc{w}{2m}\}$, where the $w_k$ are new variables.
 Define $\psi_g$ to be the \ppsh-formula in variables $V_g$ that arises from $\psi_f$ by
 \begin{itemize}
  \item for $k\in[m]$, replacing one occurrence of the bound variable $v_{n+k}$  with $w_{2k-1}$, replacing the other occurrence of $v_{n+k}$ with $w_{2k}$, and
  \item dropping the sums.
 \end{itemize}
 Then $\psi_g$ is a \pph-formula; so it represents a function $g\in\tcl{\cF}$.
 Furthermore, $f$ arises from $g$ via $m$ contractions.
 Thus, $\hc{\cF}$ is in the closure of $\tcl{\cF}$ under contraction.

 Additionally, every \pph-formula is a \ppsh-formula, therefore $\tcl{\cF}\sse\hc{\cF}$.
 Thus it only remains to show that $\hc{\cF}$ is closed under contraction.
 Suppose $f\in\hc{\cF}$ where $\ari(f)=k$.
 Then there exist some \ppsh-formula $\psi_f$ over $\cF$ in variables $V_f$ which represents $f$, where
 \[
  \psi_f = \sum_{v\in V_f'\setminus V_f} \prod_{p=1}^{s_f} \varphi_{p}.
 \]
 Let $h(\vc{x}{k-2}) = \sum_{y\in\{0,1\}} f(\vc{x}{k-2},y,y)$. 
  Suppose $V_f = \{\vc{v}{n}\}$ and $V_f' = \{\vc{v}{n+m}\}$, then $V_h = \{\vc{v}{n-2}\}$ and
  \[
   V_h' = (V_f'\cup\{w\})\setminus\{v_{n-1},v_n\} = \{\vc{v}{n-2},v_{n+1}\zd v_{n+m}, w\},
  \]
  where $w\notin V_f'$ is a new variable.
  For each atomic formula $\varphi_p$ in $\psi_f$, let $\varphi_p'$ be the atomic formula that is equal to $\varphi_p$ except that any occurrence of $v_{n-1}$ or $v_n$ in the scope is replaced by $w$.
  Both $v_{n-1}$ and $v_n$ are free variables in $\psi_f$, so they each appear with multiplicity 1 in $\psi_f$.
  Then $w$ appears with multiplicity 2 in
  \[
   \psi = \sum_{v\in V_h'\setminus V_h} \prod_{p=1}^{s_f} \varphi_{p}',
  \]
  and the multiplicities of any variables in $V_f'\cap V_h'$ are unchanged.
  Thus, $\psi$ is a \ppsh-formula which represents $h$.
 Hence $h\in\hc{\cF}$.

 This implies that $\hc{\cF}$ is closed under contraction, therefore $\hc{\cF}$ is the closure of $\tcl{\cF}$ under contraction.
\end{proof}

\begin{prop}\label{prop:holant-clone-closure}
 For any set of functions $\cF\sse\allf$, the set $\hc{\cF}$ is the closure of $\cF$ under tensor product, permutation of arguments, and contraction.
\end{prop}
\begin{proof}
 By Lemmas~\ref{lem:closure_free} and \ref{lem:tensor_closure-holant_clone}, $\hc{\cF}$ arises from $\cF$ by first taking the closure under tensor product and permutation of arguments, and then taking the closure under contraction.
 This immediately implies that $\hc{\cF}$ is closed under contraction and is contained in the closure of $\cF$ under tensor product, permutation of arguments, and contraction.
 Hence it suffices to show that $\hc{\cF}$ is closed under tensor product and permutation of arguments.

 Suppose $f,g\in\hc{\cF}$ where $\ari(f)=k$ and $\ari(g)=\ell$.
 Then there exist some \ppsh-formula $\psi_f$ over $\cF$ in variables $V_f$ which represents $f$, and some \ppsh-formula $\psi_g$ over $\cF$ in variables $V_g$ which represents $g$, where
 \[
  \psi_f = \sum_{v\in V_f'\setminus V_f} \prod_{p=1}^{s_f} \varphi_{p} \qquad\text{and}\qquad \psi_g = \sum_{v\in V_g'\setminus V_g} \prod_{q=1}^{s_g} \theta_{q}.
 \]
 Without loss of generality, assume the sets of variables $V_f'$ and $V_g'$ are disjoint.
 \begin{itemize}
  \item Tensor product: suppose $h(\vc{x}{k+\ell}) = f(\vc{x}{k}) g(x_{k+1}\zd x_{k+\ell})$.
  Consider the expression
  \[
   \psi_f \psi_g = \left( \sum_{v\in V_f'\setminus V_f} \prod_{p=1}^{s_f} \varphi_{p} \right) \left( \sum_{v\in V_g'\setminus V_g} \prod_{q=1}^{s_g} \theta_{q} \right) = \sum_{v\in (V_f'\setminus V_f)\cup(V_g'\setminus V_g)} \prod_{p=1}^{s_f} \varphi_p \prod_{q=1}^{s_g} \theta_q.
  \]
  The right-hand side is a valid \ppsh-formula with free variables $V_f\cup V_g$ and bound variables $(V_f'\cup V_g')\setminus(V_f\cup V_g)$.
  The multiplicities are unchanged since $V_f'$ and $V_g'$ are disjoint.
  This \ppsh-formula represents $h$.
  \item Permutation of arguments: suppose $\pi:[k]\to[k]$ is a permutation and $h(\vc{x}{k})=f(x_{\pi(1)}\zd x_{\pi(k)})$.
  Then $\psi_f$ can be transformed into a \ppsh-formula for $h$ by permuting the variables in the scope of each atomic formula.
 \end{itemize}
 This concludes the proof that $\hc{\cF}$ is closed under tensor product and permutation of arguments.
 Thus, it follows from Lemmas~\ref{lem:closure_free} and \ref{lem:tensor_closure-holant_clone} that $\hc{\cF}$ is the closure of $\cF$ under tensor product, permutation of arguments, and contraction.
\end{proof}

\begin{cor}\label{cor:holant_closed}
 Let $\cF\sse\allf$ be any set of functions and suppose $f\in\hc{\cF}$.
 Then $\hc{\cF,f}=\hc{\cF}$.
\end{cor}

Having determined some properties of holant clones, we now consider the relationship between holant clones and functional clones.

\begin{obs}\label{obs:holant-contained-functional}
 $\hc{\cF}\sse\ang{\cF}$ for any set of functions $\cF\sse\allf$ because any \ppsh-formula is a pps-formula.
\end{obs}

Unlike a functional clone, a holant clone does not automatically contain the function $\EQ_2$.
Furthermore, even if $\EQ_2\in\hc{\cF}$, the holant clone is not necessarily equal to the functional clone generated by $\cF$: as an example, we will show that $\hc{\{\EQ_2\}}\neq\ang{\{\EQ_2\}}$.

Note that $\EQ_k\in\ang{\{\EQ_2\}}$ for any $k$ because $\EQ_1(x)=\sum_{y\in\{0,1\}}\EQ_2(x,y)$ and, for $k>2$,
\[
 \EQ_k(\vc{x}{k}) = \prod_{j=1}^{k-1}\EQ_2(x_j,x_k),
\]
so $\EQ_k$ can be represented by a pps-formula over $\{\EQ_2\}$ with $k$ variables and $(k-1)$ atomic formulas.

On the other hand, $\EQ_k$ cannot be in $\hc{\{\EQ_2\}}$ if $k$ is odd.
Indeed, consider any function $f\in\hc{\{\EQ_2\}}$, then there exists some \ppsh-formula $\psi_f = \sum_{v\in V_f'\setminus V_f} \prod_{j=1}^{s} \varphi_{j}$ over $\{\EQ_2\}$ such that
\[
 f(\bx) = \sum_{\by\in\{0,1\}^m} \prod_{j=1}^s \EQ_2(\bx,\by|_{\varphi_j}).
\]
Now, $\EQ_2$, the single generating function of $\hc{\{\EQ_2\}}$ has even arity.
A tensor product of $s$ copies of $\EQ_2$ also has even arity, and permutations of arguments do not affect the arity.
Furthermore, each contraction decreases the arity by 2, so $f$ must have even arity.
But $f$ was arbitrary; thus we have shown that any function in $\hc{\{\EQ_2\}}$ has even arity.
Hence $\EQ_k\notin\hc{\{\EQ_2\}}$ if $k$ is odd, which implies $\hc{\{\EQ_2\}}\neq\ang{\{\EQ_2\}}$.

Yet, as we now show, there is a function such that any holant clone containing it is a functional clone: the ternary equality function $\EQ_3$.
Indeed, Cai, Huang and Lu already argued that adding the ternary equality function to the set of allowed functions $\cF$ is sufficient to inter-reduce a counting CSP and a holant problem, i.e.\ $\NCSP(\cF) \equiv_{PT} \hol(\cF,\EQ_3)$ \cite[Proposition~1]{cai_Holant_2012}.

\begin{prop}\label{prop:holant-functional-clone}
 Let $\cF\sse\allf$ be any set of functions.
 Then $\hc{\cF,\EQ_3}=\ang{\cF}$.
\end{prop}
\begin{proof}
 By Observation~\ref{obs:holant-contained-functional}, $\hc{\cF,\EQ_3}\sse\ang{\cF,\EQ_3}$ for any $\cF$.
 To show that $\hc{\cF,\EQ_3}\sse\ang{\cF}$, it therefore suffices to show that $\EQ_3\in\ang{\cF}$, which implies $\ang{\cF,\EQ_3}=\ang{\cF}$ by \cite[Lemma~2.1]{bulatov_expressibility_2013}.
 But
 \[
  \EQ_3(x_1,x_2,x_3) = \EQ_2(x_1,x_2) \EQ_2(x_1,x_3),
 \]
 so $\EQ_3(x_1,x_2,x_3)$ can be represented as a pps-formula over $\{\EQ_2\}$ in the variables $V=\{v_1,v_2,v_3\}$, which has two atomic formulas with scopes $(v_1,v_2)$ and $(v_1,v_3)$, respectively.
 Thus $\EQ_3\in\ang{\cF}$ for any $\cF$, and hence $\hc{\cF,\EQ_3}\sse\ang{\cF,\EQ_3}=\ang{\cF}$.
 
 To show the inclusion $\ang{\cF}\sse\hc{\cF,\EQ_3}$, first note that $\cF\sse\hc{\cF,\EQ_3}$ by definition.
 Furthermore, $\EQ_2\in\hc{\cF,\EQ_3}$ since
 \[
  \EQ_2(x_1,x_2) = \sum_{y,z\in\{0,1\}} \EQ_3(x_1,x_2,y) \EQ_3(y,z,z),
 \]
 i.e.\ $\EQ_2$ is represented by a \ppsh-formula over $\{\EQ_3\}$.
 Therefore, $\cF\cup\{\EQ_2\}\sse\hc{\cF,\EQ_3}$, which implies $\ang{\cF}\sse\ang{\hc{\cF,\EQ_3}}$.
 It thus suffices to prove that $\ang{\hc{\cF,\EQ_3}}=\hc{\cF,\EQ_3}$.
 By Lemma~\ref{lem:functional-clone-closure}, this is equivalent to showing that $\hc{\cF,\EQ_3}$ is already closed under introduction of fictitious arguments, product, permutation of arguments, and summation.
 Closure under permutation of arguments is shown in Proposition~\ref{prop:holant-clone-closure}.
 Now let $f,g\in\hc{\cF,\EQ_3}$ be $k$-ary functions and consider the remaining three operations.
 \begin{itemize}
  \item Introduction of fictitious arguments: suppose $h(\vc{x}{k+1}) = f(\vc{x}{k})$.
  Then
  \[
   h(\vc{x}{k+1}) = \sum_{y\in\{0,1\}} f(\vc{x}{k}) \EQ_3(x_{k+1},y,y),
  \]
  so $h$ can be represented by a \ppsh-formula over $f$ and $\EQ_3$, and thus $h\in\hc{\cF,\EQ_3}$ by Corollary~\ref{cor:holant_closed}.
  \item Product: suppose $h(\vc{x}{k}) = f(\vc{x}{k}) g(\vc{x}{k})$. Then
  \[
   h(\vc{x}{k}) = \sum_{\vc{y}{k},\vc{z}{k}\in\{0,1\}} f(\vc{y}{k}) g(\vc{z}{k}) \prod_{j=1}^k \EQ_3(x_j,y_j,z_j).
  \]
  Thus $h$ can be represented by a \ppsh-formula over $f,g$, and $\EQ_3$, so $h\in\hc{\cF,\EQ_3}$ by Corollary~\ref{cor:holant_closed}.
  \item Summation: suppose $h(\vc{x}{k-1})=\sum_{x_k\in\{0,1\}} f(\vc{x}{k})$. Then
  \[
   h(\vc{x}{k-1}) = \sum_{y,z\in\{0,1\}} f(\vc{x}{k-1},y) \EQ_3(y,z,z),
  \]
  so $h$ can be represented by a \ppsh-formula over $f$ and $\EQ_3$, and thus $h\in\hc{\cF,\EQ_3}$ by Corollary~\ref{cor:holant_closed}.
 \end{itemize}
 We have shown that $\hc{\cF,\EQ_3}$ is a functional clone and contains $\ang{\cF}$.
 In combination with the first part of the proof, this implies that $\hc{\cF,\EQ_3}=\ang{\cF}$.
\end{proof}

\begin{cor}\label{cor:holant-functional}
 Let $\cF\sse\allf$ be any set of functions and suppose $\EQ_3\in\hc{\cF}$.
 Then $\hc{\cF}=\ang{\cF}$.
\end{cor}
This follows immediately from Proposition~\ref{prop:holant-functional-clone} and Corollary~\ref{cor:holant_closed}.

\begin{rem}
 Holant clones are related to the notion of $T_2$-constructibility in \cite[Section~4.1]{yamakami_approximation_2012}: if $f$ is $T_2$-constructible over some set $\cF\sse\allf$, then $f$ is in the closure of $\hc{\cF\cup\{\dl_0,\dl_1,\EQ_1\}}$ under scaling, i.e.\ $f\in\{c\cdot g\mid c\in\Anz,\; g\in\hc{\cF\cup\{\dl_0,\dl_1,\EQ_1\}}\}$.
 Conversely, any function in $\hc{\cF}$ is $T_2$-constructible over $\cF$.
 
 The two notions of ``membership in the holant clone'' and ``$T_2$-constructibility'' coincide if $\cF$ contains the unary functions $\dl_0,\dl_1$, and $\EQ_1$, as well as all non-zero nullary functions.
 We work with holant clones here because they will be better suited to the analysis of problems that are not conservative (i.e.\ where the unary functions $\dl_0$, $\dl_1$, and $\EQ_1$ may not be available).
 
 For historical reasons, functional clones are not generally closed under scaling; we define holant clones to be compatible with that.
 
 Lin and Wang have also considered the set of functions that can be realised from $\cF$ using gadgets, which they denote $S(\cF)$, though they do not abstract the definition away from the notion of holant gadgets \cite{lin_complexity_2017}.
\end{rem}

\subsection{Holographic transformations of holant clones}

We now consider the relationship of holant clones with certain holographic transformations.

\begin{obs}\label{obs:holographic_tensor}
 Suppose $\cF\sse\allf$ is a set of functions and $M$ is a 2 by 2 matrix over $\AA$.
 Then $\tcl{M\circ\cF}=M\circ\tcl{\cF}$ because $\tcl{\cF}$ arises from $\cF$ via tensor product and permutation of arguments by Lemma~\ref{lem:closure_free}.
\end{obs}

\begin{lem}\label{lem:hc_orthogonal_holographic}
 Let $\cF\sse\allf$ be any set of functions and let $O\in\cO$.
 Then $\hc{O\circ\cF}=O\circ\hc{\cF}$.
\end{lem}
\begin{proof}
 Let $\om$ be the binary function that corresponds to the matrix $O$, i.e.\ $\om(x,y)=O_{xy}$.
 Then, for any $g\in\allf_k$,
 \begin{equation}\label{eq:function_oht}
  (O\circ g)(\vc{x}{k}) = \sum_{\vc{z}{k}\in\{0,1\}} g(\vc{z}{k}) \prod_{\ell=1}^k \om(x_\ell,z_\ell)
 \end{equation}
 
 We first show that $\hc{O\circ\cF}\sse O\circ\hc{\cF}$.
 Suppose $h\in\hc{O\circ\cF}$, then there exists a \ppsh-formula $\psi_h = \sum_{v\in V_h'\setminus V_h} \prod_{j=1}^{s} \varphi_{j}$ representing $h$, where $V_h=\{\vc{v}{n}\}$ is the set of free variables and $V_h'=\{\vc{v}{n+m}\}$ is the full set of variables.
 For each bound variable $v\in V_h'\setminus V_h$, we say that the occurrence of $v$ at position $\ell$ in the scope of some atomic formula $\varphi_j$ is the \emph{first occurrence} of $v$ if $v$ does not occur in the scope of any atomic formula $\varphi_k$ with $k<j$ and if $v$ does not occur in the scope of $\varphi_j$ in a position with index less than $\ell$.
 The occurrence of $v$ that is not the first occurrence is called the \emph{second occurrence}.
 
 Now, for each atomic formula $\varphi_j$ in $\psi_h$, the specified function $f_{\varphi_j}$ takes the form $O\circ g_j$ for some $g_j\in\cF$, hence it can be expressed as in \eqref{eq:function_oht}.
 Let $a_j:=\ari(g_j)$, and write $(\bx,\by|_{\varphi_j})_\ell$ for the $\ell$-th argument of the function represented by $\varphi_j$, then
 \[
  h(\bx) = \sum_{\by\in\{0,1\}^m} \prod_{j=1}^s \sum_{z_{j,1}\zd z_{j,a_j}\in\{0,1\}} g_j(z_{j,1}\zd z_{j,a_j}) \prod_{\ell_j=1}^{a_j} \om((\bx,\by|_{\varphi_j})_{\ell_j},z_{j,\ell_j}).
 \]
 We now re-index the arguments as follows: for $\ell\in[a_j]$, define
 \[
  (\bz|_{\varphi_j})_\ell := \begin{cases} z_k & \text{if } (\bx,\by|_{\varphi_j})_\ell = x_k, \\ z_{k+n} & \text{if } (\bx,\by|_{\varphi_j})_\ell = y_k \text{ and this is the first occurrence of $y_k$,} \\ z_{k+m+n} & \text{if } (\bx,\by|_{\varphi_j})_\ell = y_k \text{ and this is the second occurrence of $y_k$}. \end{cases}
 \]
 Then we can write
 \[
  h(\bx) = \sum_{\by\in\{0,1\}^m} \sum_{\bz\in\{0,1\}^{2m+n}} \prod_{j=1}^s g_j(\bz|_{\varphi_j}) \prod_{k=1}^n \omega(x_k,z_k) \prod_{\ell=1}^{m} \omega(y_\ell,z_{\ell+n}) \omega(y_\ell,z_{\ell+m+n}).
 \]
 Note that each $z_k$ appears exactly once as the argument of some $g_j$, which is why it was possible to pull out and combine the sums over all $z_k$.
 Furthermore, note that the arguments $y_\ell$ only appear in the final part of the product, therefore we can rearrange the sums further:
 \begin{align*}
  h(\bx) &= \sum_{\bz\in\{0,1\}^{2m+n}} \prod_{j=1}^s g_j(\bz|_{\varphi_j}) \prod_{k=1}^n \omega(x_k,z_k) \prod_{\ell=1}^{m} \left(\sum_{y_\ell\in\{0,1\}} \omega(y_\ell,z_{\ell+n}) \omega(y_\ell,z_{\ell+m+n}) \right) \\
  &= \sum_{\bz\in\{0,1\}^{2m+n}} \prod_{j=1}^s g_j(\bz|_{\varphi_j}) \prod_{k=1}^n \omega(x_k,z_k) \prod_{\ell=1}^{m} \EQ_2(z_{\ell+n},z_{\ell+m+n}).
 \end{align*}
 Here, $\sum_{y_\ell\in\{0,1\}} \omega(y_\ell,z_{\ell+n}) \omega(y_\ell,z_{\ell+m+n}) = \EQ_2(z_{\ell+n},z_{\ell+m+n})$ by orthogonality of $O$.
 Now, let
 \[
  (\bz'|_{\varphi_j})_\ell = \begin{cases} z_k & \text{if } (\bx,\by|_{\varphi_j})_\ell = x_k, \\ z_{k+n} & \text{if } (\bx,\by|_{\varphi_j})_\ell = y_k \end{cases}
 \]
 for each $\ell\in[a_j]$, i.e.\ $\bz'$ does not distinguish between the two occurrences of $y_k$.
 Then, by summing out $z_{m+n+1}\zd z_{2m+n}$ and rearranging the remaining sums,
 \[
  h(\bx) = \sum_{\vc{z}{n}\in\{0,1\}} \left( \sum_{z_{n+1}\zd z_{m+n}\in\{0,1\}} \prod_{j=1}^s g_j(\bz'|_{\varphi_j}) \right) \prod_{k=1}^n \omega(x_k,z_k).
 \]
 It is straightforward to see that the term in parentheses is a \ppsh-formula over $\cF$, and the part outside the parentheses denotes a holographic transformation by $O$.
 Therefore, $h\in O\circ\hc{\cF}$.
 
 By going through the same algebraic steps in the opposite direction -- i.e.\ starting from a \ppsh-formula witnessing that $h\in O\circ\hc{\cF}$, introducing binary equality functions and replacing them with contractions over $\omega$ -- one can show that $O\circ\hc{\cF}\sse\hc{O\circ\cF}$.
 Therefore, $\hc{O\circ\cF}=O\circ\hc{\cF}$ as desired.
\end{proof}

The following lemma will be useful when considering holographic transformations by $K_1$ or $K_2$.

\begin{lem}\label{lem:Z-contraction}
 Suppose $f,g\in\allf\setminus\allf_0$ and $K\in\{K_1,K_2\}$.
 Let $k:=\ari(f)$, $\ell:=\ari(g)$ and define $h(\vc{x}{k+\ell-2}) := \sum_{y_1,y_2\in\{0,1\}} f(y_1,\vc{x}{k-1}) g(x_k\zd x_{k+\ell-2},y_2) \NEQ(y_1,y_2)$.
 Then
 \[
  \sum_{y\in\{0,1\}} (K\circ f)(y,\vc{x}{k-1}) (K\circ g)(x_k\zd x_{k+\ell-2},y) = (K\circ h)(\vc{x}{k+\ell-2}).
 \]
 If $k\geq 2$, define $h'(\vc{x}{k-2}) := \sum_{y_1,y_2\in\{0,1\}} f(\vc{x}{k-2},y_1,y_2) \NEQ(y_1,y_2)$, then
 \[
  \sum_{y\in\{0,1\}} (K\circ f)(\vc{x}{k-2},y,y) = (K\circ h')(\vc{x}{k-2}).
 \]
\end{lem}
\begin{proof}
 Let $\zeta\in\allf_2$ be the function associated with the matrix $K$, i.e.\ $\zeta(x,y)=K_{xy}$.
 Then, for any $f'\in\allf_{k'}$, $(K\circ f')(\vc{x}{k'}) = \sum_{\vc{z}{k'}\in\{0,1\}} f'(\vc{z}{k'}) \prod_{j=1}^{k'} \zeta(x_j,z_j)$.
 Thus
 \begin{align*}
  \sum_{y\in\{0,1\}} &(K\circ f)(y,x_2\zd x_k) (K\circ g)(x_{k+1}\zd x_{k+\ell-1},y) \\
  &= \sum_{y\in\{0,1\}} \left( \sum_{\vc{z}{k}\in\{0,1\}} f(\vc{z}{k})  \zeta(y,z_1) \prod_{j=2}^{k} \zeta(x_j,z_j) \right) \\
  &\hphantom{= \sum_{y\in\{0,1\}}}\, \left( \sum_{z_{k+1}\zd z_{k+\ell}\in\{0,1\}} g(z_{k+1}\zd z_{k+\ell}) \zeta(y,z_{k+\ell}) \prod_{j=k+1}^{k+\ell-1} \zeta(x_j,z_j) \right) \\
  &= \sum_{\vc{z}{k+\ell}\in\{0,1\}} f(\vc{z}{k}) g(z_{k+1}\zd z_{k+\ell}) \left( \sum_{y\in\{0,1\}} \zeta(y,z_1) \zeta(y,z_{k+\ell}) \right) \prod_{j=2}^{k+\ell-1} \zeta(x_j,z_j).
 \end{align*}
 Now, the function $\sum_{y\in\{0,1\}} \zeta(y,z_1) \zeta(y,z_{k+\ell})$ corresponds to $(K^T K)_{z_1 z_{k+\ell}}$, and, by Observation~\ref{obs:Z_properties}, $K^T K=X$ for $K\in\{K_1,K_2\}$.
 Thus, $\sum_{y\in\{0,1\}} \zeta(y,z_1) \zeta(y,z_{k+\ell}) = \NEQ(z_1,z_{k+\ell})$.
 Therefore,
 \begin{align*}
  \sum_{y\in\{0,1\}} (K\circ f)&(y,x_2\zd x_k) (K\circ g)(x_{k+1}\zd x_{k+\ell-1},y) \\
  &= \sum_{\vc{z}{k+\ell}\in\{0,1\}} f(\vc{z}{k}) g(z_{k+1}\zd z_{k+\ell}) \NEQ(z_1,z_{k+\ell}) \prod_{j=2}^{k+\ell-1} \zeta(x_j,z_j) \\
  &= \sum_{z_{2}\zd z_{k+\ell-1}\in\{0,1\}} h(z_{2}\zd z_{k+\ell-1}) \prod_{j=2}^{k+\ell-1} \zeta(x_j,z_j),
 \end{align*}
 which is equal to $(K\circ h)(x_2\zd x_{k+\ell-1})$, as desired.
 
 The proof of the second statement is analogous.
\end{proof}

\subsection{Restricted holant clones and bipartite holant clones}
\label{s:restricted}

We now define a variant of holant clones in which only certain arguments of functions can be contracted together.
Let $\Ld$ be any finite set.
We will take $\Ld = \{L,R\}$.
 Then a \emph{labelled function} with labels $\Ld$ is a function $f\in\allf_k$ together with a \emph{type} $\ld_f\in\Ld^k$.
Suppose $\cF$ is a set of labelled functions with labels $\Ld$, and consider some atomic formula $\varphi$ in variables $V$, associated with the function $f\in\cF$.
We say an occurrence of the variable $v$ at position $k$ in the scope of $\varphi$ is associated with the label $a\in\Ld$ if $(\ld_f)_k=a$.
Let $N\sse\{ \{a,b\}\mid a,b\in\Ld\}$ be some set of unordered  pairs from $\Ld$.
A \ppsh-formula $\psi$ over $\cF$ is \emph{restricted by $N$} if the following holds: for each bound variable $v$ of $\psi$, the two associated labels $a_1$ and $a_2$ satisfy $\{a_1,a_2\}\in N$.
The type of the function represented by the \ppsh-formula is induced by the labels associated with its free variables.

\begin{dfn}
 Suppose $\cF,\cG\sse\allf$ are sets of functions.
 Take $\Ld:=\{L,R\}$ and $N:=\{\{L,R\}\}$.
 Let $\cF\sqcup\cG$ denote the labelled set in which each function from $\cF$ occurs with type all-$L$ and each function from $\cG$ occurs with type all-$R$:
 \[
  \cF\sqcup\cG := \{(f,(L\zd L))\mid f\in\cF\}\cup\{(g,(R\zd R))\mid g\in\cG\}.
 \]
 A \ppsh-formula over $\cF\sqcup\cG$ restricted by $N$ is called a \emph{bipartite \ppsh-formula}.
 The set of all functions that can be represented by bipartite \ppsh-formulas over $\cF\sqcup\cG$ is called the \emph{bipartite holant clone} and denoted $\hc{\cF\sqcup\cG}$.
 
 We will sometimes use the two subsets of a bipartite holant clone containing those functions in which all variables are labelled the same:
 \begin{align*}
  \bhc{\cF}{\cG}{L} &= \{f \mid (f,(L\zd L))\in\hc{\cF\sqcup\cG} \} \\
  \bhc{\cF}{\cG}{R} &= \{f \mid (f,(R\zd R))\in\hc{\cF\sqcup\cG} \}.
 \end{align*}
\end{dfn}

Note that a function may appear multiple times with different labels in a bipartite holant clone.

\begin{prop}
 Let $\cF,\cG\sse\allf$ be sets of functions, then $\hc{\cF\sqcup\cG}$ is the closure of $\cF\sqcup\cG$ under tensor product, permutation of arguments, and contraction of two arguments with distinct labels.
\end{prop}
\begin{proof}
 The condition that the two occurrences of each bound variables have to be associated with different labels corresponds exactly to allowing only those contractions in which the two arguments have distinct labels.
 Having noted this, the proof of the proposition is analogous to that of Proposition~\ref{prop:holant-clone-closure}.
\end{proof}

\begin{cor}\label{cor:bip_clone_closure}
 Let $\cF,\cG\sse\allf$ be sets of functions.
 Suppose that $f\in\bhc{\cF}{\cG}{L}$ and suppose that $g\in\bhc{\cF}{\cG}{R}$.
 Then the following hold:
 \begin{align*}
  \hc{(\cF\cup\{f\})\sqcup\cG} &= \hc{\cF\sqcup\cG} & 
  \hc{\cF\sqcup(\cG\cup\{g\})} &= \hc{\cF\sqcup\cG} \\
  \bhc{(\cF\cup\{f\})}{\cG}{L} &= \bhc{\cF}{\cG}{L} &
  \bhc{\cF}{(\cG\cup\{g\})}{L} &= \bhc{\cF}{\cG}{L} \\
  \bhc{(\cF\cup\{f\})}{\cG}{R} &= \bhc{\cF}{\cG}{R} &
  \bhc{\cF}{(\cG\cup\{g\})}{R} &= \bhc{\cF}{\cG}{R}.
 \end{align*}
\end{cor}

A \emph{bipartite holant gadget} is a holant gadget (cf.\ Section~\ref{s:holant_gadgets_clones}) defined over some bipartite signature grid (see Definition~\ref{dfn:bipartite_signature_grid}).
The following lemma is an analogue of Lemma~\ref{lem:holant_clone_gadget} in the non-bipartite setting.
The proof is essentially the same.

\begin{lem}\label{lem:bipartite_hc_gadget}
 Let $\cF,\cG\sse\allf$ be sets of functions and suppose $f\in\allf$.
 Then $f$ can be realised by a bipartite holant gadget over $\cF|\cG$ if and only if $f\in\hc{\cF\sqcup\cG}$.

\end{lem}

\begin{lem}\label{lem:hc_Z_holographic}
 Let $\cF\sse\allf$ and suppose $K\in\{K_1,K_2\}$.
 Then $\hc{K\circ\cF} = K\circ\bhc{\cF}{\{\NEQ\}}{L}$.
\end{lem}
\begin{proof}
 First, we prove $\hc{K\circ\cF} \sse K\circ\bhc{\cF}{\{\NEQ\}}{L}$.
 Assume $f\in\hc{K\circ\cF}$, then there exists a \ppsh-formula $\psi_f=\sum_{v\in V'\setminus V}\prod_{j=1}^s \varphi_j$ over $K\circ\cF$, where $V$ is the set of free variables and $V'$ is the set of all variables.
 Let $n:=\abs{V}$ and $m:=\abs{V'\setminus V}$.
 
 The proof is by induction on the number of bound variables $m$.
 The base case is $m=0$, i.e.\ $V'=V$.
 Then $f\in\tcl{K\circ\cF}$ by definition, and by Observation~\ref{obs:holographic_tensor}, $\tcl{K\circ\cF}=K\circ\tcl{\cF}$.
 But $K\circ\tcl{\cF}\sse K\circ\bhc{\cF}{\{\NEQ\}}{L}$, so $f\in K\circ\bhc{\cF}{\{\NEQ\}}{L}$ as desired.
 
 For the induction hypothesis, assume $g\in\hc{K\circ\cF}$ implies $g\in K\circ\bhc{\cF}{\{\NEQ\}}{L}$ if there exists a \ppsh-formula representing $g$ which has $m$ bound variables.
 Now consider a function $f\in\hc{K\circ\cF}$ which is represented by a \ppsh-formula $\psi_f=\sum_{v\in V'\setminus V}\prod_{j=1}^s \varphi_j$ with $(m+1)$ bound variables and $n$ free variables.
 We can write
 \[
  \psi_f = \sum_{v_{n+m+1}} \left( \sum_{v_{n+1}\zd v_{n+m}} \prod_{j=1}^s \varphi_j \right)
 \]
 Let $\psi'$ be the \ppsh-formula that arises from $\psi_f$ by replacing the two occurrences of the bound variable $v_{n+m+1}$ with distinct new free variables $w,w'$.
 Then $\psi'$ represents a function $h\in\hc{K\circ\cF}$ of arity $(n+2)$.
 Since $\psi'$ has $m$ bound variables, by the induction hypothesis, $h\in K\circ\bhc{\cF}{\{\NEQ\}}{L}$.
 Thus, there exists $h'\in\bhc{\cF}{\{\NEQ\}}{L}$ such that $h=K\circ h'$.
 Without loss of generality, assume that the variables $w,w'$ correspond to the final two arguments of $h$; otherwise permute the arguments, which does not affect membership in either of the holant clones.
 Then
 \[
  f(\vc{x}{n}) = \sum_{y\in\{0,1\}} (K\circ h')(\vc{x}{n},y,y).
 \]
 Thus, by Lemma~\ref{lem:Z-contraction}, $f=K\circ f'$, where
 \[
  f'(\vc{x}{n}) := \sum_{y_1,y_2\in\{0,1\}} h'(\vc{x}{n},y_1,y_2) \NEQ(y_1,y_2).
 \]
 But $h'\in\bhc{\cF}{\{\NEQ\}}{L}$ implies that $f'\in\bhc{\cF}{\{\NEQ\}}{L}$ since the contractions satisfy the bipartite structure, therefore $f\in K\circ\bhc{\cF}{\{\NEQ\}}{L}$ as desired.
 Hence $\hc{K\circ\cF} \sse K\circ\bhc{\cF}{\{\NEQ\}}{L}$.
 
 Now consider the opposite direction, i.e.\ we want to show that $K\circ\bhc{\cF}{\{\NEQ\}}{L} \sse \hc{K\circ\cF}$.
 Recall that each atomic formula in a bipartite \ppsh-formula has type either all-$L$ or all-$R$.
 Assume $h\in \bhc{\cF}{\{\NEQ\}}{L}$, then there exists a bipartite \ppsh-formula $\psi_h = \sum_{v\in V'\setminus V}\prod_{j=1}^s \varphi_j \prod_{k=1}^t\theta_k$ over $\cF\sqcup\{\NEQ\}$ representing $h$, where $\varphi_j$ are the atomic formulas of type all-$L$ and $\theta_k$ are the atomic formulas of type all-$R$.
 
 Again, the proof is by induction, this time on the number $t$ of atomic formulas of type all-$R$.
 The argument for the base case $t=0$ is analogous to the above since $t=0$ means there are no contractions.
 
 For the induction hypothesis, assume that $g\in K\circ\bhc{\cF}{\{\NEQ\}}{L}$ implies $g\in\hc{K\circ\cF}$ if there exists a \ppsh-formula representing $g$ which has $t$ atomic formulas of type all-$R$.
 Now consider a function $h\in\bhc{\cF}{\{\NEQ\}}{L}$ which is represented by a \ppsh-formula $\psi_h = \sum_{v\in V'\setminus V}\prod_{j=1}^s \varphi_j \prod_{k=1}^{t+1}\theta_k$ which has $(t+1)$ atomic formulas of type all-$R$.
 Let $n:=\abs{V}$ and $m:=\abs{V'\setminus V}$.
 The two variables appearing in the scope of $\theta_{t+1}$ have to be bound since $h$ has type all-$L$.
 Without loss of generality, assume they are $v_{n+1},v_{n+2}$, otherwise relabel the variables.
 We can thus write
 \[
  \psi_h = \sum_{v_{n+1},v_{n+2}} \left( \sum_{v_{n+3}\zd v_{n+m}} \prod_{j=1}^s \varphi_j \prod_{k=1}^{t}\theta_k \right) \theta_{t+1}.
 \]
 The term in parentheses is a valid bipartite \ppsh-formula of type all-$L$ with $t$ atomic formulas of type all-$R$.
 Thus it represents a function $h'\in\bhc{\cF}{\{\NEQ\}}{L}$.
 Furthermore, by the induction hypothesis, $K\circ h'\in\hc{K\circ\cF}$.
 Now,
 \[
  h(\vc{x}{n}) = \sum_{y_1,y_2\in\{0,1\}} h'(\vc{x}{n},y_1,y_2)\NEQ(y_1,y_2).
 \]
 But then by Lemma~\ref{lem:Z-contraction}, $(K\circ h)(\vc{x}{n}) = \sum_{y\in\{0,1\}} (K\circ h')(\vc{x}{n},y,y)$.
 So $K\circ h$ arises from $K\circ h'$ by contraction, thus $K\circ h\in\hc{K\circ\cF}$.
 Hence $K\circ\bhc{\cF}{\{\NEQ\}}{L} \sse \hc{K\circ\cF}$.
 
 By combining the two parts of the proof, we find that $\hc{K\circ\cF} = K\circ\bhc{\cF}{\{\NEQ\}}{L}$.
\end{proof}

\begin{lem}\label{lem:hc_Z_equal}
 Let $K\in\{K_1,K_2\}$ and suppose $\cF\sse\allf$ contains both $\EQ_2$ and $\NEQ$.
 Then $\hc{K\circ\cF}=K\circ\hc{\cF}$.
\end{lem}
\begin{proof}
 For any sets $\cG_1, \cG_2\sse\allf$, the bipartite holant clone $\hc{\cG_1\sqcup\cG_2}$ is a subset of the holant clone $\hc{\cG_1\cup\cG_2}$ because the latter does not involve any restrictions on which arguments can be contracted together.
 In particular, $\NEQ\in\cF$ implies that $\bhc{\cF}{\{\NEQ\}}{L}\sse\hc{\cF\cup\{\NEQ\}}=\hc{\cF}$.
 Thus, by Lemma~\ref{lem:hc_Z_holographic}, $\hc{K\circ\cF}=K\circ\bhc{\cF}{\{\NEQ\}}{L}\sse K\circ\hc{\cF}$.
 
 To show $K\circ\hc{\cF}\sse\hc{K\circ\cF}$, suppose $f\in\hc{\cF}$, so that $K\circ f\in K\circ\hc{\cF}$.
 The proof is by induction on the number of bound variables $m$ in the \ppsh-formula representing $f$.
 
 The base case is $m=0$, i.e.\ there are no bound variables.
 Thus $f\in\tcl{\cF}$ by definition.
 But by Observation~\ref{obs:holographic_tensor}, $K\circ\tcl{\cF} = \tcl{K\circ\cF}$, so $K\circ f\in\tcl{K\circ\cF}\sse\hc{K\circ\cF}$.
 
 For the induction hypothesis, assume $g\in \hc{\cF}$ implies $K\circ g\in\hc{K\circ\cF}$ if there exists some \ppsh-formula representing $g$ which has $m$ bound variables.
 Let $f\in\hc{\cF}$ be represented by $\psi_f=\sum_{v\in V'\setminus V} \prod_{j=1}^s \varphi_j$ with $\abs{V}=n$ and $\abs{V'\setminus V}=m+1$.
 We can write
 \[
  \psi_f = \sum_{v_{n+m+1}} \left( \sum_{v_{n+1}\zd v_{n+m}} \prod_{j=1}^s \varphi_j \right)
 \]
 Let $\psi'$ be the \ppsh-formula that arises from $\psi_f$ by replacing the two occurrences of the bound variable $v_{n+m+1}$ with distinct new free variables $w,w'$.
 Then $\psi'$ represents a function $f'\in\hc{\cF}$ of arity $(n+2)$.
 Since $\psi'$ has $m$ bound variables, by the induction hypothesis, $K\circ f'\in\hc{K\circ\cF}$.
 Without loss of generality, assume $f(\bx)=\sum_{y\in\{0,1\}} f'(\bx,y,y)$, otherwise permute the arguments (which does not affect membership in holant clones).
 Let $\zeta\in\allf_2$ be the function corresponding to the matrix $K$ and recall from Observation~\ref{obs:Z_properties} that $K^T K = \left(\begin{smallmatrix}0&1\\1&0\end{smallmatrix}\right)$.
 Thus,
 \begin{align*}
  (K\circ f)(\bx) &= \sum_{\bz\in\{0,1\}^n} f(\bz) \prod_{j=1}^n \zeta(x_j,z_j) = \sum_{\bz\in\{0,1\}^n} \left( \sum_{y\in\{0,1\}} f'(\bz,y,y) \right) \prod_{j=1}^n \zeta(x_j,z_j) \\
  &= \sum_{\bz\in\{0,1\}^n} \left( \sum_{\vc{y}{4}\in\{0,1\}} f'(\bz,y_1,y_4) \NEQ(y_1,y_2) \EQ_2(y_2,y_3) \NEQ(y_3,y_4) \right) \prod_{j=1}^n \zeta(x_j,z_j) \\
  &= \sum_{\bz\in\{0,1\}^n} \left( \sum_{\vc{y}{6}\in\{0,1\}} f'(\bz,y_1,y_4) \zeta(y_5,y_1) \zeta(y_5,y_2) \EQ_2(y_2,y_3) \zeta(y_6,y_3) \zeta(y_6,y_4) \right) \\
  &\qquad\qquad\qquad\prod_{j=1}^n \zeta(x_j,z_j) \\
  &= \sum_{y_5,y_6\in\{0,1\}} \left( \sum_{\bz\in\{0,1\}^n} \sum_{y_1,y_4\in\{0,1\}} f'(\bz,y_1,y_4) \zeta(y_5,y_1) \zeta(y_6,y_4) \prod_{j=1}^n \zeta(x_j,z_j) \right) \\
  &\hphantom{= \sum_{y_5,y_6\in\{0,1\}} }\; \left( \sum_{y_2,y_3\in\{0,1\}} \EQ_2(y_2,y_3) \zeta(y_5,y_2) \zeta(y_6,y_3) \right) \\
  &= \sum_{y_5,y_6\in\{0,1\}} (K\circ f')(\bx,y_5,y_6) (K\circ\EQ_2)(y_5,y_6).
 \end{align*}
 But $K\circ f'$ and $K\circ\EQ_2$ are in $\hc{K\circ\cF}$, so $K\circ f\in\hc{K\circ\cF}$.
 This implies that $K\circ\hc{\cF}\sse\hc{K\circ\cF}$ and therefore, together with the first part, $\hc{K\circ\cF}=K\circ\hc{\cF}$.
\end{proof}

\section{A complexity-related preorder on sets of functions}\label{secfour}

To simplify proofs that hold for both exact evaluation of holant values and for approximation, we define a preorder on finite sets of functions and pairs of finite sets of functions which encodes information about the complexity of the corresponding holant problems.

\begin{dfn}\label{dfn:function_reduction}
 Let $\cF,\cG,\cF',\cG'\sse\allf$ be finite sets of functions and take $S$ to be either $\cF$ or $\cF|\cG$ and take $S'$ to be either $\cF'$ or $\cF'|\cG'$.
 Then we write $S\leq S'$ if there exists a polynomial-time algorithm which takes a (possibly bipartite) signature grid $\Omega=(G,S,\sigma)$ and constructs a (possibly bipartite) signature grid $\Omega'=(G',S',\sigma')$ such that $Z_\Omega = Z_{\Omega'}$.
\end{dfn}

Note that the bipartite and non-bipartite cases can be mixed in Definition~\ref{dfn:function_reduction}. For example, it is perfectly fine
to take $S=\cF$ and $S' = \cF'|\cG'$.

\begin{obs}
 The relation $\leq$ on finite sets of functions and pairs of finite sets of functions is a preorder: it is reflexive by taking $\Omega'=\Omega$, and it is transitive.
\end{obs}

The following observation formalises the relationship between the preorder $\leq$ and the complexity of holant problems.

\begin{obs}\label{obs:preorder_reductions}
 Let $\cF,\cG,\cF',\cG'\sse\allf$ be finite sets of functions, let $S$ be either $\cF$ or $\cF|\cG$, and let $S'$ be either $\cF'$ or $\cF'|\cG'$.
 Suppose $S\leq S'$. Then
 \begin{itemize}
  \item $\hol(S)\leq_{PT}\hol(S')$,
  \item for any $\kappa\geq 1$, $\HN(S;\kappa)\leq_{PT}\HN(S';\kappa)$, and
  \item for any $\rho \in [0,2\pi)$, $\HA(S;\rho)\leq_{PT}\HA(S';\rho)$.
 \end{itemize} 
\end{obs}

The following results are well-known, we state them here in terms of holant clones and the preorder $\leq$ and give proofs for completeness.

\begin{lem}\label{lem:preorder_properties}
 Let $\cF,\cG\sse\allf$ be finite sets of functions.
 Then the following relationships hold:
 \begin{enumerate}[label={(\alph*)},ref={\alph*}]
  \item\label{it:1to_bipartite} $\cF\leq\cF|\{\EQ_2\}$
  \item\label{it:1from_bipartite} $\cF|\{\EQ_2\}\leq\cF$
  \item\label{it:2to_bipartite} $\cF\cup\cG\leq\cF\cup\{\EQ_2\}\big|\cG\cup\{\EQ_2\}$
  \item\label{it:2from_bipartite} $\cF|\cG\leq \cF\cup\cG$
  \item\label{it:from_clone} Suppose $\cF'\sse\hc{\cF}$ is finite, then $\cF'\leq\cF$.
  \item\label{it:from_bip_clone} Suppose $\cF'\sse\bhc{\cF}{\cG}{L}$ and $\cG'\sse\bhc{\cF}{\cG}{R}$ are finite, then $\cF'|\cG'\leq\cF|\cG$.
 \end{enumerate}
\end{lem}
\begin{proof}
 We prove the statements individually.
 \begin{enumerate}[label={(\alph*)},ref={\alph*}]
  \item Consider a signature grid $\Omega=(G,\cF,\sigma)$, where $G=(V,E)$.
   Let $G'=(V,E,E')$ be the bipartite graph with edges $E'$ and vertices $V\cup E$ that arises from $G$ by subdividing each edge.
   Define $\sigma'$ to be the function $\sigma'(v)=\sigma(v)$ for all $v\in V$ and $\sigma'(w)=\EQ_2$ for all $w\in E$.
   Then $\Omega'=(G',\cF|\{\EQ_2\},\sigma')$ is a valid signature grid and $Z_{\Omega'}=Z_\Omega$.
  \item Consider a signature grid $\Omega=(G,\cF|\{\EQ_2\},\sigma)$, where $G=(V,W,E)$ is bipartite.
   Note that all vertices in $W$ must have degree 2 since they are assigned the binary function $\EQ_2$.
   Let $G'=(V,E')$ be the graph resulting from $G$ by eliminating each vertex in $W$ and merging the two edges incident on it.
   (This cannot lead to any vertex-free loops since $G$ was bipartite.)
   Define $\sigma'(v)=\sigma(v)$ for all $v\in V$.
   Then $\Omega'=(G',\cF,\sigma')$ is a valid signature grid and $Z_{\Omega'}=Z_\Omega$.
  \item Consider a signature grid $\Omega=(G,\cF\cup\cG,\sigma)$, where $G=(V,E)$.
   Let $G'$ be the graph arising as follows: partition the original set of vertices $V$ into two parts $V_\cF=\{v\in V\mid \sigma(v)\in\cF\}$ and $V_\cG=\{v\in V\mid \sigma(v)\in\cG\}$.
   Subdivide any edge connecting two vertices that are assigned functions from the same set, adding the new vertex to the other part.
   The resulting graph has vertex parts $V_\cF\cup\{\{v,w\}\in E\mid v,w\in V_{\cG}\}$ and $V_\cG\cup\{\{v,w\}\in E\mid v,w\in V_{\cF}\}$ and it is bipartite.
   Define $\sigma'$ to be the function that acts as $\sigma$ on the original vertices and assigns $\EQ_2$ to each new vertex.
   Then $\Omega'=(G',\cF\cup\{\EQ_2\}|\cG\cup\{\EQ_2\},\sigma')$ is a valid signature grid and $Z_{\Omega'}=Z_\Omega$.
  \item Consider a signature grid $\Omega=(G,\cF|\cG,\sigma)$, where $G=(V,W,E)$ is bipartite.
   Let $G'=(V\cup W, E)$ be the graph that arises from $G$ by forgetting the bipartition.
   Define $\sigma':V\cup W\to\cF\cup\cG$ to be the function that acts as $\sigma$ on each vertex.
   Then $\Omega'=(G',\cF\cup\cG,\sigma')$ is a valid signature grid and $Z_{\Omega'}=Z_\Omega$.
  \item Consider a signature grid $\Omega=(G,\cF',\sigma)$, where $G=(V,E)$.
   By Lemma~\ref{lem:holant_clone_gadget}, each function $f\in\cF'$ can be realised by a holant gadget over $\cF$.
   Let $G'$ be the graph that arises from $G$ by replacing each vertex assigned a function $f\in\hc{\cF}$ by a gadget representing $f$, and let $\sigma'$ be the map that is induced by these gadgets.
   Then $\Omega'=(G',\cF,\sigma')$ is a valid signature grid and $Z_{\Omega'}=Z_\Omega$.
   As each gadget in $\cF'$ has a fixed finite size, this can be done in time linear in the size of the original signature grid.
  \item The argument is analogous to the previous case, using bipartite holant gadgets and Lemma~\ref{lem:bipartite_hc_gadget}, and noting that the labels match up correctly. \qedhere
 \end{enumerate}

\end{proof}

We will use the following theorem and corollary.

\begin{thm}[Valiant's holant theorem]\label{thm:holant_theorem}
 Suppose $\cF,\cG\sse\allf$ are two finite sets of functions and $M\in\GL$.
 Then $\cF|\cG \leq M\circ\cF|(M^{-1})^T\circ\cG$.
\end{thm}

\begin{cor}
 Suppose $\cF\sse\allf$ is a finite set of functions and suppose $O\in\cO$.
 Then $\cF \leq O\circ\cF$.
\end{cor}
\begin{proof}
 Note that $O\circ\EQ_2=\EQ_2$ for any orthogonal 2 by 2 matrix $O$.
 Therefore,
 \[
  \cF \leq \cF|\{\EQ_2\} \leq O\circ\cF|O\circ\{\EQ_2\} \leq O\circ\cF|\{\EQ_2\} \leq O\circ\cF,
 \]
 where the first relationship is by Lemma~\ref{lem:preorder_properties}\eqref{it:1to_bipartite}, the second relationship is Theorem~\ref{thm:holant_theorem}, the third relationship follows from reflexivity of $\leq$ together with the property of $\EQ_2$ and orthogonal matrices noted above, and the last relationship is by Lemma~\ref{lem:preorder_properties}\eqref{it:1from_bipartite}.
\end{proof}

\section{Conservative holant}

In this section, we first define some important sets of functions~$\cU$, $\cT$, $\cE$ and~$\cM$.
We give some lemmas describing their holant clones, and describing holographic transformations of them using $K_1$, $K_2$ and matrices in~$\cO$.
Section~\ref{s:ternary} gives a  partition of entangled ternary functions, which follows from the classification of entangled three-qubit states in
quantum theory, and which will be used later. Section~\ref{sec:known} states some known results that we will use.
Section~\ref{sec:univholant} develops theory which will allow us to use the quantum result that ``single-qubit and CNOT gates are universal''.
In our language, this is Lemma~\ref{lem:allf_from_universality}.
In Section~\ref{s:power}, we prove Theorem~\ref{thm:conservative_hc}, our main result about universality in the conservative case. 
In Section~\ref{sec:doapprox}, we prove Theorem~\ref{thm:main}, our main result about the complexity of approximating conservative holant problems.

Let $\cU:=\allf_1$ be the set of all unary functions for consistency with the literature.
The following subsets of $\allf$ will also be important:
\begin{itemize}
 \item the set of all unary and binary functions $\cT:=\allf_1\cup\allf_2$,
 \item the set of \emph{generalised equality functions}
  \[
   \cE := \left\{ f\in\allf \,\middle|\, \exists \ba\in\{0,1\}^{\ari(f)} \text{ s.t. } f(\bx)=0 \text{ for all } \bx\notin\{\ba, \bar{\ba}\} \right\},
  \]
  where $\bar{\ba}$ denotes the bitwise complement of $\ba$, and
 \item the set of \emph{generalised matching functions} $\cM := \left\{ f\in\allf \,\middle|\, f(\bx)=0 \text{ unless } \abs{\bx}\leq 1 \right\}$, where $\abs{\bx}$ is the Hamming weight of $\bx$.
\end{itemize}

\begin{obs}\label{obs:unary_holographic_inv}
 Let $M\in\GL$, then $M\circ\cU=\cU$.
\end{obs}

\begin{obs}\label{obs:unaries_in_families}
 $\cU\sse\cT$, $\cU\sse\cE$ and $\cU\sse\cM$.
\end{obs}

\begin{obs}\label{obs:constants_in_U_clone}
 $\allf_0\sse\hc{\cU}$ because any constant $c\in\allf_0$ can be represented as a contraction $\sum_y \dl_0(y) u_c(y)$, where $u_c=[c,0]\in\cU$.
\end{obs}

Note that $\tcl{\cU}$ is the set of all degenerate functions.

\begin{lem}\label{lem:degenerate_in_families}
 Suppose $M\in\GL$, then $\tcl{\cU}\sse\hc{M\circ\cT}$, $\tcl{\cU}\sse\hc{M\circ\cE}$, and $\tcl{\cU}\sse\hc{M\circ\cM}$.
\end{lem}
\begin{proof}
 By Observations~\ref{obs:unary_holographic_inv} and \ref{obs:unaries_in_families}, if $M\in\GL$, we have $\cU\sse M\circ\cT$, $\cU\sse M\circ\cE$ and $\cU\sse M\circ\cM$.
 Now, for any $\cF\sse\allf$ such that $\cU\sse\cF$, we have $\tcl{\cU}\sse\tcl{\cF}$ by the definition of $\tcl{-}$.
 Furthermore, by Lemma~\ref{lem:tensor_closure-holant_clone}, $\tcl{\cF}\sse\hc{\cF}$.
 Thus, $\tcl{\cU}\sse\hc{\cF}$ for $\cF\in\{M\circ\cT, M\circ\cE, M\circ\cM\}$, as desired.
\end{proof}

\begin{lem}\label{lem:M_triangular}
 Suppose $f$ is a symmetric ternary entangled function in $\cM$.
 Then we can write $f=[f_0,f_1,0,0]$ with $f_1\neq 0$.
 Now, if $U=\left(\begin{smallmatrix}a&b\\0&d\end{smallmatrix}\right)$ is an invertible upper triangular matrix over $\AA$, then $U\circ f\in\cM$.
 Conversely, if $M\in\GL$ and $M\circ f\in\cM$, then $M$ must be upper triangular.
\end{lem}
\begin{proof}
 The set $\cM$ is the set of functions that are non-zero only on inputs of Hamming weight at most 1, so any symmetric ternary function $f\in\cM$ can be written as $f=[f_0,f_1,0,0]$.
 If $f_1=0$, then $f(x,y,z)=f_0 \dl_0(x)\dl_0(y)\dl_0(z)$, so $f$ is degenerate.

 It is straightforward to check that, if $M=\left(\begin{smallmatrix}a&b\\c&d\end{smallmatrix}\right)$, then
 \begin{equation}\label{eq:M-circ-f}
  M\circ f = [(af_0+3bf_1)a^2,\; (acf_0+2bcf_1+adf_1)a,\; (acf_0+bcf_1+2adf_1)c,\; (cf_0+3df_1)c^2],
 \end{equation}
 a piece of code for doing so can be found in Appendix~\ref{a:M_triangular}.
 The first statement thus follows by letting $c=0$ to find $U\circ f = [(af_0+3bf_1)a^2, a^2df_1, 0, 0] \in\cM$.
 
 For the second statement, suppose $M$ is invertible and $M\circ f\in\cM$.
 This is equivalent to $(acf_0+bcf_1+2adf_1)c=0$ and $(cf_0+3df_1)c^2=0$.
 Assume, for a contradiction, that $c\neq 0$.
 Then the latter equality implies that $d = -cf_0/(3f_1)$.
 Plugging this into the former equality yields
 \[
  0 = c\left(acf_0 + bcf_1 -\frac{cf_0}{3f_1}2af_1\right) = \frac{c^2}{3} \left( af_0 + 3bf_1 \right),
 \]
 hence $b=-af_0/(3f_1)$.
 But then $ad-bc = -acf_0/(3f_1)+acf_0/(3f_1) = 0$, contradicting the assumption that $M$ is invertible.
 Thus, $M\circ f\in\cM$ implies $c=0$, i.e.\ $M$ is upper triangular.
\end{proof}

\begin{obs}\label{obs:E_bit-flip}
 Flipping bits does not affect whether a function is in $\cE$, i.e.\ $\left(\begin{smallmatrix}0&1\\1&0\end{smallmatrix}\right)\circ\cE=\cE$.
 Thus, by Observation~\ref{obs:Z_properties}, $K_1\circ\cE=K_2\circ\cE$.
\end{obs}

\begin{rem}
 Since $\EQ_3\in\cE$, by Corollary~\ref{cor:holant-functional}, $\hc{\cE}$ is a functional clone.
 In the literature, this functional clone is often called the set of \emph{product-type functions} and denoted $\cP$.
\end{rem}

\begin{lem}\label{lem:E_contraction}
 Suppose $f,g\in\cE$ and let $k:=\ari(f)$, $\ell:=\ari(g)$.
 Then
 \[
  \sum_{y\in\{0,1\}} f(\vc{x}{k-1},y) g(x_k\zd x_{k+\ell-2},y)
 \]
 is in $\cE\cup\allf_0$.
 Furthermore, if $k\geq 2$, then $\sum_{y\in\{0,1\}} f(\vc{x}{k-2},y,y)$ is also in $\cE\cup\allf_0$.
\end{lem}
\begin{proof}
 Since $f\in\cE$, there exists a bit string $\ba\in\{0,1\}^{k}$ such that $f(\bx)=0$ unless $\bx\in\{\ba,\bar{\ba}\}$.
 Similarly, since $g\in\cE$, there exists a bit string $\bb\in\{0,1\}^{\ell}$ such that $g(\by)=0$ unless $\by\in\{\bb,\bar{\bb}\}$.
 
 For simplicity, we will prove the second statement first.
 Let
 \[
  h(\vc{x}{k-2}) = \sum_{y\in\{0,1\}} f(\vc{x}{k-2},y,y).
 \]
 If $k=2$, then $\ari(h)=0$.
 Thus $h\in\allf_0$ and we are done.
 Otherwise, let $\bc$ be the bit string consisting of the first $(k-2)$ bits of $\ba$.
 Suppose $\bx\in\{0,1\}^{k-2}\setminus\{\bc,\bar{\bc}\}$.
 Then $h(\bx) = f(\bx,0,0)+f(\bx,1,1) = 0$.
 Therefore $h$ is nonzero on at most two complementary inputs and thus $h\in\cE$.
 
 Now, to prove the first statement, let
 \[
  h'(\vc{x}{k+\ell-2}) = \sum_{y\in\{0,1\}} f(\vc{x}{k-1},y) g(x_k\zd x_{k+\ell-2},y).
 \]
 If $k=\ell=1$, then $\ari(h')=0$, so $h\in\allf_0$ and we are done.
 Otherwise, define $\bc\in\{0,1\}^{k+\ell-2}$ as follows: if $a_k=b_\ell$, let $\bc$ be the bit string consisting of the first $(k-1)$ bits of $\ba$ followed by the first $(\ell-1)$ bits of $\bb$; otherwise let $\bc$ be the bit string consisting of the first $(k-1)$ bits of $\ba$ followed by the first $(\ell-1)$ bits of $\bar{\bb}$.
 It is straightforward to check that $h'(\bx)=0$ unless $\bx\in\{\bc,\bar{\bc}\}$. Therefore $h'\in\cE$.
\end{proof}

\begin{lem}\label{lem:M_contraction}
 Suppose $f,g\in\cM$ and let $k:=\ari(f)$, $\ell:=\ari(g)$.
 Then
 \[
  \sum_{y_1,y_2\in\{0,1\}} f(y_1,\vc{x}{k-1}) g(x_k\zd x_{k+\ell-2},y_2) \NEQ(y_1,y_2)
 \]
 is in $\cM\cup\allf_0$.
 Furthermore, if $k\geq 2$, then $\sum_{y_1,y_2\in\{0,1\}} f(\vc{x}{k-2},y_1,y_2) \NEQ(y_1,y_2)$ is also in $\cM\cup\allf_0$.
\end{lem}
\begin{proof}
 Since $f\in\cM$, $f(\bx)=0$ unless $\abs{\bx}\leq 1$.
 Similarly, since $g\in\cM$, $g(\bx)=0$ unless $\abs{\bx}\leq 1$.
 
 To prove the first statement, let
 \[
  h(\vc{x}{k+\ell-2}) = \sum_{y_1,y_2\in\{0,1\}} f(y_1,\vc{x}{k-1}) g(x_k\zd x_{k+\ell-2},y_2) \NEQ(y_1,y_2).
 \]
 If $k=\ell=1$, then $\ari(h)=0$, so $h\in\allf_0$ and we are done.
 Otherwise, suppose $\bx\in\{0,1\}^{k+\ell-2}$ satisfies $\abs{\bx}>1$.
 Then either $\sum_{j=1}^{k-1} x_j > 1$, or $\sum_{j=k}^{k+\ell-2} x_j > 1$, or  $\sum_{j=1}^{k-1} x_j = \sum_{j=k}^{k+\ell-2} x_j = 1$.
 In the first two cases, it is straightforward to see that $h(\bx)=0$.
 It remains to show that $h$ is zero in the third case.
 Now, $f(y_1,\vc{x}{k-1})$ is nonzero only if its input has Hamming weight at most 1, so $y_1$ must be 0 to get a non-zero value for $h$.
 Similarly, $g(x_k\zd x_{k+\ell-2},y_2)$ is nonzero only if its input has Hamming weight at most 1, so $y_2$ must be 0 to get a non-zero value for $h$.
 But if $y_1=y_2=0$, then $\NEQ(y_1,y_2)=0$.
 Therefore $h$ is zero on all inputs of Hamming weight greater than 1, so $h\in\cM$.
 
 To prove the second statement, let
 \[
  h'(\vc{x}{k-2}) = \sum_{y_1,y_2\in\{0,1\}} f(\vc{x}{k-2},y_1,y_2) \NEQ(y_1,y_2).
 \]
 If $k=2$, then $\ari(h')=0$, so $h'\in\allf_0$ and we are done.
 Otherwise, suppose $\bx\in\{0,1\}^{k-2}$ satisfies $\abs{\bx}>1$.
 Then $\abs{\bx}+y_1+y_2 >1$ for all $y_1,y_2\in\{0,1\}$, so $h'(\bx)=0$.
 Therefore $h'$ is zero on all inputs of Hamming weight greater than 1, so $h'\in\cM$.
\end{proof}

Recall from Lemma~\ref{lem:closure_free} that $\tcl{\cF}$ is the closure of $\cF$ under tensor product and permutation of arguments (but not contraction).

\begin{obs}\label{obs:scaling}
 Let $\cF$ be one of $\cT,\cE$ and $\cM$, let $c\in\AA$, and let $M\in\GL$.
 \begin{itemize}
  \item If $f\in\cF$, then $c\cdot f\in\cF$.
  \item If $f\in M\circ\cF$, then $c\cdot f\in M\circ\cF$.
  \item If $f\in\tcl{\cF}$, then $c\cdot f\in \tcl{\cF}$.
  \item If $f\in\tcl{M\circ\cF}$, then $c\cdot f\in\tcl{M\circ\cF}$.
 \end{itemize}
 In words, $\cT,\cE$ and $\cM$ contain arbitrary scalings of each of their elements, and this property is preserved under holographic transformation as well as tensor product and permutation of arguments.
\end{obs}

\begin{obs}\label{obs:E_holographic}
 For any $M\in\cO\cup\{K_1,K_2\}$, we have $\hc{M\circ\cE}=M\circ\hc{\cE}$: if $M\in\cO$, this is by Lemma~\ref{lem:hc_orthogonal_holographic}, and if $M\in\{K_1,K_2\}$, this is by Lemma~\ref{lem:hc_Z_equal}, noting that $\EQ_2,\NEQ\in\cE$.
\end{obs}

\begin{lem}\label{lem:closure_clone}
 Suppose $\cF$ is one of $\cT$, $M\circ\cE$ for some $M\in\cO\cup\{K_1,K_2\}$, or $K\circ\cM$ for some $K\in\{K_1,K_2\}$.
 Then $\tcl{\cF}\cup\allf_0=\hc{\cF}$.
\end{lem}
\begin{proof}
 We first show that $\tcl{\cF}\cup\allf_0\sse\hc{\cF}$.
 Indeed, $\tcl{\cG}\sse\hc{\cG}$ for all sets of functions $\cG$ by Proposition~\ref{prop:holant-clone-closure}.
 Furthermore, for all $\cF$ listed in the lemma, $\cU\sse\cF$ by Observations~\ref{obs:unaries_in_families} and \ref{obs:unary_holographic_inv}. Now, $\allf_0\sse\hc{\cU}$ by Observation~\ref{obs:constants_in_U_clone}, thus $\allf_0\sse\hc{\cF}$ and hence $\tcl{\cF}\cup\allf_0\sse\hc{\cF}$.
 
 For the other direction -- $\hc{\cF}\sse \tcl{\cF}\cup\allf_0$ -- it suffices to show that $\tcl{\cF}\cup\allf_0$ is closed under contraction.
 So assume $f\in\tcl{\cF}\cup\allf_0$ is a function of arity $k$, where $k\geq 2$, and let $h(\vc{x}{k-2}) = \sum_{y\in\{0,1\}} f(\vc{x}{k-2},y,y)$.
 If $k=2$, then $h\in\allf_0$, so we are done.
 Hence, from now on, we may assume $k>2$.
 
 By the definition of $\tcl{\cF}$, there exists a \ppsh-formula $\psi_f=\prod_{j=1}^s\varphi_j$ over $\cF$ in variables $\{\vc{v}{k}\}$ which does not involve any bound variables and which represents $f$.
 Let $\ell$ be the index of the atomic formula that involves the variable $v_k$ and let $m$ be the index of the atomic formula that involves the variable $v_{k-1}$.
 Define $\varphi_\ell'$ to be the atomic formula that is equal to $\varphi_\ell$ except that $v_k$ is replaced with $v_{k-1}$.
 For each $j\in [s]\setminus\{\ell\}$, set $\varphi_j'=\varphi_j$.
 Then $h$ is represented by the \ppsh-formula $\psi_h=\sum_{v_{k-1}}\prod_{j=1}^s\varphi_j'$.
 
 We now distinguish cases according to the different sets of functions.
 \begin{enumerate}
  \item Suppose $\cF=\cT$.
   Then each atomic formula $\varphi_j$ has arity 1 or 2.
   \begin{itemize}
    \item If $\ell=m$, then
     \begin{equation}\label{eq:h_ppsh_leqm}
      \psi_h = \left(\sum_{v_{k-1}}\varphi_\ell'\right)\prod_{j\in[s]\setminus\{\ell\}}\varphi_j',
     \end{equation}
     and the function $g$ represented by $\sum_{v_{k-1}}\varphi_\ell'$ is a constant.
     Thus $h$ is a scaling of an element of $\tcl{\cT}$, so $h\in\tcl{\cT}$ by Observation~\ref{obs:scaling}.
    \item If $\ell\neq m$, then
     \begin{equation}\label{eq:h_ppsh_lneqm}
      \psi_h = \left(\sum_{v_{k-1}}\varphi_\ell'\varphi_m'\right)\prod_{j\in[s]\setminus\{\ell,m\}}\varphi_j'.
     \end{equation}
     The function $g$ represented by $\sum_{v_{k-1}}\varphi_\ell'\varphi_m'$ has arity at most 2 since each atomic formula has arity at most 2 and one of the variables is bound.
     If $\ari(g)=0$, $h$ is a scaling of an element of $\tcl{\cT}$, thus $h\in\tcl{\cT}$ by Observation~\ref{obs:scaling}.
     Otherwise, $g\in\cT$.
     Thus we can replace $\sum_{v_{k-1}}\varphi_\ell'\varphi_m'$ with a single atomic formula $\theta$ whose function is $g$ and whose scope contains the free variables of $\varphi_\ell'$ and $\varphi_m'$.
     Then the \ppsh-formula $\theta\prod_{j\in[s]\setminus\{\ell,m\}}\varphi_j'$ has no bound variables and represents $h$, so $h\in\tcl{\cT}$ by definition.
   \end{itemize}
   In each subcase, we have shown that $h\in\tcl{\cT}$, thus $\tcl{\cT}\cup\allf_0$ is closed under contraction and hence $\hc{\cT}\sse \tcl{\cT}\cup\allf_0$.
  \item Suppose $\cF=M\circ\cE$ for some $M\in\cO\cup\{K_1,K_2\}$.
   First, consider the case $\cF=\cE$.
   \begin{itemize}
    \item If $\ell=m$, then $h$ can be expressed as in \eqref{eq:h_ppsh_leqm}.
     The atomic formula $\varphi_\ell'$ specifies a function $g_\ell\in\cE$ of arity at least 2.
     Hence $\sum_{v_{k-1}}\varphi_\ell'$ specifies a function $g\in\cE\cup\allf_0$ by Lemma~\ref{lem:E_contraction}.
    \item If $\ell\neq m$, then $h$ can be expressed as in \eqref{eq:h_ppsh_lneqm}.
     Assume without loss of generality that $v_{k-1}$ occurs in the last position in the scope of both $\varphi_\ell'$ and $\varphi_m'$; if this is not the case, permute the arguments appropriately.
     Again, the function $g$ represented by $\sum_{v_{k-1}}\varphi_\ell'\varphi_m'$ is in $\cE\cup\allf_0$ by Lemma~\ref{lem:E_contraction}.
   \end{itemize}
   The argument thus proceeds the same way in both subcases.
   If $g$ is a constant, then $h\in\tcl{\cE}$ by Observation~\ref{obs:scaling}.
   Otherwise, $g\in\cE$, which in turn implies $h\in\tcl{\cE}$ by definition.
   In each subcase, we have shown that $h\in\tcl{\cE}$, thus $\tcl{\cE}\cup\allf_0$ is closed under contraction and hence $\hc{\cE}\sse \tcl{\cE}\cup\allf_0$.
   
   We have shown that the desired result holds for $\cF=\cE$.
   Now, recall that $\hc{M\circ\cE}=M\circ\hc{\cE}$ for any $M\in\cO\cup\{K_1,K_2\}$ by Observation~\ref{obs:E_holographic}, that $\tcl{M\circ\cE}=M\circ\tcl{\cE}$ by Observation~\ref{obs:holographic_tensor}, and that $M\circ\allf_0=\allf_0$ by Definition~\ref{def:holographic_transformation}.
   Therefore,
   \[
    \hc{M\circ\cE}=M\circ\hc{\cE}\sse M\circ(\tcl{\cE}\cup\allf_0) = \tcl{M\circ\cE}\cup\allf_0,
   \]
   for all $M\in\cO\cup\{K_1,K_2\}$.
  \item Suppose $\cF=K\circ\cM$ for some $K\in\{K_1,K_2\}$.
   \begin{itemize}
    \item If $\ell=m$, then $h$ can be expressed as in \eqref{eq:h_ppsh_leqm}.
     The atomic formula $\varphi_\ell'$ specifies a function $K\circ g_\ell$, where $g_\ell\in\cM$ has arity $r\geq 2$.
     Let
     \[
      g(\vc{x}{r-2}):= \sum_{y_1,y_2\in\{0,1\}}g_\ell(\vc{x}{r-2},y_1,y_2)\NEQ(y_1,y_2),
     \]
     then by Lemma~\ref{lem:Z-contraction}, the \ppsh-formula $\sum_{v_{k-1}} \varphi_\ell'$ represents $K\circ g$.
     Furthermore, by Lemma~\ref{lem:M_contraction}, $g\in\cM\cup\allf_0$.
    \item If $\ell\neq m$, the function $h$ can be expressed as in \eqref{eq:h_ppsh_lneqm}.
     Assume without loss of generality that $v_{k-1}$ occurs in the first position in the scope of $\varphi_\ell'$ and in the last position in the scope of $\varphi_m'$; if this is not the case, permute the arguments appropriately.
     The atomic formula $\varphi_\ell'$ specifies a function $K\circ g_\ell$, where $g_\ell\in\cM$.
     Similarly, the atomic formula $\varphi_m'$ specifies a function $K\circ g_m$, where $g_m\in\cM$.
     Let
     \begin{align*}
     & g(\vc{x}{r_\ell+r_m-2}) := \\
    &  \sum_{y_1,y_2\in\{0,1\}} g_\ell(y_1,\vc{x}{r_\ell-1}) g_m(x_{r_\ell}\zd x_{r_\ell+r_m-2},y_2) \NEQ(y_1,y_2),
     \end{align*}
     then by Lemma~\ref{lem:Z-contraction}, the \ppsh-formula $\sum_{v_{k-1}} \varphi_\ell'\varphi_m'$ represents $K\circ g$.
     Furthermore, $g\in\cM\cup\allf_0$ by Lemma~\ref{lem:M_contraction}.
   \end{itemize}
   Now, the argument proceeds the same way in both subcases.
   If $g$ is a constant, then $h\in\tcl{K\circ\cM}$ by Observation~\ref{obs:scaling}.
   Otherwise, $g\in\cM$, which in turn implies $h\in\tcl{K\circ\cM}$ by definition.
   In each subcase, we have shown that $h\in\tcl{K\circ\cM}$, thus $\tcl{K\circ\cM}\cup\allf_0$ is closed under contraction and hence $\hc{K\circ\cM}\sse\tcl{K\circ\cM}\cup\allf_0$.
 \end{enumerate}
 We have shown for each set $\cF$ listed above that $\tcl{\cF}\cup\allf_0\sse\hc{\cF}$ and that $\hc{\cF}\sse\tcl{\cF}\cup\allf_0$.
 Thus we conclude that $\hc{\cF}=\tcl{\cF}\cup\allf_0$ for each of these sets $\cF$, as desired.
\end{proof}

The holant clones $\hc{M\circ\cE}$ for some $M\in\cO\cup\{K_1,K_2\}$ and $\hc{K\circ\cM}$ for some $K\in\{K_1,K_2\}$ have sets of generators that are much more concise than $M\circ\cE$ or $K\circ\cM$.
In the case of $\hc{M\circ\cE}$, one binary and one ternary function together with the set $\cU$ of unary functions are sufficient, see Lemma~\ref{lem:cE_generators} below.
For $\hc{K\circ\cM}$, one unary and one ternary function together with the set of holographically transformed binary generalised disequality functions are sufficient, this is shown in Lemma~\ref{lem:cM_generators}.
By binary generalised disequality functions, we mean all binary functions that are zero if their inputs are equal (though they need not be symmetric under interchange of the inputs).

\begin{lem}\label{lem:cE_generators}
 For any $M\in\cO\cup\{K_1,K_2\}$, we have $\hc{M\circ\cE}=\hc{M\circ\cU,M\circ\EQ_3,M\circ\NEQ}$.
\end{lem}
\begin{proof}
 First, recall that $\hc{M\circ\cE}=M\circ\hc{\cE}$ for any $M\in\cO\cup\{K_1,K_2\}$ by Observation~\ref{obs:E_holographic}.
 Furthermore,
 \[
  \hc{M\circ\cU,M\circ\EQ_3,M\circ\NEQ} = \hc{M\circ(\cU\cup\{\EQ_3,\NEQ\})},
 \]
 so if $M\in\cO$, then the holographic transformation can be pulled out of the holant clone by Lemma~\ref{lem:hc_orthogonal_holographic}.
 If instead $M\in\{K_1,K_2\}$, note that
 \begin{equation}\label{eq:EQ2_in_E}
  \EQ_2(x_1,x_2) = \sum_{y_1,y_2\in\{0,1\}} \EQ_3(x_1,x_2,y_1) \NEQ(y_1,y_2) \EQ_1(y_2),
 \end{equation}
 thus by Lemma~\ref{lem:Z-contraction}, $(M\circ\EQ_2)(x_1,x_2) = \sum_{y\in\{0,1\}}(M\circ\EQ_3)(x_1,x_2,y)(M\circ\EQ_1)(y)$.
 Hence $M\circ\EQ_2\in\hc{M\circ(\cU\cup\{\EQ_3,\NEQ\})}$ and thus
 \begin{align*}
  \hc{M\circ(\cU\cup\{\EQ_3,\NEQ\})}
  &= \hc{M\circ(\cU\cup\{\EQ_3,\EQ_2,\NEQ\})} \\
  &= M\circ\hc{\cU\cup\{\EQ_3,\EQ_2,\NEQ\}} \\
  &= M\circ\hc{\cU\cup\{\EQ_3,\NEQ\}},
 \end{align*}
 where the first step is by Corollary~\ref{cor:holant_closed}, the second step is by Lemma~\ref{lem:hc_Z_equal}, and the last step is again by Corollary~\ref{cor:holant_closed} since \eqref{eq:EQ2_in_E} also implies $\EQ_2\in\hc{\cU\cup\{\EQ_3,\NEQ\}}$.
 Therefore, for any $M\in\cO\cup\{K_1,K_2\}$, we have
 \[
  \hc{M\circ\cU,M\circ\EQ_3,M\circ\NEQ} = M\circ\hc{\cU,\EQ_3,\NEQ}.
 \]
 
 It thus suffices to prove $\hc{\cE}=\hc{\cU,\EQ_3,\NEQ}$; the desired result then follows since for any $M\in\cO\cup\{K_1,K_2\}$,
 \[
  \hc{M\circ\cE} = M\circ\hc{\cE} = M\circ\hc{\cU,\EQ_3,\NEQ} = \hc{M\circ\cU,M\circ\EQ_3,M\circ\NEQ}.
 \]
 Note that $\cU\cup\{\EQ_3,\NEQ\}\sse\cE$, so $\hc{\cU,\EQ_3,\NEQ}\sse\hc{\cE}$ by the definition of holant clones.
 
 It only remains to show that $\hc{\cE}\sse\hc{\cU,\EQ_3,\NEQ}$.
 
 Suppose $f\in\cE$ and let $n=\ari(f)$.
 If $n=1$, then $f\in\cU\sse\hc{\cU,\EQ_3,\NEQ}$, so we are done. 
 If $n>1$, there exists a bit string $\ba\in\{0,1\}^n$ such that $f(\bx)=0$ for all $\bx\notin\{\ba,\bar{\ba}\}$.
 We proceed by induction on the Hamming weight of $\ba$, or equivalently on the number of 1's in the bit string.
 
 For the base case, suppose $\abs{\ba}=0$, i.e.\ $\ba$ contains no 1's.
 Let $u_f:=[f(\ba),f(\bar{\ba})]\in\cU$.
 Then
 \[
  f(\vc{x}{n}) = \sum_{\vc{y}{n-1}\in\{0,1\}} \EQ_3(x_1,x_2,y_1) u_f(y_{n-1}) \prod_{k=1}^{n-2} \EQ_3(x_{k+2},y_{k},y_{k+1}),
 \]
 so $f\in\hc{\cU,\EQ_3,\NEQ}$.
 (Note that if $n=2$, the product at the end is empty, i.e.\ equal to 1.)
 
 For the induction step, assume there is some non-negative integer $m$ such that if $g\in\allf_n$ has the property $g(\bx)=0$ for all $\bx\in\{0,1\}^n\setminus\{\bb,\bar{\bb}\}$, where $\abs{\bb}=m$, then $g\in\hc{\cU,\EQ_3,\NEQ}$.
 Consider a function $f\in\cE$ which satisfies $f(\bx)=0$ unless $\bx\in\{\ba,\bar{\ba}\}$, where $\abs{\ba}=m+1$.
 Let $j\in [n]$ be some index such that $a_j=1$, and define
 \begin{equation}\label{eq:in_E_clone}
  f'(\vc{x}{n}) := \sum_{y\in\{0,1\}} \NEQ(x_j,y) f(\vc{x}{j-1},y,x_{j+1}\zd x_n).
 \end{equation}
 Then $f'$ is a function which satisfies $f'(\bx)=0$ unless $\bx\in\{\ba',\bar{\ba}'\}$, where $\ba'$ is the bit string that agrees with $\ba$ except in the $j$-th position.
 Since $a_j=1$, this implies that $\abs{\ba'}=m$.
 Thus $f'\in\hc{\cU,\EQ_3,\NEQ}$ by the induction hypothesis.
 But \eqref{eq:in_E_clone} is equivalent to
 \[
  f(\vc{x}{n})=\sum_{y\in\{0,1\}} \NEQ(x_j,y) f'(\vc{x}{j-1},y,x_{j+1}\zd x_n),
 \]
 so $f\in\hc{\cU,\EQ_3,\NEQ,f'}$.
 Hence, by Corollary~\ref{cor:holant_closed}, $f\in\hc{\cU,\EQ_3,\NEQ}$.
 
 We have shown that $f\in\cE$ implies $f\in\hc{\cU,\EQ_3,\NEQ}$, therefore $\cE\sse\hc{\cU,\EQ_3,\NEQ}$.
 By the definition of holant clones, this implies $\hc{\cE}\sse\hc{\hc{\cU,\EQ_3,\NEQ}}$.
 Now by Proposition~\ref{prop:holant-clone-closure},
 \[
  \hc{\hc{\cU,\EQ_3,\NEQ}}=\hc{\cU,\EQ_3,\NEQ}.
 \]
 Therefore, we have $\hc{\cE}\sse\hc{\cU,\EQ_3,\NEQ}$, which completes the proof.
\end{proof}

\begin{lem}\label{lem:cM_generators}
 Let $\cD:=\{f\in\allf_2\mid f(0,0)=f(1,1)=0\}$ and suppose $K\in\{K_1,K_2\}$.
 Then $\hc{K\circ\cM}=\hc{K\circ\cD,K\circ\ONE_3,K\circ\EQ_1}$.
\end{lem}
\begin{proof}
 Since $\cD\sse\cM$ and $\ONE_3,\EQ_1\in\cM$, we have $\hc{K\circ\cD,K\circ\ONE_3,K\circ\EQ_1}\sse\hc{K\circ\cM}$ by the definition of holant clones.
 
 Furthermore, by Proposition~\ref{prop:holant-clone-closure}, to prove $\hc{K\circ\cM}\sse\hc{K\circ\cD,K\circ\ONE_3,K\circ\EQ_1}$, it is enough to show $K\circ\cM\sse\hc{K\circ\cD,K\circ\ONE_3,K\circ\EQ_1}$.
 Now by Lemma~\ref{lem:hc_Z_holographic}, we have
 \[
  \hc{K\circ\cD,K\circ\ONE_3,K\circ\EQ_1} = \hc{K\circ(\cD\cup\{\ONE_3,\EQ_1\})} = K\circ\bhc{(\cD\cup\{\ONE_3,\EQ_1\})}{\{\NEQ\}}{L}.
 \]
 As the matrix $K$ is invertible, it therefore suffices to prove
 \[
  \cM \sse \bhc{(\cD\cup\{\ONE_3,\EQ_1\})}{\{\NEQ\}}{L}.
 \]
 
 Suppose $f\in\cM$ with $\ari(f)=n$.
 If $n=1$, let $d\in\cD$ be the function satisfying $d(0,1)=f(0)$ and $d(1,0)=f(1)$. Then
 \[
  f(x) = \sum_{y,z\in\{0,1\}} d(x,y)\NEQ(y,z)\EQ_1(z),
 \]
 which implies $f\in\bhc{(\cD\cup\{\ONE_3,\EQ_1\})}{\{\NEQ\}}{L}$, as desired.
 This completes the argument for $n=1$.
 We now make two digressions.
 \begin{itemize}
  \item Since Observation~\ref{obs:unaries_in_families} shows that $\cU\sse\cM$, the $n=1$ argument applies to an arbitrary unary function $f$.
   Therefore, $\cU\sse\bhc{(\cD\cup\{\ONE_3,\EQ_1\})}{\{\NEQ\}}{L}$.
  \item Before considering arbitrary higher-arity functions, first consider the functions $\ONE_n$ for positive integers $n$.
   We already know $\ONE_1\in\bhc{(\cD\cup\{\ONE_3,\EQ_1\})}{\{\NEQ\}}{L}$ since $\ONE_1\in\cU$.
   The function $\ONE_2$ is in $\cD$, which immediately implies $\ONE_2\in\bhc{(\cD\cup\{\ONE_3,\EQ_1\})}{\{\NEQ\}}{L}$.
   Furthermore, $\ONE_3\in\bhc{(\cD\cup\{\ONE_3,\EQ_1\})}{\{\NEQ\}}{L}$ by definition.
   Suppose, for the purpose of an argument by induction, that $\ONE_k\in\bhc{(\cD\cup\{\ONE_3,\EQ_1\})}{\{\NEQ\}}{L}$ for some integer $k\geq 3$.
   It is easy to check that
   \[
    \ONE_{k+1}(\vc{x}{k+1}) = \sum_{y,z\in\{0,1\}} \ONE_k(\vc{x}{k-1},y)\NEQ(y,z)\ONE_3(z,x_k,x_{k+1}).
   \]
   Using the induction hypothesis, this implies $\ONE_{k+1}\in\bhc{(\cD\cup\{\ONE_3,\EQ_1\})}{\{\NEQ\}}{L}$.
   Hence, by induction, $\ONE_k\in\bhc{(\cD\cup\{\ONE_3,\EQ_1\})}{\{\NEQ\}}{L}$ for any positive integer $k$.
 \end{itemize}
 We now go back to considering an arbitrary function $f\in\cM$ with $n:=\ari(f)>1$.
 For any $k\in[n]$, let $\be_k$ be the $n$-bit string which satisfies $(\be_k)_j=0$ if $j\neq k$ and $(\be_k)_k=1$.
 Define $d_0\in\cD$ to be the binary function satisfying $d_0(0,1)=1$ and $d_0(1,0)=f(0\zd 0)$.
 Furthermore, for any $k\in [n]$, define $d_k\in\cD$ to be the binary function satisfying $d_k(0,1)=1$ and $d_k(1,0)=f(\be_k)$.
 Then we claim
 \begin{multline}\label{eq:function_K_circ_cM}
  f(\vc{x}{n}) = \sum \ONE_{n+1}(\vc{z}{n+1}) \NEQ(z_{n+1},w_1) d_0(w_2,w_1) \NEQ(w_2,w_3) \EQ_1(w_3) \\ \prod_{j=1}^n d_j(x_j,y_j)\NEQ(y_j,z_j),
 \end{multline}
 where the sum is over $w_1,w_2,w_3,\vc{y}{n},\vc{z}{n+1}\in\{0,1\}$.
 To show that this equality is true, we distinguish several cases depending on the Hamming weight of $\bx=\vc{x}{n}$.
 \begin{itemize}
  \item Suppose $\abs{\bx}>1$, then the LHS is 0 because $f\in\cM$.
   On the RHS, note that for each $j\in[n]$ the product $d_j(x_j,y_j) \NEQ(y_j,z_j)$ is zero unless $x_j=z_j$.
   Thus, if $\abs{\bx}>1$, then all terms with a non-zero contribution to the sum must have $\abs{\bz}>1$.
   But for those terms, $\ONE_{n+1}(\bz)=0$.
   Hence the RHS is 0 and the equality holds.
  \item Suppose $\abs{\bx}=1$, i.e.\ $\bx=\be_\ell$ for some $\ell\in [n]$.
   Then
   \[
    \sum_{\vc{y}{n}\in\{0,1\}} \prod_{j=1}^n d_j((\be_\ell)_j,y_j)\NEQ(y_j,z_j) = f(\be_\ell) \prod_{j=1}^n \EQ_2((\be_\ell)_j,z_j),
   \]
   so the RHS of \eqref{eq:function_K_circ_cM} becomes
   \[
    \sum_{z_{n+1},w_1,w_2,w_3\in\{0,1\}} \ONE_{n+1}(\be_\ell,z_{n+1}) \NEQ(z_{n+1},w_1) d_0(w_2,w_1) \NEQ(w_2,w_3) \EQ_1(w_3) f(\be_\ell).
   \]
   Now the function $\ONE_{n+1}(\be_\ell,z_{n+1})$ is non-zero only if $z_{n+1}=0$, hence we can furthermore simplify this to
   \begin{multline*}
    \sum_{w_1,w_2,w_3\in\{0,1\}} \NEQ(0,w_1) d_0(w_2,w_1) \NEQ(w_2,w_3) \EQ_1(w_3) f(\be_\ell) \\
    = \sum_{w_2,w_3\in\{0,1\}} d_0(w_2,1) \NEQ(w_2,w_3) \EQ_1(w_3) f(\be_\ell) \\
    = \sum_{w_3\in\{0,1\}} \NEQ(0,w_3) \EQ_1(w_3) f(\be_\ell) = f(\be_\ell),
   \end{multline*}
   as desired.
  \item Suppose $\abs{\bx}=0$, i.e.\ $\bx$ is the all-zero bit string.
   Then, for all $j\in [n]$,
   \[
    \sum_{y_j\in\{0,1\}} d_j(0,y_j)\NEQ(y_j,z_j) = \EQ_2(0,z_j).
   \]
   Thus, the RHS of \eqref{eq:function_K_circ_cM} becomes
   \begin{multline*}
    \sum_{z_{n+1},w_1,w_2,w_3\in\{0,1\}} \ONE_{n+1}(0\zd 0, z_{n+1}) \NEQ(z_{n+1},w_1) d_0(w_2,w_1) \NEQ(w_2,w_3) \EQ_1(w_3) \\
    = \sum_{w_1,w_2,w_3\in\{0,1\}} \NEQ(1,w_1) d_0(w_2,w_1) \NEQ(w_2,w_3) \EQ_1(w_3) \\
    = \sum_{w_2,w_3\in\{0,1\}} d_0(w_2,0) \NEQ(w_2,w_3) \EQ_1(w_3) \\
    = \sum_{w_3\in\{0,1\}} f(0\zd 0) \NEQ(1,w_3) \EQ_1(w_3) = f(0\zd 0),
   \end{multline*}
   again as desired.
 \end{itemize}
 Hence \eqref{eq:function_K_circ_cM} holds for all $\bx\in\{0,1\}^n$.
 Thus, $f\in\bhc{(\cD\cup\{\ONE_3,\ONE_{n+1},\EQ_1\})}{\{\NEQ\}}{L}$.
 By Corollary~\ref{cor:bip_clone_closure}, since $\ONE_{n+1}\in\bhc{(\cD\cup\{\ONE_3,\EQ_1\})}{\{\NEQ\}}{L}$, we have $f\in\bhc{(\cD\cup\{\ONE_3,\EQ_1\})}{\{\NEQ\}}{L}$.
 As $f$ was an arbitrary function in $\cM$, this establishes $\cM\sse\bhc{(\cD\cup\{\ONE_3,\EQ_1\})}{\{\NEQ\}}{L}$, completing the proof.
\end{proof}

\subsection{Ternary functions and the entanglement classification}
\label{s:ternary}

Ternary functions play an important part in the holant dichotomies since $\hol(\cF)$ is computable exactly in polynomial time for any finite $\cF\sse\hc{\cT}$: i.e.\ any set of functions that are tensor products of unary and binary functions.

Recall that a function of arity at least 2 is entangled if it is not decomposable as a tensor product.
Thus, by Lemma~\ref{lem:closure_free}, a ternary function $f$ is entangled if and only if $f\notin\tcl{\cT}$.
By Lemma~\ref{lem:closure_clone}, this is equivalent to $f\notin\hc{\cT}$.

We say two vectors $\mathbf{f},\mathbf{f}'\in\AA^{2^3}$ are related by a \emph{generalised holographic transformation} if there exist matrices $A,B,C\in\GL$ such that $\mathbf{f}=(A\otimes B\otimes C)\mathbf{f}'$, where $\otimes$ denotes the Kronecker product of matrices.\footnote{In quantum information theory, this notion is called ``equivalence under stochastic local operations with classical communication'' (SLOCC).}
Entangled ternary functions can be partitioned into two classes: those that can be transformed to $\EQ_3$ by a generalised holographic transformation and those that can be transformed to $\ONE_3$ by a generalised holographic transformation.
This corresponds to the classification of entangled three-qubit states in quantum theory, which are represented by vectors in $\CC^{2^3}$ \cite{dur_three_2000}.
In quantum theory, the normalised vector corresponding to $\EQ_3$ is known as the GHZ state and the normalised vector corresponding to $\ONE_3$ is known as the $W$ state.

Formally, let $\bg$ be the vector corresponding to $\EQ_3$ and let $\bw$ be the vector corresponding to $\ONE_3$.
Suppose $f\in\allf_3\setminus\hc{\cT}$, then either there exist matrices $A,B,C\in\GL$ such that $\mathbf{f}=(A\otimes B\otimes C)\bg$, in which case $\mathbf{f}$ is said to be in the GHZ class, or there exist matrices $A,B,C\in\GL$ such that $\mathbf{f}=(A\otimes B\otimes C)\bw$, in which case $\mathbf{f}$ is said to be in the $W$ class.

If $f$ is symmetric, then ordinary (non-generalised) holographic transformations suffice, i.e.\ we may assume without loss of generality that $A=B=C$.

It is possible to determine from some polynomials in the values of $f$ whether a ternary function $f$ is in the GHZ class, the $W$ class, or decomposable.
The following proposition gives a special case of this result, which was derived in \cite{li_simple_2006} for arbitrary (i.e.\ not necessarily symmetric) vectors in $\CC^{2^3}$.

\begin{prop}\label{prop:entanglement_classification}
 Suppose $f\in\allf_3$ is symmetric, and write $f=[f_0,f_1,f_2,f_3]$.
 There exists a matrix $M\in\GL$ such that $f=M\circ\EQ_3$ if and only if
 \[
  (f_0f_3-f_1f_2)^2 - 4(f_1^2-f_0f_2)(f_2^2-f_1f_3) \neq 0.
 \]
 If this polynomial is zero, the following cases occur:
 \begin{itemize}
  \item If furthermore $f_1^2=f_0f_2$ and $f_2^2=f_1f_3$, then $f$ is degenerate.
  \item Otherwise, there exists a matrix $M\in\GL$ such that $f=M\circ\ONE_3$.
 \end{itemize}
\end{prop}

Note that the matrix $M$ in Proposition~\ref{prop:entanglement_classification} may not be unique (though this will not cause any issues).
For example, if $f=\EQ_3$, then
\[
 f=\begin{pmatrix}1&0\\0&1\end{pmatrix}\circ\EQ_3 \qquad\text{and}\qquad f=\begin{pmatrix}1&0\\0&e^{2i\pi/3}\end{pmatrix}\circ\EQ_3.
\]

\subsection{Known results that we will use}\label{sec:known}

The following results have been adapted to our notation.

\begin{thm}[{\cite[Theorem~2.2]{cai_dichotomy_2011}}]\label{thm:holant-star-dichotomy}
 Suppose $\cF$ is a finite subset of $\allf$.
 If
 \begin{enumerate}
  \item $\cF\sse\hc{\cT}$, or
  \item there exists $O\in\cO$ such that $\cF\sse\hc{O\circ \cE}$, or
  \item $\cF\sse\hc{K_1\circ\cE}=\hc{K_2\circ\cE}$, or
  \item there exists a matrix $K\in\{K_1,K_2\}$ such that $\cF\sse\hc{K\circ\cM}$,
 \end{enumerate}
 then, for any finite subset $S\sse\cU$, the problem $\hol(\cF,S)$ is polynomial-time computable.
 Otherwise, there exists a finite subset $S\sse\cU$ such that $\hol(\cF,S)$ is \numP-hard.
 The dichotomy is still true even if the inputs are restricted to planar graphs.
\end{thm}

By Lemma~\ref{lem:closure_clone}, if $\cG$ is one of $\cT$, $M\circ\cE$ with $M\in\cO\cup\{K_1,K_2\}$, or $K\circ\cM$ with $K\in\{K_1,K_2\}$, then $\tcl{\cG}\cup\allf_0=\hc{\cG}$.
Recall that nullary functions do not affect the complexity of a holant problem.
Thus, while the original statement of \cite[Theorem~2.2]{cai_dichotomy_2011} uses only closure under tensor product (with closure under permutations being included implicitly), we can use holant clones instead without affecting the result.
The same applies to the lemmas in this section.
Furthermore, by Observation~\ref{obs:E_bit-flip}, $\hc{K_1\circ\cE}=\hc{K_2\circ\cE}$, so it suffices to consider the case of a holographic transformation by $K_1$ despite both being treated separately in \cite{cai_dichotomy_2011}.

\begin{lem}[{\cite[Lemma~6.1]{cai_dichotomy_2011}}]\label{lem:arity_reduction}
 Let $\cR$ be any one of $\hc{\cT}$, $\hc{O\circ\cE}$ for some $O\in\cO$, $\hc{K_1\circ\cE}=\hc{K_2\circ\cE}$, or $\hc{K\circ\cM}$ for some $K\in\{K_1,K_2\}$.
 Let $r = 3$ if $\cR=\hc{\cT}$, and $r = 2$ in the other three cases.
 Suppose $f\in\allf\setminus\cR$ and $r<\ari(f)$.
 Then there exists $g\in\hc{f,\cU}$ such that $g\notin\cR$ and $r\leq\ari(g)<\ari(f)$.
\end{lem}

Note that, by repeated application, Lemma~\ref{lem:arity_reduction} 
implies that there is a 
$g\in\hc{f,\cU}$ such that $g\notin\cR$ and $ \ari(g)=r$.

\begin{lem}\label{lem:symmetric_ternary}
 Suppose $f\in\allf_3$ satisfies $f\notin\hc{\cT}$, and $s_1,s_2\in\allf_2$ satisfy $s_1\notin\hc{K_1\circ\cM}$ and $s_2\notin\hc{K_2\circ\cM}$.
 Then there is a symmetric ternary function $g\in\hc{f,s_1,s_2}$ such that $g\notin\hc{\cT}$.
\end{lem}
\begin{proof}
 This follows from \cite[Lemmas~16 and 17]{backens_new_2017}\footnote{The given lemma numbers are for the conference version of that paper. In the full version, these Lemmas are numbered 18 and 19, respectively.}.
 As the terminology and notation used in that paper is rather different from that used here, we give a translation:
 \begin{itemize}
  \item The phrase ``$n$-qubit state'' in the lemmas means a vector corresponding to a function of arity $n$; in particular, the vector $\ket{\psi}$ corresponds to our $\mathbf{f}$, $\ket{\GHZ}$ is the vector we denoted $\bg$ in Section~\ref{s:ternary}, $\ket{W}$ is the vector we denoted $\bw$, and $\ket{\phi}$ will be either $\bs_1$ or $\bs_2$.
  \item Up to an irrelevant scalar factor, the matrix $K$ in \cite[Lemma~17]{backens_new_2017} is the matrix we call $K_1$, and $KX$ in the lemma is the matrix we call $K_2$.
 \end{itemize}

 We are now ready to give the proof.
 The function $f\in\allf_3\setminus\hc{\cT}$ is ternary and entangled.
 By the entanglement classification described in Section~\ref{s:ternary}, it is in either the GHZ class or the $W$ class.
 
 \textbf{Case~1}: Suppose the function $f$ is in the GHZ class, i.e.\ there exist matrices $A,B,C\in\GL$ such that $\mathbf{f}=(A\otimes B\otimes C)\bg$, where $\bg$ is the vector corresponding to $\EQ_3$.
 In this case, the desired result follows from \cite[Lemma~16]{backens_new_2017}.
 The lemma states that either there exists a triangle gadget in $\hc{f}$ which represents a non-degenerate symmetric ternary function, or that $f$ is already symmetric.
 In the former case, the function associated with the triangle gadget is the desired $g$; in the latter case $g=f$.
 
 \textbf{Case~2}: Suppose the function $f$ is in the $W$ class, i.e.\ there exist matrices $A,B,C\in\GL$ such that $\mathbf{f}=(A\otimes B\otimes C)\bw$, where $\bw$ is the vector corresponding to $\ONE_3$.
 In this case, the desired result follows from \cite[Lemma~17]{backens_new_2017}.
 The lemma splits into three subcases.
 \begin{itemize}
  \item If $f\in K_1\circ\cM$, we also need a binary entangled function which is not in $K_1\circ\cM$: this is satisfied by taking $\ket{\phi}$ to be $\bs_1$.
  \item If $f\in K_2\circ\cM$, we also need a binary entangled function which is not in $K_2\circ\cM$: this is satisfied by taking $\ket{\phi}$ to be $\bs_2$.
  \item If $f\notin K_1\circ\cM\cup K_2\circ\cM$, the binary functions are not actually required.
 \end{itemize}
 The lemma then states that a symmetric ternary entangled function (which we will call $g$) can be ``constructed''.
 From the proof in the full version of the paper it can be seen this means that the function is contained in the holant clone generated by all the functions required in the lemma, i.e.\ $g\in\hc{f,s_1}$, $g\in\hc{f,s_2}$, or $g\in\hc{f}$, depending on the subcase.
 
 Now, $\hc{f},\hc{f,s_1},\hc{f,s_2}$ are all subsets of $\hc{f,s_1,s_2}$.
 Hence, in each case, $g\in\hc{f,s_1,s_2}$, as desired.
 Furthermore, $g$ is entangled by construction, so since $g$ has arity 3, $g\notin\hc{\cT}$.
\end{proof}

A similar but slightly weaker result was also proved in \cite[Lemma~7.3]{cai_dichotomy_2011}.

\begin{lem}[{\cite[Lemma 7.4]{cai_dichotomy_2011}}]\label{lem:binary_symmetric}
 Let $\cR$ be any one of $\hc{O\circ\cE}$ for some $O\in\cO$, $\hc{K_1\circ\cE}=\hc{K_2\circ\cE}$, or $\hc{K\circ\cM}$ for some $K\in\{K_1,K_2\}$.
 Suppose $f$ is a symmetric ternary function in $\cR\setminus\hc{\cT}$ and $g$ is a binary function in $\allf_2\setminus\cR$.
 Then there is a symmetric binary function $h\in\hc{f,g,\cU}$ such that $h\notin\cR$.
\end{lem}

\subsection{Universal quantum circuits as holant clones}\label{sec:univholant}

From Lemma~\ref{lem:closure_clone}, we have
\begin{itemize}
 \item $\hc{\cT}=\tcl{\cT}\cup\allf_0$,
 \item $\hc{M\circ\cE}=\tcl{M\circ\cE}\cup\allf_0$ for any $M\in\cO\cup\{K_1,K_2\}$, and
 \item $\hc{K\circ\cM}=\tcl{K\circ\cM}\cup\allf_0$ for any $K\in\{K_1,K_2\}$.
\end{itemize}
This makes it straightforward to see that these holant clones do not contain all functions.
For example,
\begin{itemize}
 \item if $f$ is any entangled function of arity at least 3, then $f\notin\hc{\cT}$,
 \item $\ONE_3\notin\hc{M\circ\cE}$ for any $M\in\cO\cup\{K_1,K_2\}$, and
 \item $\EQ_3\notin\hc{K\circ\cM}$ for any $K\in\{K_1,K_2\}$.
\end{itemize}

We now show that any set of functions which is not a subset of any of the tractable families of Theorem~\ref{thm:holant-star-dichotomy} generates a holant clone that is equal to $\allf$.
To do this, we make use of another result from quantum theory, which is known simply as ``single-qubit and CNOT gates are universal'' \cite[Section~4.5.2]{nielsen_quantum_2010}.
This statement leaves a lot of meaning implicit, it is more accurate to say
``quantum circuits consisting of single-qubit and CNOT gates are universal for unitary operators''.
We will now unpack the technical terms in this phrase and give a rigorous statement using the notation of restricted holant clones (cf.\ Section~\ref{s:restricted}).

For our purposes, a \emph{gate} or \emph{operator} can be thought of simply as an even-arity function; gates are implicitly assumed to be unitary.
We say a function $f\in\allf_{2n}$ is \emph{unitary} if\footnote{This definition corresponds to the matrix $M_f$ of values of $f$, whose columns are indexed by the first $n$ arguments and whose rows are indexed by the second $n$ arguments, being a unitary matrix.}
\[
 \sum_{\vc{z}{n}\in\{0,1\}} f^*(\vc{x}{n},\vc{z}{n})f(\vc{y}{n},\vc{z}{n}) = \prod_{j=1}^n \EQ_2(x_j,y_j),
\]
where $f^*$ is the function whose values are the complex conjugates of the values of $f$.
(Functions of odd arity cannot be unitary.)
A \emph{single-qubit gate} is a binary unitary function; we will denote the set of all such functions by
\[
 \cS := \{g\in\allf_2\mid g \text{ is unitary}\}.
\]
The \emph{CNOT gate} corresponds to the arity-4 function which satisfies
\[
 \CNOT(\bx) = \begin{cases} 1 &\text{if } \bx\in\{0000,0101,1011,1110\} \\ 0 &\text{otherwise.} \end{cases}
\]
It is straightforward to verify that $\CNOT$ is indeed unitary.

A \emph{quantum circuit} is the quantum version of a Boolean circuit.
In our terminology, it can be defined as a type of restricted holant gadget.

\begin{dfn}
 Let $\Ld=\{in,out\}$ be a set of labels and let $N=\{\{in,out\}\}$.
 Let $\cF$ be a set of unitary functions where, for each function, the first half of the arguments are labelled $in$ and the second half of the arguments are labelled $out$.
 Then a \ppsh-formula over $\cF$ restricted by $N$ is called a \emph{quantum circuit} if it corresponds to a holant gadget whose underlying graph is acyclic when made into a directed graph by orienting each edge from $out$ to $in$.
 For consistency with the definition of unitarity, the arguments of the function represented by this \ppsh-formula are required to be ordered so that the ones labelled $in$ come before the ones labelled $out$.
 We denote the set of all functions that can be represented by quantum circuits over $\cF$ by $\qc{\cF}$.
\end{dfn}

The expression \emph{universal for unitary operators} means that any unitary function $f$ is the effective function of some quantum circuit over $\cS\cup\{\CNOT\}$.
Using our terminology, we can thus state the result ``single-qubit and CNOT gates are universal'' \cite[Section~4.5.2]{nielsen_quantum_2010} as follows.

\begin{prop}[Universality of single-qubit and CNOT gates]\label{prop:universality}
 Suppose $f\in\allf$ is unitary, then $f\in\qc{\cS,\CNOT}$.
\end{prop}

To avoid the need to argue about quantum circuits, we will instead use the following weakening of Proposition~\ref{prop:universality}, which is sufficient for the subsequent proofs.

\begin{cor}\label{cor:universality}
 Suppose $f\in\allf$ is unitary, then $f\in\hc{\cS,\CNOT}$.
\end{cor}

For any $\ba\in\AA^n$, denote the usual complex Euclidean norm by $\norm{\ba} := (\sum_{j=1}^n \abs{a_j}^2)^{1/2}$.
We will use the following lemma from linear algebra.
The result is elementary, but we give a proof for completeness.

\begin{lem}\label{lem:unitary}
 Suppose $\ba\in\AA^k$.
 Let $\bb\in\AA^k$ be the vector that has $\norm{\ba}$ as its first component and zeroes everywhere else, i.e.\ $\bb=(\norm{\ba},0,0\zd0)$ with $(k-1)$ zeroes.
 Then there exists a unitary matrix $U\in\AA^{k\times k}$ such that $\ba=U\bb$.
\end{lem}
\begin{proof}
 First, suppose $\norm{\ba}=0$, then both $\ba$ and $\bb$ are the all-zero vector and $\ba=U\bb$ holds for any unitary matrix.
 Thus, from now on we may assume that $\ba$ is not a zero vector, i.e.\ $\norm{\ba}>0$.
 
 We will find a suitable unitary matrix $U$ by constructing an orthonormal basis for $\AA^k$ which contains the unit vector $\be_\ba:=(\norm{\ba})^{-1}\ba$.
 Any matrix whose columns form an orthonormal basis is unitary.
 If we assemble the basis vectors into a matrix in such a way that the first column is the vector $\be_\ba$, then $\ba=U\bb$ since
 \[
  (U\bb)_\ell = \sum_{j=1}^k U_{\ell j}b_j = \norm{\ba} U_{\ell 1} = \norm{\ba} (\be_\ba)_\ell = a_\ell.
 \]

 To find a suitable basis of $\AA^k$, let $B_1:=\{\ba\}$ and repeat the following step for each $j\in\{1,\ldots,k-1\}$: Find a vector $\bv_j\in\AA^k\setminus\operatorname{span}(B_j)$, where $\operatorname{span}(B_j)$ is the vector space spanned by the elements of $B_j$, and set $B_{j+1}:= B_j\cup\{\bv_j\}$.
 The set $B_k$ constructed in this way must be a basis for $\AA^k$.
 
 Next, construct an orthonormal basis $B_k'$ from $B_k$ via the Gram-Schmidt process with $\ba$ as the initial vector.
 This ensures that $B_k'$ contains the unit vector $\be_\ba$.
 Let $U$ be the matrix whose first column is $\be_\ba$ and whose subsequent columns correspond to the other elements of $B_k'$ in some arbitrary order.
 Then $\ba=U\bb$, as desired.
\end{proof}

By combining this lemma with the universality of single-qubit gates and CNOT, we find the following powerful result.

\begin{lem}\label{lem:allf_from_universality}
 $\allf\sse\hc{\cU,\cS,\CNOT}$.
\end{lem}
\begin{proof}
 First, note that, by Observation~\ref{obs:constants_in_U_clone}, $\allf_0\sse\hc{\cU}\sse\hc{\cU,\cS,\CNOT}$.
 
 Let $n$ be a positive integer, suppose $f\in\allf_n$, and let $f'(\vc{x}{n}) := \norm{\mathbf{f}} \prod_{k=1}^n \dl_0(x_k)$.
 Then $f'\in\hc{\cU,\cS,\CNOT}$ and $\mathbf{f}'=(\norm{\mathbf{f}},0,0\zd 0)$, with $(2^n-1)$ zeroes.
 Therefore, by Lemma~\ref{lem:unitary}, there exists a unitary matrix $U\in\AA^{2^n\times 2^n}$ such that $\mathbf{f}=U\mathbf{f}'$.
 Let $g_U\in\allf_{2n}$ be the function that satisfies $g_U(\vc{x}{n},\vc{y}{n}) = U_{\vc{x}{n},\vc{y}{n}}$ for all $\vc{x}{n},\vc{y}{n}\in\{0,1\}$.
 Then $g_U$ is a unitary function and
 \[
  f(\vc{x}{n}) = \sum_{\vc{y}{n}\in\{0,1\}} g_U(\vc{x}{n},\vc{y}{n})f'(\vc{y}{n}).
 \]
 But $g_U\in\hc{\cS,\CNOT}$ by Corollary~\ref{cor:universality}.
 Thus, by Corollary~\ref{cor:holant_closed}, $f\in\hc{\cU,\cS,\CNOT}$.
 Since $n$ and $f$ were arbitrary, and the case $\allf_0$ was considered separately, this implies $\allf\sse\hc{\cU,\cS,\CNOT}$.
\end{proof}

We will also use the following results.
The second of these is again elementary but we provide a proof for completeness.

\begin{lem}[{\cite[Lemma~24 in full version]{backens_new_2017}}]\label{lem:QR_decomposition}
 Let $M\in\GL$, then the following hold:
 \begin{itemize}
  \item There exists $Q\in\cO\cup\{K_1,K_2\}$ such that $Q^{-1}M$ is upper triangular.
  \item There exists $Q\in\cO\cup\{K_1,K_2\}$ such that $Q^{-1}M$ is lower triangular.
  \item If $Q^{-1}M$ is neither lower nor upper triangular for any orthogonal $Q$, then $M=K_1D$ or $M=K_2D$, where $D\in\GL$ is diagonal.
 \end{itemize}
\end{lem}

\begin{lem}\label{lem:PLDU}
 Any matrix $M\in\GL$ can be written as $M=PLDU$, where $P\in\{I,X\}$ is a permutation matrix, $L=\smm{1&0\\l&1}$ is lower triangular, $D=\smm{d&0\\0&e}$ is diagonal, and $U=\smm{1&u\\0&1}$ is upper triangular, with $l,d,e,u\in\AA$ and $d,e\neq 0$.
\end{lem}
\begin{proof}
 Suppose $M=\smm{\alpha&\beta\\\gamma&\dl}\in\GL$, then $\alpha\delta-\beta\gamma\neq 0$.
 We distinguish cases according to whether $\alpha$ is zero.
 \begin{itemize}
  \item Suppose $\alpha\neq 0$.
   Then division by $\alpha$ is well-defined and we have
   \[
    \pmm{\alpha&\beta\\\gamma&\dl} = \pmm{1&0\\0&1} \pmm{1&0\\\gamma/\alpha&1} \pmm{\alpha&0\\0&(\alpha\delta-\beta\gamma)/\alpha} \pmm{1&\beta/\alpha\\0&1}.
   \]
   Here, $\alpha\neq 0$ by the assumption of the case and $(\alpha\delta-\beta\gamma)/\alpha\neq 0$ by invertibility of $M$.
  \item Suppose $\alpha=0$.
   Then $\gamma\neq 0$ because $M$ is assumed to be invertible.
   Therefore, division by $\gamma$ is well defined and we have
   \[
    \pmm{\alpha&\beta\\\gamma&\dl} = \pmm{0&1\\1&0} \pmm{1&0\\\alpha/\gamma&1} \pmm{\gamma&0\\0&(\beta\gamma-\alpha\delta)/\gamma} \pmm{1&\delta/\gamma\\0&1}.
   \]
   The assumption of the case implies that $\gamma\neq 0$, and $(\beta\gamma-\alpha\delta)/\gamma\neq 0$ by invertibility of $M$.
 \end{itemize}
 Thus, in each case, we have found a decomposition of $M$ which satisfies the desired properties.
\end{proof}

\subsection{The power of conservative holant clones}
\label{s:power}

In the previous section, we laid out the existing results that we will build on in proving which sets $\cF$ generate holant clones that are equal to $\allf$.
We now move on to proving new results and finally the theorem about universality in the conservative case.

\begin{obs}\label{obs:scaling_general}
 Suppose $\cF\sse\allf$.
 By Observation~\ref{obs:constants_in_U_clone} and Corollary~\ref{cor:holant_closed}, if $f\in\hc{\cU,\cF}$, then $\ld\cdot f\in\hc{\cU,\cF}$ for any $\ld\in\AA$.
\end{obs}

\begin{lem}\label{lem:non-deg_binary}
 Let $\cI=\{f\in\allf_2\mid f(0,1)=f(1,0)=0\neq f(0,0) f(1,1)\}$ and suppose $t = [0,1,\mu]$ for some $\mu\in\Anz$.
 Then $\hc{\cI,\NEQ,t}$ contains all non-degenerate binary functions.
\end{lem}
\begin{proof}
 For any $d\in\Anz$, the function $k_d = [d,0,d^{-1}]$ is in $\cI$.
 Thus, the function
 \[
  \sum_{y_1,y_2,y_3\in\{0,1\}} k_d(x_1,y_1)t(y_1,y_2)k_d(y_2,y_3)\NEQ(y_3,x_2)
 \]
 is in $\hc{\cI,\NEQ,t}$ for any $d\in\Anz$; it is equal to
 \[
  \pmm{d&0\\0&d^{-1}} \pmm{0&1\\1&\mu} \pmm{d&0\\0&d^{-1}} \pmm{0&1\\1&0} = \pmm{1&0\\\mu/d^2&1}.
 \]
 Since $d\in\Anz$ is arbitrary, this means that any function of the form $\smm{1&0\\l&1}$ or, by permutation of the arguments, $\smm{1&u\\0&1}$, is contained in $\hc{\cI,\NEQ,t}$.
 
 The set $\cI$ contains all functions corresponding to invertible diagonal matrices, and we have just shown that $\hc{\cI,\NEQ,t}$ also contains all functions corresponding to upper triangular or lower triangular matrices with 1's on the diagonal.
 Finally, $\hc{\cI,\NEQ,t}$ contains $\NEQ$ which corresponds to the matrix $X$.
 Thus, by Proposition~\ref{prop:holant-clone-closure} and Lemma~\ref{lem:PLDU}, $\hc{\cI,\NEQ,t}$ contains any binary function corresponding to an invertible matrix: i.e.\ any non-degenerate binary function.
\end{proof}

\begin{lem}\label{lem:binaries_f_non_tractable}
 Suppose $f=R\circ\EQ_3$, where $R=\smm{a&b\\0&a^{-1}}\in\GL$ for some $a,b\in\AA\setminus\{0\}$.
 Then $\allf_2\sse\hc{\cU,f}$.
\end{lem}
\begin{proof}
 Fix $c\in\AA\setminus\{0\}$ and define $u_c=[c,c^{-1}]\in\cU$ and $u_c' := (R^T)^{-1}\circ u_c$.
 For any $d\in\AA\setminus\{0\}$, write $k_d$ for the binary function corresponding to $\smm{d&0\\0&d^{-1}}$.
 Furthermore, let $\rho$ be the binary function corresponding to the matrix $R$ and define $g_c(x_1,x_2) := \sum_{y\in\{0,1\}} f(x_1,x_2,y) u_c'(y)$.
 Then
 \begin{align*}
  g_c(x_1,x_2) &= \sum_{y,z_1,z_2,z_3\in\{0,1\}} \rho(x_1,z_1) \rho(x_2,z_2) \rho(y,z_3) \EQ_3(z_1,z_2,z_3) u_c'(y). \\
 \intertext{Now, $\sum_{y\in\{0,1\}} \rho(y,z_3) u_c'(y)$ corresponds to $R^T\circ u_c' = R^T(R^T)^{-1}\circ u_c = u_c$, so}
  g_c(x_1,x_2) &= \sum_{z_1,z_2,z_3\in\{0,1\}} \rho(x_1,z_1) \rho(x_2,z_2) \EQ_3(z_1,z_2,z_3) u_c(z_3) \\
  &= \sum_{z_1,z_2\in\{0,1\}} \rho(x_1,z_1) \rho(x_2,z_2) k_c(z_1,z_2).
 \end{align*}
 Thus, by turning the sum over binary functions into a matrix product, we find
 \begin{equation}\label{eq:code1}
  g_c = R\pmm{c&0\\0&c^{-1}}R^T = \pmm{a&b\\0&a^{-1}} \pmm{c&0\\0&c^{-1}} \pmm{a&0\\b&a^{-1}}
  = \frac{1}{a^2c} \pmm{a^2(a^2c^2+b^2)&ab\\ab&1}.
 \end{equation}
 Note that, for any $c\in\AA\setminus\{0\}$, the matrix corresponding to $g_c$ has determinant 1 since $\det R =\det R^T = 1$ and $\det\smm{c&0\\0&c^{-1}}=1$.
 The function $g_c(x_1,x_2)$ is in $\hc{\cU,f}$ for all $c\in\AA\setminus\{0\}$.
 
 We now prove that $\hc{\cU,f}$ contains
 \begin{description}
  \item[{\bf Part 1:}] the binary disequality function $\NEQ$,
  \item[{\bf Part 2:}] all binary functions corresponding to invertible diagonal matrices, and
  \item[{\bf Part 3:}] a symmetric binary function of the form $t_\mu=\smm{0&1\\1&\mu}$ for some $\mu\in\AA\setminus\{0\}$.
 \end{description}
 This is done by distinguishing three cases according to the values of $a$ and $b$.
 
 \noindent \textbf{Case~1}: Suppose $a^2b^2+1\neq 0$ and $2a^2b^2+1\neq 0$.
 
 {\bf Part 1 of Case~1:} For any $v,w\in\AA\setminus\{0\}$, define
 \begin{equation}\label{eq:def-h_vw}
  h_{v,w}(x_1,x_2) := \sum_{y_1,y_2\in\{0,1\}} g_v(x_1,y_1)g_w(y_1,y_2)g_v(y_2,x_2).
 \end{equation}
 Then $h_{v,w}\in\hc{\cU,f,g_v,g_w}=\hc{\cU,f}$ for all $v,w\in\AA\setminus\{0\}$ by Corollary~\ref{cor:holant_closed}.
 
 Now, $-i\cdot h_{s,t}=\NEQ$, where
 \[
  s = \frac{ib}{a}\left(\frac{2(a^2b^2+1)}{2a^2b^2+1}\right)^{1/2} \qquad\text{and}\qquad t = \frac{i(a^2b^2+1)}{a^3b}.
 \]
 This is straightforward to check using a computer algebra system, code for doing so can be found in Appendix~\ref{a:non-tractable-case1}.
 Hence, by Observation~\ref{obs:scaling_general} and Corollary~\ref{cor:holant_closed}, $\NEQ\in\hc{\cU,f}$.
 
 {\bf Part 2 of Case~1:} For any $p\in\Anz$ such that $(a^4p^2+a^2b^2+1)(a^2p^2+b^2)\neq 0$, define
 \[
  q_{\pm}(p) = \pm\sqrt{ \frac{-(a^2b^2+1)(a^4p^2+a^2b^2+1)}{a^6(a^2p^2+b^2)} }.
 \]
 With this expression for $q_\pm(p)$, the function $h_{p,q_\pm}$ is diagonal (cf.\ Appendix~\ref{a:non-tractable-case1}):
 \begin{equation}\label{eq:hpq}
  h_{p,q_{\pm}} = \pmm{ \mp\frac{a^{4} p^{2} + a^{2} b^{2} + 1}{a^{2} \sqrt{-\frac{{\left(a^{4} p^{2} + a^{2} b^{2} + 1\right)} {\left(a^{2} b^{2} + 1\right)}}{{\left(a^{2} p^{2} + b^{2}\right)} a^{6}}}} & 0 \\
  0 & \pm\frac{a^{2} b^{2} + 1}{{\left(a^{2} p^{2} + b^{2}\right)} a^{4} \sqrt{-\frac{{\left(a^{4} p^{2} + a^{2} b^{2} + 1\right)} {\left(a^{2} b^{2} + 1\right)}}{{\left(a^{2} p^{2} + b^{2}\right)} a^{6}}}} }
 \end{equation}
 Note that $h_{p,q_\pm}(0,0)h_{p,q_\pm}(1,1)=1$ since the matrix corresponding to $h_{p,q_\pm}$ is a product of three matrices with determinant 1.
 It is straightforward to see that both $q_\pm(p)$ and $h_{p,q_\pm}$ depend only on $p^2$.
 Thus, to show that $h_{p,q_\pm}=k_d$, it suffices to verify that
 \begin{equation}\label{eq:p-squared}
  p^2 \in P_d := \left\{ -\frac{2 a^{2} b^{2} + 1 + \sqrt{-4 {\left(a^{2} b^{2} + 1\right)} d^2 + 1}}{2 a^{4}}, -\frac{2 a^{2} b^{2} + 1 - \sqrt{-4 {\left(a^{2} b^{2} + 1\right)} d^2 + 1}}{2 a^{4}} \right\}
 \end{equation}
 is equivalent to
 \begin{equation}\label{eq:d-h_pq}
  d = h_{p,q_\pm}(0,0) = \mp\frac{a^{4} p^{2} + a^{2} b^{2} + 1}{a^{2} \sqrt{-\frac{{\left(a^{4} p^{2} + a^{2} b^{2} + 1\right)} {\left(a^{2} b^{2} + 1\right)}}{{\left(a^{2} p^{2} + b^{2}\right)} a^{6}}}},
 \end{equation}
 and that for any $d\in\Anz$ there exists an element of $P_d$ such that $p$ satisfies all the required conditions.
 This is done using a code snippet in Appendix~\ref{a:non-tractable-case1}.
 
 It remains to show that, for any $d\in\AA\setminus\{0\}$, we can choose $p^2\in P_d$ such that both $p$ and $q$ are well-defined and non-zero.
 Now, $p$ is always well-defined.
 Additionally, since $2a^2b^2+1\neq 0$ by assumption of Case~1, for any $d$ there is an element of $P_d$ which is non-zero.
 Furthermore, by \eqref{eq:d-h_pq} and $d\neq 0$, we have $\left(a^{4} p^{2} + a^{2} b^{2} + 1\right) \left(a^{2} p^{2} + b^{2}\right) \neq 0$, which implies that $q$ is well-defined and non-zero.
 Thus, for any $d\in\AA\setminus\{0\}$, $k_d\in\hc{\cU,f}$.
 Since any invertible diagonal matrix $\smm{\lambda&0\\0& \mu}$ can be written as $ {(\lambda \mu)}^{1/2}  \> k_{ d}$ for $d =  {(\lambda/\mu)}^{1/2}$, the fact that $k_d\in\hc{\cU,f}$ together with Observation~\ref{obs:scaling_general}  implies that any binary function that corresponds to an invertible diagonal matrix is in $\hc{\cU,f}$.
 
 {\bf Part 3 of Case~1:} Finally, $i\cdot g_{ib/a} = t_{1/(ab)}$, as can be verified using the code snippet in Appendix~\ref{a:non-tractable-case1}. 
 Hence  
 $t_{1/(ab)}\in\hc{\cU,f}$
 by Observation~\ref{obs:scaling_general}.
 
 Thus, in the three parts, we have shown that $\hc{\cU,f}$ contains $\NEQ$, all binary functions that correspond to invertible diagonal matrices, and the function $t_{1/(ab)}$.
 
 \textbf{Case~2}: Suppose $a^2b^2+1=0$.
 
 {\bf Part 2 of Case~2:} For any $d\in\AA\setminus\{0\}$, we have
 \begin{equation}\label{eq:case2-gadget2}
  k_d(x_1,x_2) = \sum_{y_1,y_2\in\{0,1\}} g_{1/a^2}(x_1,y_1)g_{-1/(a^2d)}(y_1,y_2)g_{1/a^2}(y_2,x_2).
 \end{equation}
 Code for verifying this equality can be found in Appendix~\ref{a:non-tractable-case2}.
 By \eqref{eq:case2-gadget2}, $k_d\in\hc{\cU,f}$.
 Hence, as in Case~1, by Observation~\ref{obs:scaling_general}, $\hc{\cU,f}$ contains all functions corresponding to invertible diagonal matrices.
   
 {\bf Parts 1 and 3 of Case~2:}
 For any $p,q\in\Anz$, define
 \[
  h'_{p,q}(x_1,x_2) :=
  \sum_{\vc{y}{4}\in\{0,1\}}g_p(x_1,y_1)g_q(y_1,y_2)g_{1/a^2}(y_2,y_3)g_q(y_3,y_4)g_p(y_4,x_2).
 \]
 To show that $\NEQ$ and  $t_\mu$ are in $\hc{\cU,f}$, we distinguish two subcases according to the value of~$b$.
 \begin{itemize}
  \item Suppose $b=i/a$.
   Then we have
   \begin{equation}\label{eq:case2-gadget1a}
    h'_{p,q} = \pmm{-2a^4p^2 + p^2 q^{-2} + 2 & -i \\ -i & 0}.
   \end{equation}
   This can be verified using the code snippet in Appendix~\ref{a:non-tractable-case2}.
   For any $a\in\AA\setminus\{0\}$, there exist non-zero values for $p,q$ such that $-2a^4p^2 + p^2 q^{-2} + 2 = 0$, which makes $h'_{p,q}$ a scaling of $\NEQ$.
   Thus, by Observation~\ref{obs:scaling_general}, $\NEQ\in\hc{\cU,f}$.
   
   Furthermore, $-i\cdot g_{1/a^2} = t_{-i}$, so $t_{-i}\in\hc{\cU,f}$.
   Again, code for verifying this can be found in Appendix~\ref{a:non-tractable-case2}.
  \item Suppose $b=-i/a$.
   Then we have
   \begin{equation}\label{eq:case2-gadget1b}
    h'_{p,q} = \pmm{-2a^4p^2 + p^2 q^{-2} + 2 & i \\ i & 0}.
   \end{equation}
   This can be verified using the code snippet in Appendix~\ref{a:non-tractable-case2}.
   For any $a\in\AA\setminus\{0\}$, there exist non-zero values for $p,q$ such that $-2a^4p^2 + p^2 q^{-2} + 2 = 0$, which makes $h'_{p,q}$ a scaling of $\NEQ$.
   Thus, by Observation~\ref{obs:scaling_general}, $\NEQ\in\hc{\cU,f}$.
   
   Furthermore, $i\cdot g_{1/a^2} = t_{i}$, so $t_{i}\in\hc{\cU,f}$.
   Again, code for verifying this can be found in Appendix~\ref{a:non-tractable-case2}.
 \end{itemize}
 
 We have shown that $\hc{\cU,f}$ contains $\NEQ$ and all binary functions that correspond to invertible diagonal matrices.
 This time, it also contains the functions $t_{-i}$ or $t_{i}$, depending on the subcase.

 \textbf{Case~3}: Suppose $2a^2b^2+1=0$.
 
  {\bf Part 2 of Case~3:}  For all $d\in\Anz$, if
 \[
  w = \frac{\pm\sqrt{-2 d^{2} + 1} - 1}{2 a^{2} d} \qquad\text{and}\qquad v = \frac{1}{a^2} \sqrt{\frac{2a^4w^2 - 1}{2(2a^4w^2 + 1)} },
 \]
 are well-defined and non-zero, we have
 \begin{equation}\label{eq:case3-gadget2}
  k_d(x_1,x_2) = \sum_{y_1,y_2\in\{0,1\}} g_{v}(x_1,y_1)g_w(y_1,y_2)g_{v}(y_2,x_2).
 \end{equation}
 This can be verified using the code snippets in Appendix~\ref{a:non-tractable-case3}.
 For any $d\in\Anz$, we can choose the sign in the definition of $w$ so that $w$ is non-zero.
 Yet if $2a^4w^2\in\{1,-1\}$, then $v$ is zero or not well-defined; i.e.\ the construction fails.
 The failure condition is equivalent to
 \[
  1 = 4a^8w^4 = \frac{(\pm\sqrt{-2 d^{2} + 1} - 1)^4}{4 d^4}.
 \]
 This condition is polynomial in $d$ and it is non-trivial.
 Hence there is a finite number of values of $d$ for which the construction fails.
 Thus, for any unsuitable $d$, there exists $e\in\Anz$ such that the construction succeeds for both $e$ and $d/e$.
 Then $k_{d}$ can be decomposed as $k_{d}(x_1,x_2)=\sum_{y\in\{0,1\}} k_{e}(x_1,y)k_{d/e}(y,x_2)$.
 Therefore, $k_d\in\hc{\cU,f}$ for all $d\in\Anz$, and by Observation~\ref{obs:scaling_general} all functions corresponding to diagonal matrices are contained in $\hc{\cU,f}$.
 
 {\bf Parts 1 and 3 of Case~3:}  For any $r\in\AA\setminus\{0\}$ such that both
 \begin{equation}
  s = \sqrt{ \frac{(2a^4r^2+1)(2a^4r^2-1)}{2a^4(4a^8r^4+1)} } \qquad\text{and}\qquad u = \frac{2a^4r^2-1}{\sqrt{2}a^2(2a^4r^2+1)}
 \end{equation}
 are well-defined and non-zero, let
 \[
  h''_r(x_1,x_2) := \sum_{\vc{y}{4}\in\{0,1\}}g_s(x_1,y_1)g_r(y_1,y_2)g_{u}(y_2,y_3)g_r(y_3,y_4)g_s(y_4,x_2).
 \]
 For all $a\in\AA\setminus\{0\}$ there exist values of $r$ which satisfy these requirements.
 
 Again, we distinguish two subcases according to the value of $b$.
 \begin{itemize}
  \item Suppose $b=i/(\sqrt{2}a)$.
   Then $i\cdot h''_r=\NEQ$ whenever $h''_r$ is well-defined, as can be verified using the code snippet in Appendix~\ref{a:non-tractable-case3}.
   As $h''_r\in\hc{\cU,f}$ for all such $r$, by Observation~\ref{obs:scaling_general}, this implies $\NEQ\in\hc{\cU,f}$.
   Additionally, $-i\cdot g_{1/(\sqrt{2}a^2)} = t_{-i\sqrt{2}}$, so $t_{-i\sqrt{2}}\in\hc{\cU,f}$.
   Again, this can be checked using a code snippet provided in Appendix~\ref{a:non-tractable-case3}.
  \item Suppose $b=-i/(\sqrt{2}a)$.
   Then $-i\cdot h''_r=\NEQ$ whenever $h''_r$ is well-defined, as can be verified using the code snippet in Appendix~\ref{a:non-tractable-case3}.
   As $h''_r\in\hc{\cU,f}$ for all such $r$, by Observation~\ref{obs:scaling_general}, this implies $\NEQ\in\hc{\cU,f}$.
   Additionally, $i\cdot g_{1/(\sqrt{2}a^2)} = t_{i\sqrt{2}}$, so $t_{i\sqrt{2}}\in\hc{\cU,f}$.
   Again, this can be checked using a code snippet provided in Appendix~\ref{a:non-tractable-case3}.
 \end{itemize}
 
 Thus, $\hc{\cU,f}$ contains $\NEQ$ and all binary functions that correspond to invertible diagonal matrices.
 It also contains $t_{\pm i\sqrt{2}}$, with the sign depending on the subcase.
 
 This completes the third case and thus the case distinction.
 
 In each case, we have within $\hc{\cU,f}$ the disequality function and all binary functions that correspond to invertible diagonal matrices, as well as a function $t_\mu=\smm{0&1\\1&\mu}$ for some $\mu\in\AA\setminus\{0\}$.
 Thus, by Lemma~\ref{lem:non-deg_binary} and Proposition~\ref{prop:holant-clone-closure}, $\hc{\cU,f}$ contains all non-degenerate binary functions.
 Now, the degenerate binary functions are all contained in $\hc{\cU}\sse\hc{\cU,f}$.
 Therefore all binary functions are contained in $\hc{\cU,f}$.
 In other words, $\allf_2\sse\hc{\cU,f}$, completing the proof.
\end{proof}

\begin{lem}\label{lem:binaries_f_tractable}
 Suppose $f=[1,0,0,a]$ and $g=[b,1,c]$ for some $a,b,c\in\AA$, where $a\neq 0$ and $g\notin\hc{\cE}$.
 Then $\allf_2\sse\hc{\cU,f,g}$.
\end{lem}
\begin{proof} 
 The property $g\notin\hc{\cE}$ implies that $b,c$ cannot both be zero and that $g$ is non-degenerate, i.e.\ $bc-1\neq 0$.
 (If $g$ was degenerate it would be in $C(\cU)$ and thus in $\hc{\cE}$ by Lemma~\ref{lem:degenerate_in_families}.)
  
 Fix $d\in\AA\setminus\{0\}$ and let $u_d := [1,d/a]\in\cU$.
 Define $f_d(x_1,x_2) := \sum_{y\in\{0,1\}} f(x_1,x_2,y) u_d(y)$, then $f_d=[1,0,d]$.
 But $d$ was an arbitrary non-zero number, so by Observation~\ref{obs:scaling_general}, all functions corresponding to invertible diagonal matrices are in $\hc{\cU,f,g}$.
 
 To realise $\NEQ$, suppose $s\in\AA\setminus\{0\}$ (to be determined later) satisfies $b^2+s\neq 0$ and $cs+b\neq 0$.
 Set
 \[
  t = \frac{-(b^2 + s)^2}{(cs + b)^2},
 \]
 and define
 \[
  h_s(x_1,x_2) := \sum_{\vc{y}{6}\in\{0,1\}} g(x_1,y_1) f_s(y_1,y_2) g(y_2,y_3) f_t(y_3,y_4) g(y_4,y_5) f_s(y_5,y_6) g(y_6,x_2).
 \]
 Then $h_s\in\hc{\cU,f,g}$ by Corollary~\ref{cor:holant_closed}.
 Furthermore, for all allowed $s$,
 \begin{equation}\label{eq:h_s}
  h_s = \pmm{ 0 & 
  \frac{-(b^2 + s) (bc - 1)^2 s}{cs + b} \\
  \frac{-(b^2 + s) (bc - 1)^2 s}{cs + b} &
  \frac{-(b^2c^2s + 2c^2s^2 + 2bcs + 2b^2 + s) (bc - 1)^2 s}{(cs + b)^2}
  }.
 \end{equation}
 This can be verified using the piece of code provided in Appendix~\ref{a:binaries_f_tractable}.
 
 We now distinguish cases according to the values of $b$ and $c$.
 \begin{itemize}
  \item If $b=0$, then $c\neq 0$.
   In this case, taking $s=-1/(2c^2)$ implies $b^2+s = -1/(2c^2) \neq 0$ and $cs+b = -1/(2c) \neq 0$, as required.
   Furthermore, it is straightforward to verify using a computer algebra system (cf.\ the code snippet in Appendix~\ref{a:binaries_f_tractable}) that $2c^3 \cdot h_{-1/(2c^2)} = \NEQ$.
   Thus $\NEQ\in\hc{\cU,f,g}$ by Observation~\ref{obs:scaling_general}.
  \item If $c=0$, then $b\neq 0$.
   In this case, taking $s=-2b^2$ implies $b^2+s = -b^2 \neq 0$ and $cs+b = b \neq 0$, as required.
   Furthermore, $(-2b^3)^{-1}\cdot h_{-2b^2} = \NEQ$ (cf.\ the code snippet in Appendix~\ref{a:binaries_f_tractable}), so $\NEQ\in\hc{\cU,f,g}$ by Observation~\ref{obs:scaling_general}.
  \item If $bc\neq 0$, recall that $bc\neq 1$ because $g$ is non-degenerate.
   Take
   \[
    s_{\pm} = -\frac{ (bc+1)^2 \pm (bc - 1)\sqrt{b^2c^2 + 6bc + 1} }{4c^2},
   \]
   then for any $b,c\in\Anz$ with $bc\neq 1$, there exists a choice of sign for which $s$ is non-zero.
   Furthermore, as $bc\neq 1$,
   \[
    0 = b^2+s_{\pm} = \frac{(3bc \mp \sqrt{b^2c^2 + 6bc + 1} + 1) (bc - 1)}{4 c^2}
   \]
   would imply
   \[
    (3bc+1)^2 = b^2c^2 + 6bc+ 1 \qquad\Longleftrightarrow\qquad 9b^2c^2 = b^2c^2,
   \]
   contradicting the assumption that $bc\neq 0$.
   Hence $b^2 + s_{\pm} \neq 0$.
   
   Additionally,
   \[
    0 = cs_{\pm}+b = -\frac{{\left(b c \pm \sqrt{b^{2} c^{2} + 6 b c + 1} - 1\right)} {\left(b c - 1\right)}}{4 c}
   \]
   would again imply the term in the first set of parentheses is zero.
   But then
   \[
    (bc-1)^2 = b^2c^2+6bc+1 \qquad\Longleftrightarrow\qquad -2bc = 6bc,
   \]
   again contradicting the assumption that $bc\neq0$.
   Hence $cs_{\pm} + b \neq 0$, and $s_{\pm}$ satisfies the required properties.
   
   Finally, as can be verified using the pieces of code in Appendix~\ref{a:binaries_f_tractable},
   \begin{equation}\label{eq:s-pm}
    \frac{{\left(b c \pm \sqrt{b^{2} c^{2} + 6 b c + 1} + 3\right)} c}{ {\left(b c \pm \sqrt{b^{2} c^{2} + 6 b c + 1} - 1\right)} {\left(b c - 1\right)}^{2} s_{\pm} }  h_{s_\pm} = \NEQ,
   \end{equation}
   so $\NEQ\in\hc{\cU,f,g}$ by Observation~\ref{obs:scaling_general}.
   (The scalar factor is non-zero and well-defined by the same argument that concluded $b^2+s_{\pm}\neq 0$ and $cs_{\pm}+b\neq 0$.)
   This completes the third case and thus the case distinction.
 \end{itemize}
 To summarise, in each case, we have shown that $\NEQ\in\hc{\cU,f,g}$.

 There are also values of $s$ for which $cs+b\neq 0$ and $b^2+s\neq 0$, and $h_s$ is not a scaling of $\NEQ$, i.e.\ $h_s$ is zero only on input bit string 00.
 For such $s$, the function $-\frac{-(cs + b)}{(b^2 + s) (bc - 1)^2 s}\cdot h_s$ takes the form $\smm{0&1\\1&\mu}$ for some $\mu\in\AA\setminus\{0\}$. 
 Since $\NEQ$ and all functions corresponding to invertible diagonal matrices are contained in $\hc{\cU,f,g}$, by Lemma~\ref{lem:non-deg_binary} and Proposition~\ref{prop:holant-clone-closure}, all non-degenerate binary functions are contained in $\hc{\cU,f,g}$.
 The degenerate binary functions are contained in $\hc{\cU}\sse\hc{\cU,f,g}$.
 Therefore $\allf_2\sse\hc{\cU,f,g}$.
\end{proof}

\begin{lem}\label{lem:min-arity}
 Let $\cR$ be any one of $\hc{\cT}$, $\hc{O\circ\cE}$ for some $O\in\cO$, $\hc{K_1\circ\cE}=\hc{K_2\circ\cE}$, or $\hc{K\circ\cM}$ for some $K\in\{K_1,K_2\}$.
 Let $r = 3$ if $\cR=\hc{\cT}$, and $r = 2$ in the other three cases. Suppose $f\in\allf\setminus\cR$. Then there exists $g\in\hc{f,\cU}$ such that $g\notin\cR$ and $\ari(g)=r$.
\end{lem}
\begin{proof}
 Firstly, note that $\hc{\cT}$ contains all unary and binary functions, so any $f\notin\hc{\cT}$ must have arity at least 3.
 Similarly, $\cU\sse\cE$ and holographic transformations by an invertible matrix map $\cU$ to itself by Observation~\ref{obs:unary_holographic_inv}.
 Thus $\cU\sse\hc{O\circ\cE}$ for any orthogonal $O$, and $\cU\sse\hc{K_1\circ\cE}$.
 Therefore any $f$ not in one of these sets must have arity at least 2.
 An analogous argument applies for $\hc{K\circ\cM}$ with $K\in\{K_1,K_2\}$.

 We have shown that, for each $\cR$, $f\in\allf\setminus\cR$ implies $\ari(f)\geq r$.
 If $\ari(f)=r$ then take $g=f$, which has all the desired properties.
 Otherwise let $f_0:=f$.
 While $\ari(f_k)>r$, let $f_{k+1}\in\hc{f_k,\cU}$ be the function identified when Lemma~\ref{lem:arity_reduction} is applied to $\cR$ and $f_k$.
 Then all $f_k$ satisfy $f_k\notin\cR$ by Lemma~\ref{lem:arity_reduction}.
 Furthermore, $f_k\in\hc{f,\cU}$ for all $k$: to see this, note that $f_0\in\hc{f,\cU}$ and that $f_j\in\hc{f,\cU}$ implies $f_{j+1}\in\hc{f_j,\cU}\sse\hc{f,f_j,\cU}=\hc{f,\cU}$ for all non-negative integers $j$.
 Finally, Lemma~\ref{lem:arity_reduction} ensures that $\ari(f_k)$ is a strictly decreasing function of $k$.
 Hence there is some $1\leq\ell<\ari(f)$ such that $\ari(f_\ell)=r$.
 Then $g=f_\ell$ has all the desired properties.
\end{proof}

\begin{lem}\label{lem:ghz_from_w}
 Suppose $M\in\GL$ and $f=M\circ\ONE_3$.
 Furthermore, suppose $s_1,s_2\in\allf_2$ satisfy $s_1\notin\hc{K_1\circ\cM}$ and $s_2\notin\hc{K_2\circ\cM}$.
 Then there exists a symmetric ternary function $g\in\hc{f,s_1,s_2}$ which satisfies $g=M'\circ\EQ_3$ for some $M'\in\GL$.
\end{lem}

This follows from the proof of Lemma~19 in the full version of \cite{backens_new_2017}, but the modifications are not obvious, so we nevertheless give a full proof here.

\begin{proof}
 We distinguish cases according to whether $f$ is in one of the tractable sets.
 
 \begin{figure}
  \centering
  \begin{tikzpicture}
   \path (3,0) node {(a)};
   \draw (5.25,1) -- ++(0,-0.5) -- ++(1,-1) -- ++(.4,-.4);
   \filldraw (5.25,0.5) circle (2pt) -- ++(-1,-1) -- ++(-.4,-.4);
   \filldraw (4.25,-.5) circle (2pt) -- ++(2,0) circle (2pt);
  \end{tikzpicture}
  \qquad\qquad\qquad\qquad
  \begin{tikzpicture}
   \path (3,0) node {(b)};
   \draw (5.25,1) -- ++(0,-0.5) -- ++(0.5,-0.5);
   \filldraw (5.75,0) circle (2pt) -- ++(0.5,-0.5) -- ++(.4,-.4);
   \filldraw (5.25,0.5) circle (2pt) -- ++(-0.5,-0.5) circle (2pt) -- ++(-0.5,-0.5) -- ++(-.4,-.4);
   \filldraw (4.25,-.5) circle (2pt) -- ++(1,0) circle (2pt) -- ++(1,0) circle (2pt);
  \end{tikzpicture}
  \caption{Two gadgets for realising symmetric ternary functions.}
  \label{fig:triangle_gadget}
 \end{figure}
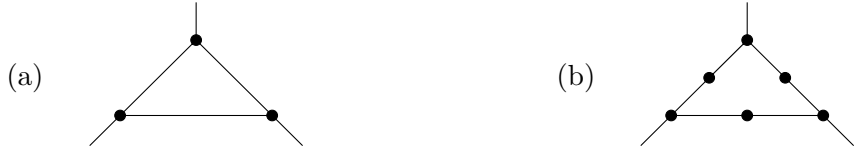
 
 \textbf{Case~1}: Suppose $f\notin K_1\circ\cM \cup K_2\circ\cM$.
 Consider the gadget in Figure~\ref{fig:triangle_gadget}a, where each vertex is assigned function $f$.
 This gadget realises a cyclically symmetric ternary function $g$, which -- since the inputs are bits -- is in fact fully symmetric.
 By definition,
 \[
  g(x_1,x_2,x_3) = \sum_{y_1,y_2,y_3\in\{0,1\}} f(x_1,y_2,y_3) f(x_2,y_3,y_1) f(x_3,y_1,y_2).
 \]
 But $f = M\circ\ONE_3$, which means $f(x,y,z)=\sum_{u,v,w\in\{0,1\}}M_{xu}M_{yv}M_{zw}\ONE_3(u,v,w)$.
 Plugging this into the equation for $g(x_1,x_2,x_3)$ yields
 \begin{multline*}
  g(x_1,x_2,x_3) = \sum M_{x_1a_1}M_{y_2a_2}M_{y_3a_3} \ONE_3(a_1,a_2,a_3) M_{x_2b_1}M_{y_3b_2}M_{y_1b_3} \ONE_3(b_1,b_2,b_3) \\ M_{x_3c_1}M_{y_1c_2}M_{y_2c_3} \ONE_3(c_1,c_2,c_3),
 \end{multline*}
 where the sum is over all $\by,\ba,\bb,\bc\in\{0,1\}^3$.
 Rearranging the terms, we find $g=M\circ g'$, where $g'(x_1,x_2,x_3)$ is equal to
  \begin{multline}  \label{eq:code2}
  \sum M_{y_1b_3} M_{y_1c_2} M_{y_2a_2} M_{y_2c_3} M_{y_3a_3} M_{y_3b_2} \ONE_3(x_1,a_2,a_3) \ONE_3(x_2,b_2,b_3) \ONE_3(x_3,c_2,c_3) \\  
  = \sum (M^TM)_{b_3c_2} (M^TM)_{c_3a_2} (M^TM)_{a_3b_2} \ONE_3(x_1,a_2,a_3) \ONE_3(x_2,b_2,b_3) \ONE_3(x_3,c_2,c_3),
   \end{multline}   
 with the sums being over all the variables that appear twice in the respective expressions.
 Now suppose $M^T M=\left(\begin{smallmatrix}a&b\\c&d\end{smallmatrix}\right)$, then a bit of algebra yields
 \begin{equation}\label{eq:symmetric_ternary}
  g' = [b^3+c^3+3abd+3acd, \; ab^2+abc+ac^2+a^2d, \; a^2b+a^2c, \; a^3].
 \end{equation}
 A code snippet for verifying this may be found in Appendix~\ref{a:ghz_from_w}.
 If there exists a matrix $M'\in\GL$ such that $g'=M'\circ\EQ_3$ then $g=(MM')\circ\EQ_3$.
 Thus it suffices to consider $g'$ going forward.
 By Proposition~\ref{prop:entanglement_classification} and some algebra, $g'$ is a holographically transformed equality function if and only if
 \begin{equation}\label{eq:GHZ-condition}
  (ad-bc)^3 a^6\neq 0.
 \end{equation}
 Again, Appendix~\ref{a:ghz_from_w} contains a piece of code for verifying this result.
 This inequality can fail to hold in two ways; we now show that each of these failures would contradict previous assumptions.
 \begin{itemize}
  \item The inequality \eqref{eq:GHZ-condition} is not satisfied if $ad-bc=0$.
   But $ad-bc=\det(M^TM)$, so this implies that $M^T M$ is not invertible, contradicting the assumption that $M$ is invertible.
   Therefore $ad-bc$ must be non-zero.
  \item The inequality \eqref{eq:GHZ-condition} is not satisfied if $a=0$.
   As $a=(M^TM)_{00}$, that property is equivalent to $M_{00}^2+M_{10}^2=0$, i.e.\ $M_{10}=\pm i M_{00}$.
   Now, if $M_{10}=i M_{00}$, then
   \[
    K_1^{-1} M = \frac{1}{\sqrt{2}} \begin{pmatrix}1&-i\\1&i\end{pmatrix} \begin{pmatrix}M_{00}&M_{01}\\i M_{00}&M_{11}\end{pmatrix} = \frac{1}{\sqrt{2}} \begin{pmatrix}2M_{00}&M_{01}-iM_{11}\\0&M_{01}+iM_{11}\end{pmatrix} =: U
   \]
   is upper triangular.
   Thus, by Lemma~\ref{lem:M_triangular}, $f=M\circ\ONE_3 = K_1\circ(U\circ\ONE_3) \in K_1\circ\cM$.
   By an analogous argument, if $M_{10}= -i M_{00}$, then $f\in K_2\circ\cM$.
   In each case, this contradicts the assumption that $f\notin K_1\circ\cM \cup K_2\circ\cM$.
   Therefore $a$ must be non-zero.
 \end{itemize}
 Thus, under the given assumptions, $(ad-bc)^3 a^6$ must be non-zero.
 Hence, by Proposition~\ref{prop:entanglement_classification}, there exists a matrix $M'\in\GL$ such that $g'=M'\circ\EQ_3$, which implies that $g=(MM')\circ\EQ_3$.

 \textbf{Case~2}: Suppose $f\in K_j\circ\cM$, where $j\in\{1,2\}$.
 Consider the gadget in Figure~\ref{fig:triangle_gadget}b where each degree-3 vertex is assigned the function $f$ and each degree-2 vertex is assigned the function $s_j$.
 This gadget realises a ternary symmetric function $g$ given by
 \[
  g(x_1,x_2,x_3) = \sum_{\by,\bz\in\{0,1\}^3} f(x_1,y_2,z_3) f(x_2,y_3,z_1) f(x_3,y_1,z_2) s_j(y_1,z_1) s_j(y_2,z_2) s_j(y_3,z_3).
 \]
 Let $S$ be the matrix associated with the binary function $s_j$.
 As in Case~1, we can substitute $f=M\circ\ONE_3$ into the expression for $g$ to find that $g=M\circ g'$, where $g'(x,y,z)$ is now equal to
 \[
  \sum (M^TSM)_{b_3c_2} (M^TSM)_{c_3a_2} (M^TSM)_{a_3b_2} \ONE_3(x,a_2,a_3) \ONE_3(y,b_2,b_3) \ONE_3(z,c_2,c_3).
 \]
 Suppose $M^T S M = \left(\begin{smallmatrix}a&b\\c&d\end{smallmatrix}\right)$, then $g'$ again takes the form \eqref{eq:symmetric_ternary}.
 Thus, as in Case~1, we need to show that $(ad-bc)^3 a^6\neq 0$.
 Again, there are two ways the inequality could fail to hold.
 \begin{itemize}
  \item The inequality is not satisfied if $ad-bc=0$.
   But $ad-bc=\det(M^TSM)$, so this implies that $M^T SM$ is not invertible.
   Since $s_j\notin\hc{K_j\circ\cM}$, $s_j$ is non-degenerate by Lemma~\ref{lem:degenerate_in_families}, and thus $S$ is invertible.
   The matrix $M$ is also invertible, a contradiction.
   Therefore $ad-bc$ must be non-zero.
  \item The inequality is not satisfied if $a=0$.
   Since $f=M\circ\ONE_3\in K_j\circ\cM$, we have $f':=(K_j^{-1}M)\circ\ONE_3\in\cM$.
   Now, both $f'$ and $\ONE_3$ are in $\cM$, hence by Lemma~\ref{lem:M_triangular}, $U:=K_j^{-1}M$ must be upper triangular.
   By Observation~\ref{obs:Z_properties}, $K_j^T = XK_j^{-1}$, so $M^TSM=U^TK_j^TSK_jU = U^TXS'XU$, where $S':=K_j^{-1}S(K_j^{-1})^{T}$ is the matrix corresponding to $K_j^{-1}\circ s_j$.
   Let $S''=XS'X$, then
   \[
    a = (U^TS''U)_{00} = \sum_{k,\ell\in\{0,1\}} U_{k0} S''_{k,\ell} U_{\ell 0} = U_{00}^2 S''_{00} + U_{00}U_{10} \left(S''_{01}+S''_{10}\right) + U_{10}^2 S''_{11} = U_{00}^2 S''_{00},
   \]
   since $U_{10}=0$ for the upper triangular matrix $U$.
   But $U_{00}\neq 0$ since $U$ is upper triangular and invertible.
   Furthermore, if $S'_{11}=(K_j^{-1}\circ s_j)(1,1)=0$ then $K_j^{-1}\circ s_j\in\cM$ and thus $s_j\in\hc{K_j\circ\cM}$.
   Therefore $s_j\notin\hc{K_j\circ\cM}$ implies that $S'_{11}\neq 0$ and thus that $S''_{00}\neq 0$.
   Hence $a$ must be non-zero.
 \end{itemize}
 Thus, under the given assumptions, $(ad-bc)^3 a^6$ must be non-zero.
 Hence, by Proposition~\ref{prop:entanglement_classification}, there exists a matrix $M'\in\GL$ such that $g'=M'\circ\EQ_3$, which implies that $g=(MM')\circ\EQ_3$.
 
 In each case, we have realised a function $g=(MM')\circ\EQ_3$, which we can bring into the form given in the lemma by re-defining $M'$ to absorb $M$.
 The function $g$ is realised by a gadget using $f$ and potentially $s_1$ or $s_2$, thus by Lemma~\ref{lem:holant_clone_gadget}, $g\in\hc{f,s_1,s_2}$.
\end{proof}

We are now ready to prove the theorem about conservative holant clones.
We say that a subset $\cF$ of $\allf$ is 
\emph{universal in the conservative case} if
$\hc{\cF \cup \cU} = \allf$. Our theorem is as follows.

\begin{thm}\label{thm:conservative_hc}
 Suppose $\cF$ is a subset of $\allf$. Then 
 $\cF$ is universal in the conservative case unless
 \begin{enumerate}
  \item $\cF\sse\hc{\cT}$, or
  \item there exists $O\in\cO$ such that $\cF\sse\hc{O\circ \cE}$, or
  \item $\cF\sse\hc{K_1\circ\cE}=\hc{K_2\circ\cE}$, or
  \item there exists a matrix $K\in\{K_1,K_2\}$ such that $\cF\sse\hc{K\circ\cM}$.
 \end{enumerate}
\end{thm}
\begin{proof}
Suppose that $\cF$ is a subset of $\allf$ that does not satisfy any of the four conditions.
By the definition of a holant clone,  it is straightforward to see that $\hc{\cU,\cF}\sse\allf$.
To prove the theorem, we will show
  $\allf\sse\hc{\cU,\cF}$.

 By Lemma~\ref{lem:min-arity}, there exists a ternary function $f\in\hc{\cF,\cU}$ such that $f\notin\hc{\cT}$, and binary functions $s_1,s_2\in\hc{\cF,\cU}$ such that $s_1\notin\hc{K_1\circ\cM}$ and $s_2\notin\hc{K_2\circ\cM}$.
 Furthermore, by Lemma~\ref{lem:symmetric_ternary}, there exists a symmetric ternary function $f'\in\hc{\cF,\cU}$ such that $f'\notin\hc{\cT}$.
 
 By the result stated in Section~\ref{s:ternary}, any symmetric ternary function $f'\notin\hc{\cT}$ satisfies either $f'=M\circ\EQ_3$ or $f'=M\circ\ONE_3$ for some matrix $M\in\GL$.
 Suppose the second property holds, i.e.\ there exists a matrix $M\in\GL$ such that $f'=M\circ\ONE_3$.
 Then by Lemma~\ref{lem:ghz_from_w}, there also exists $M'\in\GL$ such that $f''=M'\circ\EQ_3\in\hc{f,s_1,s_2}\sse\hc{\cF,\cU}$.
 Thus we may assume without loss of generality that $f'=M\circ\EQ_3$ for some $M\in\GL$ (if necessary, replacing $f'$ with the function $f''$ constructed using Lemma~\ref{lem:ghz_from_w}).
 
 We now distinguish cases according to whether or not $f'\in\hc{A\circ\cE}$ for any $A\in\cO\cup\{K_1,K_2\}$.
 
 \textbf{Case~1}: Suppose $f'=M\circ\EQ_3\notin\hc{A\circ\cE}$ for any $A\in\cO\cup\{K_1,K_2\}$.

 By Lemma~\ref{lem:QR_decomposition}, there exists $Q\in\cO\cup\{K_1,K_2\}$ such that $Q^{-1}M$ is upper triangular and $Q'\in\cO\cup\{K_1,K_2\}$ such that $(Q')^{-1}M$ is lower triangular.
 We distinguish subcases.
 \begin{itemize}
  \item If $Q\in\cO$, let $\cF'=Q^{-1}\circ\cF$ and let $f'':=Q^{-1}\circ f' = R\circ\EQ_3$, where $R:=Q^{-1}M$ is upper triangular.
  Then $Q\circ\hc{\cU,\cF',f''} = \hc{\cU,\cF,f'} \sse \hc{\cU,\cF,f}$.
  \item If $Q\notin\cO$ and $Q'\in\cO$, let $\cF'=X(Q')^{-1}\circ\cF$ and let
   \[
    f'' := X(Q')^{-1}\circ f' = X(Q')^{-1}M\circ\EQ_3 = X(Q')^{-1}MX\circ\EQ_3 = R\circ\EQ_3,
   \]
   where $(Q')^{-1}M$ is lower triangular and thus $R:=X(Q')^{-1}MX$ is upper triangular.
   Here, the third equality uses the fact that $X\circ\EQ_3=\EQ_3$.
   Then $Q'X\circ\hc{\cU,\cF',f''} = \hc{\cU,\cF,f'} \sse \hc{\cU,\cF,f}$.
  \item If $Q,Q'\notin\cO$ then $Q,Q'\in\{K_1,K_2\}$.
   By Lemma~\ref{lem:QR_decomposition}, this implies $M=K_1D$ or $M=K_2D$ for some diagonal matrix $D\in\GL$, so $f'\in\hc{K_1\circ\cE}=\hc{K_2\circ\cE}$. That contradicts the assumption of this case, hence this subcase cannot occur.
 \end{itemize}
 We have shown that there exists $O\in\cO$ such that $O\circ\hc{\cU,\cF',f''} \sse \hc{\cU,\cF,f} = \hc{\cU,\cF}$, where $\cF'=O^{-1}\circ\cF$ and $f''=O^{-1}\circ f'=R\circ\EQ_3$ for some upper triangular matrix $R\in\GL$.
 Let $\ld=\det R$ and define $f''':=\ld^{-3}(R\circ\EQ_3)=(\ld^{-1}R)\circ\EQ_3$; then $\hc{\cU,\cF',f''}=\hc{\cU,\cF',f'''}$ by Observation~\ref{obs:scaling_general}.
 
 Now, $\ld^{-1} R$ can be written as $\smm{a&b\\0&a^{-1}}$ for some $a,b\in\AA\setminus\{0\}$.
 Hence, by applying Lemma~\ref{lem:binaries_f_non_tractable} to $f'''$, we find
 $\allf_2\sse\hc{\cU,f'''}\sse\hc{\cU,\cF',f'''}$, which is equivalent to
 \[
 \allf_2 = O\circ\allf_2 \sse O\circ\hc{\cU,\cF',f'''} \sse \hc{\cU,\cF}.
 \]

 \textbf{Case~2}: Suppose $f'\in\hc{O\circ\cE}$ for some $O\in\cO$.
 Recall $f'\notin\hc{\cT}$; this implies $f'$ is entangled.
 By Lemma~\ref{lem:closure_clone}, any entangled function in $\hc{O\circ\cE}$ must itself be in $O\circ\cE$.
 Thus, there exists $\ld,a\in\Anz$ and $f'' := \ld[1,0,0,a]$ such that $f'=O\circ f''$.
 Let $\cF':=O^{-1}\circ\cF$, then $\hc{\cU,\cF}=O\circ\hc{\cU,\cF'}$ by Lemma~\ref{lem:hc_orthogonal_holographic}.
 Thus, by Observation~\ref{obs:scaling_general}, $f'\in\hc{\cU,\cF}$ implies $f'''=[1,0,0,a]\in\hc{\cU,\cF'}$.
 
 We assumed $\cF\nsubseteq\hc{O\circ\cE}$, thus by Lemma~\ref{lem:hc_orthogonal_holographic}, $\cF'\nsubseteq\hc{\cE}$.
 Hence by Lemma~\ref{lem:min-arity}, there exists a binary function $g\in\hc{\cU,\cF'}$ such that $g\notin\hc{\cE}$.
 Let $u_c=[1,c/a]\in\cU$ and define
 \[
  g'_c(x_1,x_2) = \sum_{y_1,y_2,y_3\in\{0,1\}} g(x_1,y_1) f'''(y_1,y_2,y_3) u_c(y_2) g(x_2,y_3).
 \]
 Then, writing $g_{xy}$ for $g(x,y)$, we have
 \[
  g'_c = \pmm{g_{00}^2+cg_{01}^2 & g_{00}g_{10}+cg_{01}g_{11} \\ g_{00}g_{10}+cg_{01}g_{11} & g_{10}^2 + cg_{11}^2 }.
 \]
 Now, $g\notin\hc{\cE}$ implies that $g_{00}g_{11}-g_{01}g_{10}\neq 0$ and that at most one of its values is 0: with two or more zero values, $g$ would either be degenerate (and thus in $\hc{\cE}$ by Lemma~\ref{lem:degenerate_in_families}), or it would itself be in $\cE$.
 Thus, none of the values of $g'_c$ are zero for all choices of $c$.
 As each value of $g'_c$ is a linear function of $c$ and $g'_c$ is symmetric, there are at most three choices for $c$ that make one of the values of $g'_c$ zero.
 Hence, there exists some $d\in\AA\setminus\{0\}$ such that $g'_d$ is everywhere non-zero.
 Let $g''$ be the scaling of $g'_d$ that satisfies $g''(0,1)=g''(1,0)=1$.
 Then by Lemma~\ref{lem:binaries_f_tractable}, $\allf_2 \sse \hc{\cU,f'',g''} \sse \hc{\cU,\cF'}$.
 Therefore,
 \[
  \allf_2 = O\circ\allf_2 \sse O\circ\hc{\cU,f'',g''} \sse O\circ\hc{\cU,\cF'} = \hc{\cU,\cF}.
 \]

 \textbf{Case~3}: Suppose $f'\in\hc{K_1\circ\cE}=\hc{K_2\circ\cE}$, i.e.\ $f'=\ld K_1\circ[1,0,0,a]$ for some $\ld,a\in\AA\setminus\{0\}$.
 Define $\cF'=K_1^{-1}\circ\cF$, then by Lemma~\ref{lem:hc_Z_holographic},
 \begin{equation}\label{eq:hc_K}
  \hc{\cU,\cF} = \hc{K_1\circ(\cU\cup\cF')} = K_1\circ\bhc{(\cU\cup\cF')}{\{\NEQ\}}{L}.
 \end{equation}
 Since we assumed $\cF\nsubseteq\hc{K_1\circ\cE}$, by Lemma~\ref{lem:min-arity} there exists a binary function $g\in\hc{\cU,\cF}$ such that $g\notin\hc{K_1\circ\cE}$.
 Let $g'=K_1^{-1}\circ g$, then $g'\notin\hc{\cE}$ by Lemma~\ref{lem:hc_Z_equal}.
 Furthermore, let $f'':= K_1^{-1}\circ f' = \ld[1,0,0,a]$, then $f'\in\hc{\cU,\cF}$ implies $f''\in\bhc{(\cU\cup\cF')}{\{\NEQ\}}{L}$ by \eqref{eq:hc_K}.
 Finally, let $u_c\in\cU$ be the function that satisfies $u_c = \lambda^{-1}[ca^{-1},1]$.
 Define 
 \[
  h_c(x_1,x_2) := \sum g'(x_1,y_1) \NEQ(y_1,y_2) f''(y_2,y_3,y_4) \NEQ(y_3,y_5) u_c(y_5) \NEQ(y_4,y_6) g'(x_2,y_6),
 \]
 where the sum is over $\vc{y}{6}\in\{0,1\}$.
 Then $h_c\in\bhc{(\cU\cup\cF')}{\{\NEQ\}}{L}$ and
 \[
  h_c = \pmm{(g'_{00})^2+c(g'_{01})^2 & g'_{00}g'_{10}+cg'_{01}g'_{11} \\ g'_{00}g'_{10}+cg'_{01}g'_{11} & (g'_{10})^2 + c(g'_{11})^2 },
 \]
 so, as in Case~2, we can pick $c$ such that $h_c$ is everywhere non-zero, which implies $h_c\notin\hc{\cE}$.
 
 Now, let $u_a\in\cU$ be the function that satisfies $u = \lambda^{-1}[a^{-1},1]$, then
 \[
  \sum_{y_1,y_2\in\{0,1\}} f'(x_1,x_2,y_1)\NEQ(y_1,y_2)u(y_2) = \EQ_2(x_1,x_2),
 \]
 so $\EQ_2\in\bhc{(\cU\cup\cF')}{\{\NEQ\}}{L}$ and thus, by \eqref{eq:hc_K}, $K_1\circ\EQ_2\in\hc{\cU,\cF}$.
 But then
 \[
  \sum_{y_1,y_2\in\{0,1\}}\EQ_2(x_1,y_1)\NEQ(y_1,y_2)\EQ_2(y_2,x_2)=\NEQ(x_1,x_2),
 \]
 so $\NEQ\in\bhc{(\cU\cup\cF')}{\{\NEQ\}}{L}$ and thus, by \eqref{eq:hc_K}, $K_1\circ\NEQ\in\hc{\cU,\cF}$.
 Therefore,
 \begin{multline}\label{eq:pull_out_K}
  \hc{\cU,\cF} = \hc{\cU,\cF,K_1\circ\EQ_2,K_1\circ\NEQ} = \hc{K_1\circ(\cU\cup\cF'\cup\{\EQ_2,\NEQ\})} \\
  = K_1\circ\hc{\cU,\cF',\EQ_2,\NEQ},
 \end{multline}
 where the first equality is by Corollary~\ref{cor:holant_closed}, the second is by Observation~\ref{obs:unary_holographic_inv} and by the definition of $\cF'$, and the third equality is by Lemma~\ref{lem:hc_Z_equal}.
 
 By combining \eqref{eq:pull_out_K} with \eqref{eq:hc_K}, we find that $\hc{\cU,\cF',\EQ_2,\NEQ}=\bhc{(\cU\cup\cF')}{\{\NEQ\}}{L}$.
 Hence, $f''\in\bhc{(\cU\cup\cF')}{\{\NEQ\}}{L}$ implies $f''\in\hc{\cU,\cF',\EQ_2,\NEQ}$ and, by scaling, $f''':=[1,0,0,a]\in\hc{\cU,\cF',\EQ_2,\NEQ}$.
 Furthermore, we have shown that there exists a symmetric binary function $h_c\in\hc{\cU,\cF',\EQ_2,\NEQ}$ such that $h_c$ is everywhere non-zero and $h_c\notin\hc{\cE}$.
 Let $h_c'(x_1,x_2)=(h_c(0,1))^{-1}h_c(x_1,x_2)$, then $h_c'$ is a symmetric binary function that is everywhere non-zero and satisfies $h_c'\notin\hc{\cE}$ as well as $h_c'(0,1)=h_c'(1,0)=1$.
 Additionally, $h_c'\in\hc{\cU,\cF',\EQ_2,\NEQ}$ by Observation~\ref{obs:scaling_general}.
 Thus, we can apply Lemma~\ref{lem:binaries_f_tractable} to $f'''$ and $h_c'$ to find
 \[
  \allf_2 \sse \hc{\cU,f''',h_c'} \sse \hc{\cU,\cF',\EQ_2,\NEQ}.
 \]
 Therefore,
 \[
  \allf_2 = K_1\circ\allf_2 \sse K_1\circ\hc{\cU,\cF',\EQ_2,\NEQ} = \hc{\cU,\cF},
 \]
 where the last equality uses~\eqref{eq:pull_out_K}.
 The third case is done.
 
 This completes the case distinction. 
 In each case, we have shown that $\allf_2\sse\hc{\cU,\cF}$, i.e.\ the conservative holant clone generated by $\cF$ contains all binary functions.
 
 Recall that $f'=M\circ\EQ_3\in\hc{\cU,\cF}$, where $M\in\GL$.
 Let $m\in\allf_2$ be the function corresponding to the matrix $M^{-1}$, then
 \[
  \sum_{y_1,y_2,y_2\in\{0,1\}} m(x_1,y_1) m(x_2,y_2) m(x_3,y_3) f'(y_1,y_2,y_3) = \EQ_3(x_1,x_2,x_3).
 \]
 Hence $\EQ_3\in\hc{f',\allf_2}\sse\hc{\cU,\cF}$.
 Similarly, let $h=\smm{1&1\\1&-1}$ and write $\EVEN_3 := [1,0,1,0]$, then $\EVEN_3\in\hc{\cU,\cF}$ since
 \[
  \EVEN_3(x_1,x_2,x_3) = \frac{1}{2}\sum_{y_1,y_2,y_2\in\{0,1\}} h(x_1,y_1) h(x_2,y_2) h(x_3,y_3) \EQ_3(y_1,y_2,y_3).
 \]
 Finally, $\CNOT(x_1,x_2,x_3,x_4) = \sum_{y\in\{0,1\}} \EQ_3(x_1,x_3,y)\EVEN_3(y,x_2,x_4)$, which implies that $\CNOT\in\hc{\cU,\cF}$.
 But then, by Lemma~\ref{lem:allf_from_universality} and Proposition~\ref{prop:holant-clone-closure},
 \[
  \allf \sse \hc{\cU,\cS,\CNOT} \sse \hc{\cU,\allf_2,\CNOT}\sse \hc{\cU,\cF,\allf_2,\CNOT} = \hc{\cU,\cF},
 \]
 completing the proof.
\end{proof}

\subsection{Approximating conservative holant problems}\label{sec:doapprox}

\begin{thm}\label{thm:main}
 Suppose that $\cF$ is a finite subset of $\allf$.
 Then there exists a finite subset $S\sse\cU$ such that both $\HA(\cF,S;\pi/3)$ and $\HN(\cF,S;1.01)$ are \numP-hard, unless
 \begin{enumerate}
  \item $\cF\sse\hc{\cT}$, or
  \item there exists $O\in\cO$ such that $\cF\sse\hc{O\circ \cE}$, or
  \item $\cF\sse\hc{K_1\circ\cE}=\hc{K_2\circ\cE}$, or
  \item there exists a matrix $K\in\{K_1,K_2\}$ such that $\cF\sse\hc{K\circ\cM}$.
 \end{enumerate}
\end{thm}

\begin{rem}
 The four conditions   listed in Theorem~\ref{thm:main} 
 are exactly the conditions from Theorem~\ref{thm:conservative_hc} 
 that prevent $\cF$ from being universal in the conservative case.
 They are also exactly the same as the conditions  in Theorem~\ref{thm:holant-star-dichotomy} by Cai, Lu and Xia.
 So, suppose that  $\cF$ is a finite subset of $\allf$.
 \begin{itemize}
  \item If $\cF$ satisfies the conditions
  then,  by 
  the theorem of Cai et al.,
  for any finite subset $S\sse\cU$, the problem $\hol(\cF,S)$ is polynomial-time computable.  
  \item Otherwise, Theorem~\ref{thm:main} guarantees that  there exists some finite subset $S\sse\cU$ such that even approximating the argument or the norm of the holant is \numP-hard for $\cF\cup S$.
 \end{itemize}
\end{rem}

\begin{proof}[Proof of Theorem~\ref{thm:main}.]
 We will show the desired hardness results by reduction from the problem of approximating the independent set polynomial.
 An independent set $I$ of a graph $G=(V,E)$ is a subset of the vertices such that no two vertices in $I$ are adjacent.
 Let $\ld\in\AA$, then the \emph{independent set polynomial with activity $\ld$} for a graph $G$ is $Z_G(\ld):=\sum_I \ld^{\abs{I}}$, where the sum is over all independent sets of $G$ \cite{bezakova_inapproximability_2017}.

 Define $u_\ld := [1,\ld]$ and recall that $\NAND=[1,1,0]$.
 Suppose the graph $G$ has maximum degree $\Dl\geq 3$, then $Z_G(\ld)$ is equal to the holant for a bipartite signature grid constructed via the following polynomial-time algorithm.
 
 Given $G=(V,E)$, define the bipartite graph $G'=(V',W',E')$ as follows:
 \begin{itemize}
  \item $V':=V$,
  \item $W':=E\cup\{v'\mid v\in V\}$, and
  \item $E':=\{\{v,e\}\mid e\in E, v\in e\}\cup\{\{v,v'\}\mid v\in V\}$.
 \end{itemize}
 In other words, $G'$ arises from $G$ by first subdividing each edge and then adding an additional vertex for each of the original vertices, connected to its ``parent'' by an edge.

 Let $\Omega = (G', \{\EQ_k\mid k\in[\Dl+1]\}\mid\{u_\ld,\NAND\}, \sigma)$, where $\sigma$ assigns the function of appropriate arity to each vertex (with vertices from $V'$ being assigned equality functions and vertices from $W'$ being assigned $u_\ld$ or $\NAND$).
 Then $Z_G(\ld)=Z_\Omega$, as desired.
 To see this, let $\deg(v)$ denote the degree of~$v$ in~$G'$ and note that
 \[
  Z_\Omega = \sum_{\bx:E'\to\{0,1\}} \left(\prod_{v\in V} \EQ_{\deg(v)}(\bx|_{E'(v)}) \; u_\ld(\bx|_{E'(v')})\right) \prod_{w\in E} \NAND(\bx|_{E'(w)}).
 \]
 Now, any independent set $I\sse V$ of $G$ contributes a summand $\ld^{\abs{I}}$ to $Z_\Omega$ via the assignment $\bx_I:E'\to\{0,1\}$ which satisfies $\bx_I(e)=1$ if and only if there exists $v\in I$ such that $v\in e$.
 Furthermore, any assignment $\bx$ making a non-zero contribution to $Z_\Omega$ must correspond to an independent set in this way:
 \begin{itemize}
  \item The equality functions ensure that $\bx(e)=\bx(e')$ whenever there exists $v\in V$ such that $v\in e\cap e'$.
   Thus $\bx$ induces a well-defined function $\by:V\to\{0,1\}$, such that $\by(v)$ equals the value $\bx$ takes on edges incident on $v$.
   Hence we can define $S_\bx := \{v\in V\mid \by(v)=1\}$.
  \item The functions $\NAND$ ensure that if $v\in S_\bx$ and $\{v,w\}\in E$, then $w\notin S_\bx$.
   Therefore $S_\bx$ is an independent set of $G$.
  \item The functions $u_\ld$ yield a factor of $\ld$ for each vertex in the independent set $S_\bx$.
 \end{itemize}
 Thus, $\HN(\{\EQ_k\mid k\in[\Dl+1]\}\mid\{u_\ld,\NAND\}; 1.01)$ and $\HA(\{\EQ_k\mid k\in[\Dl+1]\}\mid\{u_\ld,\NAND\}; \pi/3)$ are \numP-hard if $\ld$ is an activity for which the norm and argument of the independent set polynomial on graphs of maximum degree $\Dl$ are \numP-hard to approximate to the given accuracies.
 Such values exist in $\AA$, in fact $\ld$ can be negative real \cite[Theorems 1 and 2]{bezakova_inapproximability_2017}.
 Furthermore, since $\cF\mid\cG\leq \cF\cup\cG$ by Lemma~\ref{lem:preorder_properties}\eqref{it:2from_bipartite}, $\HN(\{\EQ_k\mid k\in[\Dl+1]\}\cup\{u_\ld,\NAND\}; 1.01)$ and $\HA(\{\EQ_k\mid k\in[\Dl+1]\}\cup\{u_\ld,\NAND\}; \pi/3)$ are \numP-hard for such $\ld$.
 
 Let $\Dl=3$ (for simplicity) and fix an activity $\ld\in\AA$ for which the independent set polynomial is hard to approximate on graphs of maximum degree 3 by \cite[Theorems 1 and 2]{bezakova_inapproximability_2017}.
 Via the above algorithm, such a choice of $\ld$ implies that
 \begin{gather*}
  \HN(\{\vc{\EQ}{4},u_\ld,\NAND\}; 1.01) \qquad\text{and} \\
  \HA(\{\vc{\EQ}{4},u_\ld,\NAND\}; \pi/3)
 \end{gather*}
 are both \numP-hard.
 We now show that for any set $\cF$ which is not one of the tractable cases, there exists a finite set $S\sse\cU$ such that
 \begin{equation}\label{eq:IS_holant}
  \{\vc{\EQ}{4},u_\ld,\NAND\} \leq \cF\cup S.
 \end{equation}
 Indeed, suppose $\cF$ is not one of the sets listed in the theorem.
 Then $\hc{\cU,\cF}=\allf$ by Theorem~\ref{thm:conservative_hc}, so in particular $\{\vc{\EQ}{4},u_\ld,\NAND\}\sse\hc{\cU,\cF}$.
 By the definition of holant clones, each of the six functions on the LHS is represented by some \ppsh-formula over $\cU\cup\cF$.
 Each of those \ppsh-formulas has finite size, hence it uses finitely many unary functions.
 Thus there exists some finite set $S\sse\cU$ such that $\{\vc{\EQ}{4},u_\ld,\NAND\}\sse\hc{\cF,S}$.
 This implies \eqref{eq:IS_holant} by Lemma~\ref{lem:preorder_properties}\eqref{it:from_clone}.
 Now, by Observation~\ref{obs:preorder_reductions}, \eqref{eq:IS_holant} in turn implies
 \begin{align*}
  \HN(\{\vc{\EQ}{4},u_\ld,\NAND\};1.01) &\leq_{PT} \HN(\cF\cup S; 1.01) \\
  \HA(\{\vc{\EQ}{4},u_\ld,\NAND\}; \pi/3) &\leq_{PT} \HA(\cF\cup S; \pi/3).
 \end{align*}
 But $\ld$ was specifically chosen so the problems on the left-hand side are \numP-hard.
 Therefore, $\HN(\cF\cup S; 1.01)$ and $\HA(\cF\cup S; \pi/3)$ are \numP-hard, as desired.
\end{proof}

\appendix

\section{Sage code}
\label{a:sage_code}

The (linear) algebra in some of the more complicated lemmas was handled using the SageMath software system \cite{sagemath}.
In the following, we give code that can be used to verify our derivations.
Code snippets can be run online at \url{http://sagecell.sagemath.org/}.

\subsection{Code for Lemma~\ref{lem:M_triangular}}
\label{a:M_triangular}

The function $f$ (and hence $M\circ f$) is symmetric by assumption.
Thus the following piece of code verifies \eqref{eq:M-circ-f}, where \texttt{mcf} denotes the function $M\circ f$.
The code outputs \texttt{True}.
\begin{verbatim}
var("a, b, c, d, f0, f1")
import itertools
def f(x,y,z):
    if x+y+z == 0:
        return f0
    if x+y+z == 1:
        return f1
    return 0
def m(x,y):
    return matrix([[a, b], [c, d]])[x,y]
def mcf(x1,x2,x3):
    output = 0
    for z1, z2, z3 in itertools.product([0,1], repeat=3):
        output = output + m(x1,z1) * m(x2,z2) * m(x3,z3) * f(z1,z2,z3)
    return output
[mcf(0,0,0), mcf(0,0,1), mcf(0,1,1), mcf(1,1,1)] == [(a*f0 + 3*b*f1)*a^2, \
(a*c*f0 + 2*b*c*f1 + a*d*f1)*a, (a*c*f0 + b*c*f1 + 2*a*d*f1)*c, \
(c*f0 + 3*d*f1)*c^2]
\end{verbatim}

\subsection{Code for Lemma~\ref{lem:binaries_f_non_tractable}, Case 1}
\label{a:non-tractable-case1}

The definition of $h_{v,w}$ in \eqref{eq:def-h_vw} is just a multiplication of three $2\times 2$ matrices.
Therefore the value of $-i\cdot h_{s,t}$ can be verified by the following piece of code, which outputs \texttt{True}.
The definition of $g$ in this piece of code (and in subsequent ones) comes from~\eqref{eq:code1}.

\begin{verbatim}
var("a, b")
def g(c):
    return 1/(a^2*c) * matrix([[a^2*(a^2*c^2+b^2), a*b], [a*b, 1]])
s = (I*b/a) * ( 2*(a^2*b^2 + 1) / (2*a^2*b^2+1) )^(1/2)
t = I * (a^2*b^2 + 1) / (a^3*b)
gs = g(s)
gt = g(t)
-I * (gs*gt*gs).factor() == matrix([[0, 1], [1, 0]])
\end{verbatim}

For \eqref{eq:hpq}, there are two cases to consider, depending on the choice of sign in $q_\pm$.
 
1. The following piece of code verifies \eqref{eq:hpq} for the plus case; it outputs \texttt{True}.
\begin{verbatim}
var("a, b, p")
def g(c):
    return 1/(a^2*c) * matrix([[a^2*(a^2*c^2+b^2), a*b], [a*b, 1]])
q = sqrt( -(a^2*b^2+1)*(a^4*p^2 + a^2*b^2 +1) / (a^6*(a^2*p^2 + b^2)) )
gp = g(p)
gq = g(q)
d1 = - (a^4*p^2 + a^2*b^2 + 1) / ( a^2 * sqrt( - (a^4*p^2 + a^2*b^2 + 1) \
     * (a^2*b^2 + 1) / ( (a^2*p^2 + b^2) * a^6 ) ) )
d2 = (a^2*b^2 + 1) / ( (a^2*p^2 + b^2) * a^4 * \
     sqrt( - (a^4*p^2 + a^2*b^2 + 1) * (a^2*b^2 + 1) / \
     ( (a^2*p^2 + b^2) * a^6 ) ) )
(gp*gq*gp).factor() == matrix([[d1, 0], [0, d2]])
\end{verbatim}

2. The following piece of code verifies \eqref{eq:hpq} for the minus case; it outputs \texttt{True}.
\begin{verbatim}
var("a, b, p")
def g(c):
    return 1/(a^2*c) * matrix([[a^2*(a^2*c^2+b^2), a*b], [a*b, 1]])
q = - sqrt( -(a^2*b^2+1)*(a^4*p^2 + a^2*b^2 +1) / (a^6*(a^2*p^2 + b^2)) )
gp = g(p)
gq = g(q)
d1 = (a^4*p^2 + a^2*b^2 + 1) / ( a^2 * sqrt( - (a^4*p^2 + a^2*b^2 + 1) \
     * (a^2*b^2 + 1) / ( (a^2*p^2 + b^2) * a^6 ) ) )
d2 = - (a^2*b^2 + 1) / ( (a^2*p^2 + b^2) * a^4 * \
     sqrt( - (a^4*p^2 + a^2*b^2 + 1) * (a^2*b^2 + 1) / \
     ( (a^2*p^2 + b^2) * a^6 ) ) )
(gp*gq*gp).factor() == matrix([[d1, 0], [0, d2]])
\end{verbatim} 

To verify that \eqref{eq:p-squared} is equivalent to \eqref{eq:d-h_pq}, it suffices to consider the square of \eqref{eq:d-h_pq}.
This is because the sign of $h_{p,q_{\pm}}(0,0)$ is determined by the sign of $q_\pm$, which can be chosen freely, even after $p$ is fixed.
There are still two cases to consider, corresponding to the two elements of $P_d$.

 1.  The following piece of code verifies that $p^2 = -\frac{2 a^{2} b^{2} + 1 + \sqrt{-4 {\left(a^{2} b^{2} + 1\right)} d^2 + 1}}{2 a^{4}}$ implies the square of \eqref{eq:d-h_pq}. It outputs \texttt{True}.
\begin{verbatim}
var("a, b, d, dsq, psq")
psq = - ( 2*a^2*b^2 + 1 + sqrt( -4*(a^2*b^2+1)*d^2 + 1 ) ) / (2*a^4)
dsq = (a^4*psq + a^2*b^2 + 1)^2 / ( (a^2)^2*( -(a^4*psq + a^2*b^2 + 1) \
      * (a^2*b^2 + 1) / ( (a^2*psq + b^2) * a^6 ) ) )
bool( dsq.full_simplify() == d^2 )
\end{verbatim}

 2. The following piece of code verifies that $p^2 = -\frac{2 a^{2} b^{2} + 1 - \sqrt{-4 {\left(a^{2} b^{2} + 1\right)} d^2 + 1}}{2 a^{4}}$ implies the square of \eqref{eq:d-h_pq}. It outputs \texttt{True}.
\begin{verbatim}
var("a, b, d, dsq, psq")
psq = - ( 2*a^2*b^2 + 1 - sqrt( -4*(a^2*b^2+1)*d^2 + 1 ) ) / (2*a^4)
dsq = (a^4*psq + a^2*b^2 + 1)^2 / ( (a^2)^2*( -(a^4*psq + a^2*b^2 + 1) \
      * (a^2*b^2 + 1) / ( (a^2*psq + b^2) * a^6 ) ) )
bool( dsq.full_simplify() == d^2 )
\end{verbatim} 

The following piece of code verifies that $i\cdot g_{ib/a} = t_{1/(ab)}$.
It outputs \texttt{True}.
\begin{verbatim}
var("a, b")
def g(c):
    return 1/(a^2*c) * matrix([[a^2*(a^2*c^2+b^2), a*b], [a*b, 1]])
I * g(I*b/a) == matrix([[0, 1], [1, 1/(a*b)]])
\end{verbatim}

\subsection{Code for Lemma~\ref{lem:binaries_f_non_tractable}, Case 2}
\label{a:non-tractable-case2}

For \eqref{eq:case2-gadget2}, there are two cases to verify, corresponding to the different choices of solution for the equation $a^2b^2+1-0$.

1.  The following piece of code verifies \eqref{eq:case2-gadget2} if $b=i/a$; it outputs \texttt{True}.
\begin{verbatim}
var("a, d")
b = I/a
def g(c):
    return 1/(a^2*c) * matrix([[a^2*(a^2*c^2+b^2), a*b], [a*b, 1]])
gd = g(-1/(a^2*d))
ga = g(1/a^2)
(ga*gd*ga).factor() == matrix([[d, 0], [0, 1/d]])
\end{verbatim}

2.   The following piece of code verifies \eqref{eq:case2-gadget2} if $b=-i/a$; it outputs \texttt{True}.
\begin{verbatim}
var("a, d")
b = -I/a
def g(c):
    return 1/(a^2*c) * matrix([[a^2*(a^2*c^2+b^2), a*b], [a*b, 1]])
gd = g(-1/(a^2*d))
ga = g(1/a^2)
(ga*gd*ga).factor() == matrix([[d, 0], [0, 1/d]])
\end{verbatim} 

The following piece of code verifies \eqref{eq:case2-gadget1a} with $b=i/a$; it outputs \texttt{True}.
\begin{verbatim}
var("a, p, q")
b = I/a
def g(c):
    return 1/(a^2*c) * matrix([[a^2*(a^2*c^2+b^2), a*b], [a*b, 1]])
gp = g(p)
gq = g(q)
ga = g(1/a^2)
(gp*gq*ga*gq*gp).expand() == matrix([[-2*a^4*p^2 + p^2/q^2 + 2, -I], [-I, 0]])
\end{verbatim}

The following piece of code verifies that $-i\cdot g_{1/a^2} = t_{-i}$ with $b=i/a$; it outputs \texttt{True}.
\begin{verbatim}
var("a")
b = I/a
def g(c):
    return 1/(a^2*c) * matrix([[a^2*(a^2*c^2+b^2), a*b], [a*b, 1]])
-I*g(1/a^2) == matrix([[0, 1], [1, -I]])
\end{verbatim}

The following piece of code verifies \eqref{eq:case2-gadget1b} with $b=-i/a$; it outputs \texttt{True}.
\begin{verbatim}
var("a, p, q")
b = -I/a
def g(c):
    return 1/(a^2*c) * matrix([[a^2*(a^2*c^2+b^2), a*b], [a*b, 1]])
gp = g(p)
gq = g(q)
ga = g(1/a^2)
(gp*gq*ga*gq*gp).expand() == matrix([[-2*a^4*p^2 + p^2/q^2 + 2, I], [I, 0]])
\end{verbatim}

The following piece of code verifies that $i\cdot g_{1/a^2} = t_{i}$ with $b=-i/a$; it outputs \texttt{True}.
\begin{verbatim}
var("a")
b = -I/a
def g(c):
    return 1/(a^2*c) * matrix([[a^2*(a^2*c^2+b^2), a*b], [a*b, 1]])
I*g(1/a^2) == matrix([[0, 1], [1, I]])
\end{verbatim}

\subsection{Code for Lemma~\ref{lem:binaries_f_non_tractable}, Case 3}
\label{a:non-tractable-case3}

To verify \eqref{eq:case3-gadget2}, there are four cases to consider, corresponding to the choice of solution of $2a^2b^2+1=0$ and the choice of sign in the definition of $w$.

1. The following piece of code verifies \eqref{eq:case3-gadget2} for the case where $b=i/(\sqrt{2}a)$ and $w$ has a plus sign; it outputs \texttt{True}.   
\begin{verbatim}
var("a, d")
b = I/(sqrt(2)*a)
def g(c):
    return 1/(a^2*c) * matrix([[a^2*(a^2*c^2+b^2), a*b], [a*b, 1]])
def simp(M):
    # this function simplifies every element of the matrix M
    return matrix([[x.full_simplify() for x in r] for r in M.rows()])
w = (sqrt(-2*d^2 + 1) - 1)/(2*a^2*d)
v = (1/a^2) * sqrt( (2*a^4*w^2 - 1) / ( 2*(2*a^4*w^2 + 1) ) )
gv = g(v)
gw = g(w)
simp(gv*gw*gv) == matrix([[d, 0], [0, 1/d]])
\end{verbatim}

2. The following piece of code verifies \eqref{eq:case3-gadget2} for the case where $b=i/(\sqrt{2}a)$ and $w$ has a minus sign; it outputs \texttt{True}.
\begin{verbatim}
var("a, d")
b = I/(sqrt(2)*a)
def g(c):
    return 1/(a^2*c) * matrix([[a^2*(a^2*c^2+b^2), a*b], [a*b, 1]])
def simp(M):
    return matrix([[x.full_simplify() for x in r] for r in M.rows()])
w = (-sqrt(-2*d^2 + 1) - 1)/(2*a^2*d)
v = (1/a^2) * sqrt( (2*a^4*w^2 - 1) / ( 2*(2*a^4*w^2 + 1) ) )
gv = g(v)
gw = g(w)
simp(gv*gw*gv) == matrix([[d, 0], [0, 1/d]])
\end{verbatim}

3. The following piece of code verifies \eqref{eq:case3-gadget2} for the case where $b=-i/(\sqrt{2}a)$ and $w$ has a plus sign; it outputs \texttt{True}.
\begin{verbatim}
var("a, d")
b = -I/(sqrt(2)*a)
def g(c):
    return 1/(a^2*c) * matrix([[a^2*(a^2*c^2+b^2), a*b], [a*b, 1]])
def simp(M):
    return matrix([[x.full_simplify() for x in r] for r in M.rows()])
w = (sqrt(-2*d^2 + 1) - 1)/(2*a^2*d)
v = (1/a^2) * sqrt( (2*a^4*w^2 - 1) / ( 2*(2*a^4*w^2 + 1) ) )
gv = g(v)
gw = g(w)
simp(gv*gw*gv) == matrix([[d, 0], [0, 1/d]])
\end{verbatim}

4.  The following piece of code verifies \eqref{eq:case3-gadget2} for the case where $b=-i/(\sqrt{2}a)$ and $w$ has a minus sign; it outputs \texttt{True}.
\begin{verbatim}
var("a, d")
b = -I/(sqrt(2)*a)
def g(c):
    return 1/(a^2*c) * matrix([[a^2*(a^2*c^2+b^2), a*b], [a*b, 1]])
def simp(M):
    return matrix([[x.full_simplify() for x in r] for r in M.rows()])
w = (-sqrt(-2*d^2 + 1) - 1)/(2*a^2*d)
v = (1/a^2) * sqrt( (2*a^4*w^2 - 1) / ( 2*(2*a^4*w^2 + 1) ) )
gv = g(v)
gw = g(w)
simp(gv*gw*gv) == matrix([[d, 0], [0, 1/d]])
\end{verbatim} 

The following piece of code verifies that $i\cdot h''_r=\NEQ$ if $b=i/(\sqrt{2}a)$; it outputs \texttt{True}.
\begin{verbatim}
var("a, r")
b = I/(sqrt(2)*a)
def g(c):
    return 1/(a^2*c) * matrix([[a^2*(a^2*c^2+b^2), a*b], [a*b, 1]])
s = sqrt( (2*a^4*r^2 + 1)*(2*a^4*r^2 - 1) / ( 2*a^4 * (4*a^8*r^4 + 1) ) )
u = (2*a^4*r^2-1) / ( sqrt(2)*a^2 * (2*a^4*r^2+1) )
gs = g(s)
gr = g(r)
gu = g(u)
I*(gs*gr*gu*gr*gs).factor() == matrix([[0, 1], [1, 0]])
\end{verbatim}

The following piece of code verifies that $-i\cdot g_{1/(\sqrt{2}a^2)} = t_{-i\sqrt{2}}$ if $b=i/(\sqrt{2}a)$; it outputs \texttt{True}.
\begin{verbatim}
var("a")
b = I/(sqrt(2)*a)
def g(c):
    return 1/(a^2*c) * matrix([[a^2*(a^2*c^2+b^2), a*b], [a*b, 1]])
-I*g(1/(sqrt(2)*a^2)) == matrix([[0, 1], [1, -I*sqrt(2)]])
\end{verbatim}

The following piece of code verifies that $-i\cdot h''_r=\NEQ$ if $b=-i/(\sqrt{2}a)$; it outputs \texttt{True}.
\begin{verbatim}
var("a, r")
b = -I/(sqrt(2)*a)
def g(c):
    return 1/(a^2*c) * matrix([[a^2*(a^2*c^2+b^2), a*b], [a*b, 1]])
s = sqrt( (2*a^4*r^2 + 1)*(2*a^4*r^2 - 1) / ( 2*a^4 * (4*a^8*r^4 + 1) ) )
u = (2*a^4*r^2-1) / ( sqrt(2)*a^2 * (2*a^4*r^2+1) )
gs = g(s)
gr = g(r)
gu = g(u)
-I*(gs*gr*gu*gr*gs).factor() == matrix([[0, 1], [1, 0]])
\end{verbatim}

The following piece of code verifies that $i\cdot g_{1/(\sqrt{2}a^2)} = t_{i\sqrt{2}}$ if $b=-i/(\sqrt{2}a)$; it outputs \texttt{True}.
\begin{verbatim}
var("a")
b = -I/(sqrt(2)*a)
def g(c):
    return 1/(a^2*c) * matrix([[a^2*(a^2*c^2+b^2), a*b], [a*b, 1]])
I*g(1/(sqrt(2)*a^2)) == matrix([[0, 1], [1, I*sqrt(2)]])
\end{verbatim}

\subsection{Code for Lemma~\ref{lem:binaries_f_tractable}}
\label{a:binaries_f_tractable}

The following code snippet verifies \eqref{eq:h_s}; it outputs \texttt{True}.
\begin{verbatim}
var("a, b, c, s")
t = -(b^2 + s)^2/(c*s + b)^2
g = matrix([[b, 1], [1, c]])
fs = matrix([[1, 0], [0, s]])
ft = matrix([[1, 0], [0, t]])
x = -(b^2 + s)*(b*c - 1)^2*s/(c*s + b)
y = -(b^2*c^2*s + 2*c^2*s^2 + 2*b*c*s + 2*b^2 + s)*(b*c - 1)^2*s/(c*s + b)^2
(g*fs*g*ft*g*fs*g).factor() == matrix([[0, x], [x, y]])
\end{verbatim}

The following code snippet verifies that $2c^3 \cdot h_{-1/(2c^2)} = \NEQ$ if $b=0\neq c$; it outputs \texttt{True}.
\begin{verbatim}
var("a, c")
b = 0
s = -1 / (2*c^2)
t = -(b^2 + s)^2/(c*s + b)^2
g = matrix([[b, 1], [1, c]])
fs = matrix([[1, 0], [0, s]])
ft = matrix([[1, 0], [0, t]])
2*c^3 * (g*fs*g*ft*g*fs*g).factor() == matrix([[0, 1], [1, 0]])
\end{verbatim}

The following code snippet verifies that $(-2b^3)^{-1}\cdot h_{-2b^2} = \NEQ$ if $b\neq 0 = c$; it outputs \texttt{True}.
\begin{verbatim}
var("a, b")
c = 0
s = -2*b^2
t = -(b^2 + s)^2/(c*s + b)^2
g = matrix([[b, 1], [1, c]])
fs = matrix([[1, 0], [0, s]])
ft = matrix([[1, 0], [0, t]])
(-2*b^3)^(-1) * (g*fs*g*ft*g*fs*g).factor() == matrix([[0, 1], [1, 0]])
\end{verbatim}

To verify \eqref{eq:s-pm}, there are two cases depending on the choice of sign.
 
 1. The following piece of code verifies \eqref{eq:s-pm} for the plus case; it outputs \texttt{True}.
\begin{verbatim}
var("a, b, c")
def simp(M):
    return matrix([[x.full_simplify() for x in r] for r in M.rows()])
s = -1/4*(b^2*c^2 + 2*b*c + sqrt(b^2*c^2 + 6*b*c + 1)*(b*c - 1) + 1)/c^2
t = -(b^2 + s)^2/(c*s + b)^2
X = (b*c + sqrt( b^2*c^2 + 6*b*c + 1 ) + 3) * c / \
    ( (b*c + sqrt( b^2*c^2 + 6*b*c + 1 ) - 1) * (b*c - 1)^2 * s )
g = matrix([[b, 1], [1, c]])
fs = matrix([[1, 0], [0, s]])
ft = matrix([[1, 0], [0, t]])
simp(X * (g*fs*g*ft*g*fs*g)) == matrix([[0, 1], [1, 0]])
\end{verbatim}

2.   The following piece of code verifies \eqref{eq:s-pm} for the minus case; it outputs \texttt{True}.
\begin{verbatim}
var("a, b, c")
def simp(M):
    return matrix([[x.full_simplify() for x in r] for r in M.rows()])
s = -1/4*(b^2*c^2 + 2*b*c - sqrt(b^2*c^2 + 6*b*c + 1)*(b*c - 1) + 1)/c^2
t = -(b^2 + s)^2/(c*s + b)^2
X = (b*c - sqrt( b^2*c^2 + 6*b*c + 1 ) + 3) * c / \
    ( (b*c - sqrt( b^2*c^2 + 6*b*c + 1 ) - 1) * (b*c - 1)^2 * s )
g = matrix([[b, 1], [1, c]])
fs = matrix([[1, 0], [0, s]])
ft = matrix([[1, 0], [0, t]])
simp(X * (g*fs*g*ft*g*fs*g)) == matrix([[0, 1], [1, 0]])
\end{verbatim}

\subsection{Code for Lemma~\ref{lem:ghz_from_w}}
\label{a:ghz_from_w}

As $g$ (and thus $g'$) is symmetric by construction, the following code snippet verifies \eqref{eq:symmetric_ternary}.
It outputs \texttt{True}.
The definition of \texttt{gp} (which is~$g'$) is taken from~\eqref{eq:code2}.
The matrix \text{m} is  $M^T M$ from the paper.
\begin{verbatim}
var("a, b, c, d")
import itertools
def ONE(x,y,z):
    if x+y+z == 1:
        return 1
    return 0
def m(x,y):
    return matrix([[a, b], [c, d]])[x,y]
def gp(x1,x2,x3):
    output = 0
    for a2, a3, b2, b3, c2, c3 in itertools.product([0,1], repeat=6):
        output = output + m(b3,c2) * m(c3,a2) * m(a3,b2) * ONE(x1,a2,a3) \
                                   * ONE(x2,b2,b3) * ONE(x3,c2,c3)
    return output
[gp(0,0,0), gp(0,0,1), gp(0,1,1), gp(1,1,1)] == \
[b^3 + c^3 + 3*a*b*d + 3*a*c*d, a*b^2 + a*b*c + a*c^2 + a^2*d, \
 a^2*b + a^2*c, a^3]
\end{verbatim}

The following code snippet verifies \eqref{eq:GHZ-condition}, it outputs \texttt{True}.
\begin{verbatim}
var("a, b, c, d")
import itertools
def ONE(x,y,z):
    if x+y+z == 1:
        return 1
    return 0
def m(x,y):
    return matrix([[a, b], [c, d]])[x,y]
def gp(x1,x2,x3):
    output = 0
    for a2, a3, b2, b3, c2, c3 in itertools.product([0,1], repeat=6):
        output = output + m(b3,c2) * m(c3,a2) * m(a3,b2) * ONE(x1,a2,a3) \
                                   * ONE(x2,b2,b3) * ONE(x3,c2,c3)
    return output
g0 = gp(0,0,0)
g1 = gp(0,0,1)
g2 = gp(0,1,1)
g3 = gp(1,1,1)
bool( ((g0*g3 - g1*g2)^2 - 4*(g1^2 - g0*g2)*(g2^2 - g1*g3)).factor() \
      == -4 * (b*c - a*d)^3 * a^6 )
\end{verbatim}

\bibliographystyle{plain}
\bibliography{\jobname}

\begin{thebibliography}{10}

\bibitem{backens_new_2017}
Miriam Backens.
\newblock A {New} {Holant} {Dichotomy} {Inspired} by {Quantum} {Computation}.
\newblock In Ioannis Chatzigiannakis, Piotr Indyk, Fabian Kuhn, and Anca
  Muscholl, editors, {\em 44th {International} {Colloquium} on {Automata},
  {Languages}, and {Programming} ({ICALP} 2017)}, volume~80 of {\em Leibniz
  {International} {Proceedings} in {Informatics} ({LIPIcs})}, pages
  16:1--16:14, Dagstuhl, Germany, 2017. Schloss Dagstuhl--Leibniz-Zentrum fuer
  Informatik.
\newblock Full version at http://arxiv.org/abs/1702.00767.

\bibitem{Backens}
Miriam Backens.
\newblock A complete dichotomy for complex-valued holant\textsuperscript{c}.
\newblock In {\em 45th International Colloquium on Automata, Languages, and
  Programming, {ICALP} 2018, July 9-13, 2018, Prague, Czech Republic}, pages
  12:1--12:14, 2018.

\bibitem{backens_boolean_2018}
Miriam Backens, Andrei Bulatov, Leslie~Ann Goldberg, Colin McQuillan, and
  Stanislav \v{Z}ivn\'{y}.
\newblock Boolean approximate counting {CSPs} with weak conservativity, and
  implications for ferromagnetic two-spin.
\newblock {\em J. Comput. Syst. Sci.}, December 2019.
\newblock
  \href{https://dx.doi.org/10.1016/j.jcss.2019.12.003}{DOI:10.1016/j.jcss.2019.12.003}.

\bibitem{bezakova_inapproximability_2017}
Ivona Bez{\'{a}}kov{\'{a}}, Andreas Galanis, Leslie~Ann Goldberg, and Daniel
  Stefankovic.
\newblock Inapproximability of the independent set polynomial in the complex
  plane.
\newblock In {\em Proceedings of the 50th Annual {ACM} {SIGACT} Symposium on
  Theory of Computing, {STOC} 2018, Los Angeles, CA, USA, June 25-29, 2018},
  pages 1234--1240, 2018.

\bibitem{bulatov_functional_2017}
Andrei Bulatov, Leslie~Ann Goldberg, Mark Jerrum, David Richerby, and Stanislav
  \v{Z}ivn\'{y}.
\newblock Functional clones and expressibility of partition functions.
\newblock {\em Theoretical Computer Science}, 687:11--39, July 2017.

\bibitem{Bul11}
Andrei~A. Bulatov.
\newblock Complexity of conservative constraint satisfaction problems.
\newblock {\em {ACM} Trans. Comput. Log.}, 12(4):24:1--24:66, 2011.

\bibitem{Bul13}
Andrei~A. Bulatov.
\newblock The complexity of the counting constraint satisfaction problem.
\newblock {\em J. {ACM}}, 60(5):34:1--34:41, 2013.

\bibitem{bulatov_expressibility_2013}
Andrei~A. Bulatov, Martin Dyer, Leslie~Ann Goldberg, Mark Jerrum, and Colin
  McQuillan.
\newblock The {Expressibility} of {Functions} on the {Boolean} {Domain}, with
  {Applications} to {Counting} {CSPs}.
\newblock {\em Journal of the ACM}, 60(5):32:1--32:36, October 2013.

\bibitem{BG05}
Andrei~A. Bulatov and Martin Grohe.
\newblock The complexity of partition functions.
\newblock {\em Theor. Comput. Sci.}, 348(2-3):148--186, 2005.

\bibitem{CaiChen10}
Jin{-}Yi Cai and Xi~Chen.
\newblock A decidable dichotomy theorem on directed graph homomorphisms with
  non-negative weights.
\newblock In {\em 51th Annual {IEEE} Symposium on Foundations of Computer
  Science, {FOCS} 2010, October 23-26, 2010, Las Vegas, Nevada, {USA}}, pages
  437--446, 2010.

\bibitem{cai_complexity_2017}
Jin-Yi Cai and Xi~Chen.
\newblock {\em Complexity {Dichotomies} for {Counting} {Problems}, {Volume} 1:
  {Boolean} {Domain}}.
\newblock Cambridge University Press, November 2017.

\bibitem{CaiChenComplex}
Jin{-}Yi Cai and Xi~Chen.
\newblock Complexity of counting {CSP} with complex weights.
\newblock {\em J. {ACM}}, 64(3):19:1--19:39, 2017.

\bibitem{Cplus13}
Jin{-}Yi Cai, Xi~Chen, and Pinyan Lu.
\newblock Graph homomorphisms with complex values: {A} dichotomy theorem.
\newblock {\em {SIAM} J. Comput.}, 42(3):924--1029, 2013.

\bibitem{cai_Holant_2012}
Jin-Yi Cai, Sangxia Huang, and Pinyan Lu.
\newblock From {Holant} to \#{CSP} and {Back}: {Dichotomy} for {Holant}$^c$
  {Problems}.
\newblock {\em Algorithmica}, 64(3):511--533, March 2012.

\bibitem{cai_dichotomy_2011}
Jin-Yi Cai, Pinyan Lu, and Mingji Xia.
\newblock Dichotomy for {Holant}* {Problems} of {Boolean} {Domain}.
\newblock In {\em Proceedings of the {Twenty}-{Second} {Annual} {ACM}-{SIAM}
  {Symposium} on {Discrete} {Algorithms}}, Proceedings, pages 1714--1728.
  Society for Industrial and Applied Mathematics, January 2011.

\bibitem{ApproxCSP}
Xi~Chen, Martin~E. Dyer, Leslie~Ann Goldberg, Mark Jerrum, Pinyan Lu, Colin
  McQuillan, and David Richerby.
\newblock The complexity of approximating conservative counting {CSPs}.
\newblock {\em J. Comput. Syst. Sci.}, 81(1):311--329, 2015.

\bibitem{CohenJeavons}
D.~Cohen and P.~Jeavons.
\newblock The complexity of constraint languages.
\newblock In F.~Rossi, P.~van Beek, and T.~Walsh, editors, {\em Handbook of
  Constraint Programming}, chapter~8. Elsevier, 2006.

\bibitem{dur_three_2000}
W.~D\"{u}r, G.~Vidal, and J.~I. Cirac.
\newblock Three qubits can be entangled in two inequivalent ways.
\newblock {\em Physical Review A}, 62(6):062314, November 2000.

\bibitem{reltrichotomy}
Martin~E. Dyer, Leslie~Ann Goldberg, and Mark Jerrum.
\newblock An approximation trichotomy for {Boolean} {\#}{CSP}.
\newblock {\em J. Comput. Syst. Sci.}, 76(3-4):267--277, 2010.

\bibitem{Dplus07}
Martin~E. Dyer, Leslie~Ann Goldberg, and Mike Paterson.
\newblock On counting homomorphisms to directed acyclic graphs.
\newblock {\em J. {ACM}}, 54(6):27, 2007.

\bibitem{DG00}
Martin~E. Dyer and Catherine~S. Greenhill.
\newblock The complexity of counting graph homomorphisms.
\newblock {\em Random Struct. Algorithms}, 17(3-4):260--289, 2000.

\bibitem{DR13}
Martin~E. Dyer and David Richerby.
\newblock An effective dichotomy for the counting constraint satisfaction
  problem.
\newblock {\em {SIAM} J. Comput.}, 42(3):1245--1274, 2013.

\bibitem{Gplus10}
Leslie~Ann Goldberg, Martin Grohe, Mark Jerrum, and Marc Thurley.
\newblock A complexity dichotomy for partition functions with mixed signs.
\newblock {\em {SIAM} J. Comput.}, 39(7):3336--3402, 2010.

\bibitem{FerroIsing}
Leslie~Ann Goldberg and Mark Jerrum.
\newblock {The Complexity of Ferromagnetic Ising with Local Fields}.
\newblock {\em Combinatorics, Probability {\&} Computing}, 16(1):43--61, 2007.

\bibitem{KZ13}
Vladimir Kolmogorov and Stanislav Zivny.
\newblock The complexity of conservative valued {CSPs}.
\newblock {\em J. {ACM}}, 60(2):10:1--10:38, 2013.

\bibitem{li_simple_2006}
Dafa Li, Xiangrong Li, Hongtao Huang, and Xinxin Li.
\newblock Simple criteria for the {SLOCC} classification.
\newblock {\em Physics Letters A}, 359(5):428--437, December 2006.

\bibitem{lin_complexity_2017}
Jiabao Lin and Hanpin Wang.
\newblock The {Complexity} of {Holant} {Problems} over {Boolean} {Domain} with
  {Non}-{Negative} {Weights}.
\newblock In Ioannis Chatzigiannakis, Piotr Indyk, Fabian Kuhn, and Anca
  Muscholl, editors, {\em 44th {International} {Colloquium} on {Automata},
  {Languages}, and {Programming} ({ICALP} 2017)}, volume~80 of {\em Leibniz
  {International} {Proceedings} in {Informatics} ({LIPIcs})}, pages
  29:1--29:14, Dagstuhl, Germany, 2017. Schloss Dagstuhl--Leibniz-Zentrum fuer
  Informatik.

\bibitem{nielsen_quantum_2010}
Michael~A. Nielsen and Isaac~L. Chuang.
\newblock {\em Quantum computation and quantum information}.
\newblock Cambridge University Press, Cambridge; New York, 2010.

\bibitem{sagemath}
{The Sage Developers}.
\newblock {\em {S}ageMath, the {S}age {M}athematics {S}oftware {S}ystem
  ({V}ersion 8.3)}, 2018.
\newblock \url{http://www.sagemath.org}. Code snippets can be run at
  \url{http://sagecell.sagemath.org/}.

\bibitem{TomoCSP}
Tomoyuki Yamakami.
\newblock Approximate counting for complex-weighted {Boolean} constraint
  satisfaction problems.
\newblock {\em Inf. Comput.}, 219:17--38, 2012.

\bibitem{yamakami_approximation_2012}
Tomoyuki Yamakami.
\newblock Approximation complexity of complex-weighted degree-two counting
  constraint satisfaction problems.
\newblock {\em Theoretical Computer Science}, 461:86--105, November 2012.

\end{thebibliography}

\end{document}